\documentclass{article}

\usepackage{arxiv}

\usepackage[utf8]{inputenc} 
\usepackage[T1]{fontenc}    
\usepackage{hyperref}       
\usepackage{url}            
\usepackage{booktabs}       
\usepackage{amsfonts}       
\usepackage{nicefrac}       
\usepackage{microtype}      
\usepackage{lipsum}
\usepackage{graphicx}
\usepackage{subfigure}
\usepackage{amsmath}
\usepackage{amssymb}
\usepackage{amsthm}
\usepackage{algorithm}
\usepackage[noend]{algpseudocode}
\usepackage[dvipsnames]{xcolor}
\usepackage{enumerate}
\usepackage{comment}
\usepackage{tikz}

\newcommand{\lnarrow}[1]{%
\begin{tikzpicture}[#1]%
\draw (0.5ex,0) -- (1ex,1ex);%
\draw (1ex,1ex) -- (0,1ex);%
\end{tikzpicture}%
}

\newcommand{\unarrow}[1]{%
\begin{tikzpicture}[#1]%
\draw (0,0) -- (1ex,0);%
\draw (0,0) -- (0.5ex,1ex);%
\end{tikzpicture}%
}

\newcommand{\uwide}[1]{%
\begin{tikzpicture}[#1]%
\draw (0,0) -- (0.75ex,0);%
\draw (0.75ex,0) -- (1.25ex,1ex);%
\end{tikzpicture}%
}

\newcommand{\lwide}[1]{%
\begin{tikzpicture}[#1]%
\draw (0.75ex,0.5ex) -- (0,0.5ex);%
\draw (0,0.5ex) -- (-0.5ex,-0.5ex);%
\end{tikzpicture}%
}

\newcommand{\dlnarrow}[1]{%
\begin{tikzpicture}[#1]%
\draw (0.5ex,0) -- (0.5ex,1ex);%
\draw (0.5ex,1ex) -- (-0.5ex,0);%
\end{tikzpicture}%
}

\newcommand{\dunarrow}[1]{%
\begin{tikzpicture}[#1]%
\draw (0,0) -- (0,1ex);%
\draw (0,0) -- (1ex,1ex);%
\end{tikzpicture}%
}

\newcommand{\duwide}[1]{%
\begin{tikzpicture}[#1]%
\draw (0,0) -- (0,0.75ex);%
\draw (0,0) -- (-0.75ex,-0.75ex);%
\end{tikzpicture}%
}

\newcommand{\dlwide}[1]{%
\begin{tikzpicture}[#1]%
\draw (0,0) -- (0.75ex,0.75ex);%
\draw (0,0) -- (0,-0.75ex);%
\end{tikzpicture}%
}

\newcommand{\dashd}[1]{%
\begin{tikzpicture}[#1]%
\draw (0,0) -- (1ex,0);%
\draw (0.5ex,0) -- (1ex,1ex);%
\end{tikzpicture}%
}

\newcommand{\ddash}[1]{%
\begin{tikzpicture}[#1]%
\draw (0,1ex) -- (1ex,1ex);%
\draw (0.5ex,1ex) -- (0,0);%
\end{tikzpicture}%
}

\newcommand{\diaghor}[1]{%
\begin{tikzpicture}[#1]%
\draw (0,0) -- (1ex,1ex);%
\draw (0.5ex,0.5ex) -- (0,0.5ex);%
\end{tikzpicture}%
}

\newcommand{\diagver}[1]{%
\begin{tikzpicture}[#1]%
\draw (0,0) -- (1ex,1ex);%
\draw (0.5ex,0.5ex) -- (0.5ex,1ex);%
\end{tikzpicture}%
}

\newcommand{\hordiag}[1]{%
\begin{tikzpicture}[#1]%
\draw (0,0) -- (1ex,1ex);%
\draw (0.5ex,0.5ex) -- (1ex,0.5ex);%
\end{tikzpicture}%
}

\newcommand{\verdiag}[1]{%
\begin{tikzpicture}[#1]%
\draw (0,0) -- (1ex,1ex);%
\draw (0.5ex,0.5ex) -- (0.5ex,0);%
\end{tikzpicture}%
}

\newcommand{\ubranch}[1]{%
\begin{tikzpicture}[#1]%
\draw (0,0) -- (0,1ex);%
\draw (0,0.5ex) -- (1ex,1ex);%
\end{tikzpicture}%
}

\newcommand{\lbranch}[1]{%
\begin{tikzpicture}[#1]%
\draw (0,0) -- (1ex,0.5ex);%
\draw (1ex,0) -- (1ex,1ex);%
\end{tikzpicture}%
}

\DeclareMathOperator*{\argmin}{argmin}

\newtheorem{thm}{Theorem}
\newtheorem{lem}{Lemma}

\newtheorem{rmk}{Remark}

\title{Edge Intersection Graphs of Paths on a Triangular Grid}

\author{
  Vitor T. F. de Luca\\
  Universidade do Estado do Rio de Janeiro\\
  Rio de Janeiro, Brazil\\
  \texttt{toccivitor8@gmail.com} \\
   \And
   María Pía Mazzoleni\\
  Universidad Nacional de La Plata\\
  La Plata, Argentina\\
  \texttt{pia@mate.unlp.edu.ar} \\
   \And
 Fabiano S. Oliveira \\
  Universidade do Estado do Rio de Janeiro\\
  Rio de Janeiro, Brazil\\
  \texttt{fabiano.oliveira@ime.uerj.br} \\
  \And
  Tanilson D. Santos\\
  Universidade Federal do Tocantins\\
  Palmas, Brazil\\
  \texttt{tanilson.dias@uft.edu.br}\\
  \AND
  Jayme L. Szwarcfiter \\
  Universidade Federal do Rio de Janeiro \hspace{0.5cm} Universidade do Estado do Rio de Janeiro \\
  Rio de Janeiro, Brazil \\
  \texttt{jayme@nce.ufrj.br} \\
}

\begin{document}
\maketitle

\begin{abstract}
We introduce a new class of intersection graphs, the edge intersection graphs of paths on a triangular grid, called EPG\textsubscript{t} graphs. We show similarities and differences from this new class to the well-known class of EPG graphs. A turn of a path at a grid point is called a \emph{bend}. An EPG\textsubscript{t} representation in which every path has at most $k$ bends is called a B$_k$-EPG\textsubscript{t} representation and the corresponding graphs are called B$_k$-EPG\textsubscript{t} graphs. We provide examples of B$_{2}$-EPG graphs that are B$_{1}$-EPG\textsubscript{t}. We characterize the representation of cliques with three vertices and chordless 4-cycles in B$_{1}$-EPG\textsubscript{t} representations. We also prove that B$_{1}$-EPG\textsubscript{t} graphs have Strong Helly number $3$. Furthermore, we prove that B$_{1}$-EPG\textsubscript{t} graphs are $7$-clique colorable.
\end{abstract}


\section{Introduction}
In 2009, Golumbic, Lipshteyn and Stern~\cite{golumbic2009edge} introduced the notion of edge intersection graphs of paths on a rectangular grid. This family of graphs, called EPG graphs, is a generalization of the edge intersection graphs of paths on a degree four tree~\cite{golumbic1985edge,golumbic1985edge1,golumbic2006representations,golumbic2008representing}. They shifted from trees, as underlying structures, which do not allow cycles, to a rectangular grid. We consider here an even more general structure, from which a family of paths is taken, the triangular grid. A triangular grid consists of a rectangular grid with an extra direction (see Figure~\ref{fig: triangular_grid}). We call this extra direction the \emph{diagonal}. In most applications, a triangular grid is usually displayed as depicted in Figure~\ref{fig: triangular_grid2}. However, in the context of this paper, it is more natural to consider it depicted as in Figure~\ref{fig: triangular_grid1}, since we are treating such a grid as a generalization of the rectangular one. Notice that both drawings are equivalent, in the sense that the grid of Figure~\ref{fig: triangular_grid2} is that of Figure~\ref{fig: triangular_grid1} rotated 30 degrees in clockwise direction. We call the edge intersection graphs of paths on a triangular grid as EPG\textsubscript{t} graphs.

One motivation for studying these graphs is the same from EPG graphs, coming originally from circuit layout problems~\cite{brady1990stretching,molitor1991survey}. Another motivation is a rather natural optimization one, which consists of deciding whether an EPG\textsubscript{t} graph admits a representation having paths bending at most $k$ times.

The triangular grid has been studied in the context of the \emph{channel assignment problem with separation (CAPS)}. In cellular networks, a large number of base stations are expected to cover communications over a region~\cite{zander2000trends}. Such a covering can be achieved by placing base stations according to a regular plane tessellation. The most important regular tessellation of the plane is triangular tessellation~\cite{bertossi2004channel}, and the corresponding topology of such a tessellation is the triangular grid, known as \emph{triangular lattice} on those applications. The reason for adopting this particular tessellation comes from the fact that base stations are uniformly distributed in the coverage region, and an individual base station generally has six directional transceivers~\cite{janssen1998distributed}. Thus, the base station’s coverage area can be idealized as a regular triangular tessellation. The channel assignment problem with separation (CAPS) deals with assigning frequencies to stations such that there is no interference between frequencies assigned to nearby stations while trying to minimize the span (the difference between highest and lowest frequencies) of assigned frequencies.

In this paper, we introduce this new class of EPG\textsubscript{t} graphs and provide a characterization of representations of cliques and 4-cycles on those grids, extending the analogous results for the EPG graphs.

\begin{figure}[htb]
    \centering
    \subfigure[][]{\includegraphics[scale=0.5]{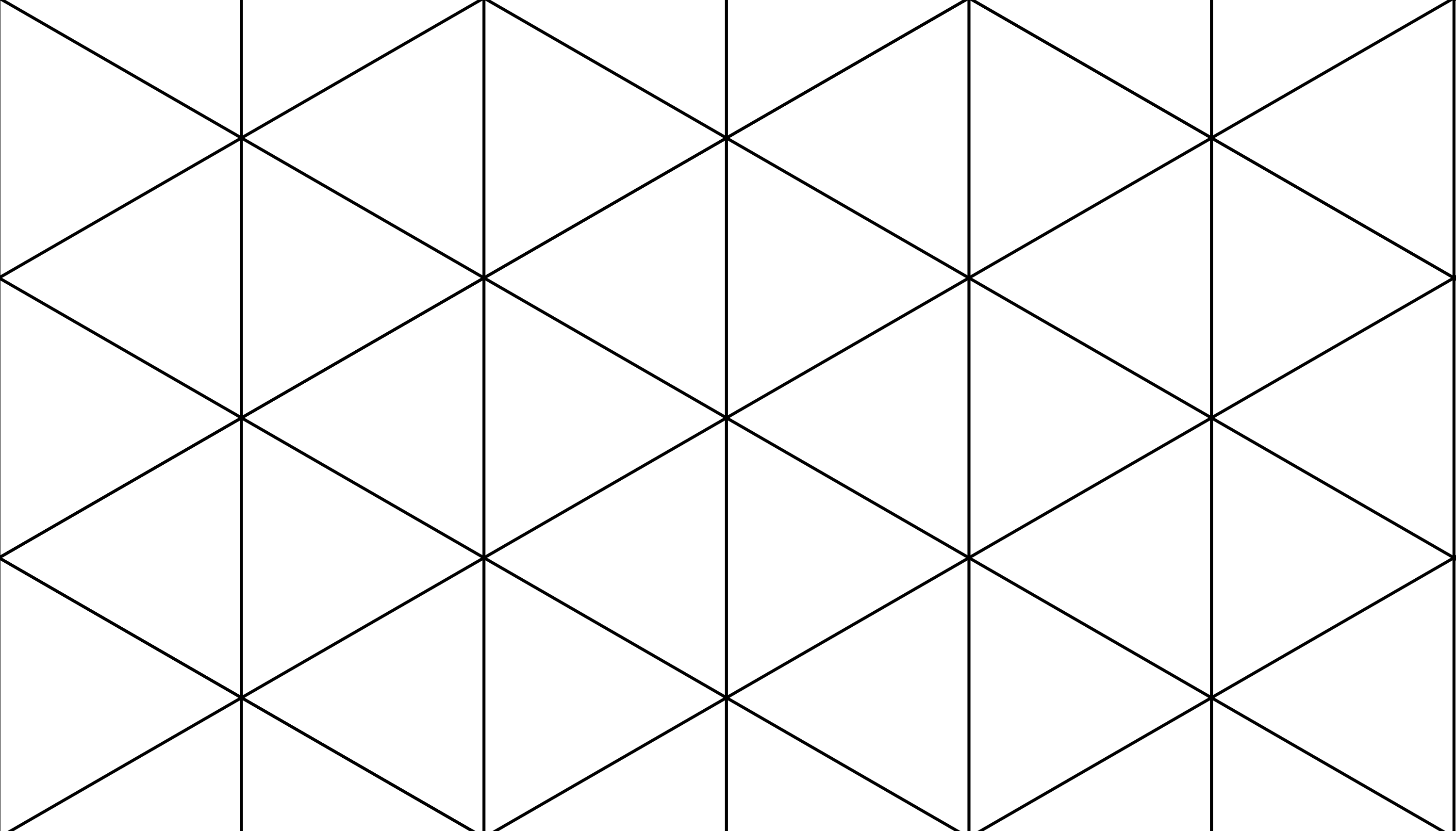} \label{fig: triangular_grid2}}
    \qquad
    \subfigure[][]{\includegraphics[scale=0.5]{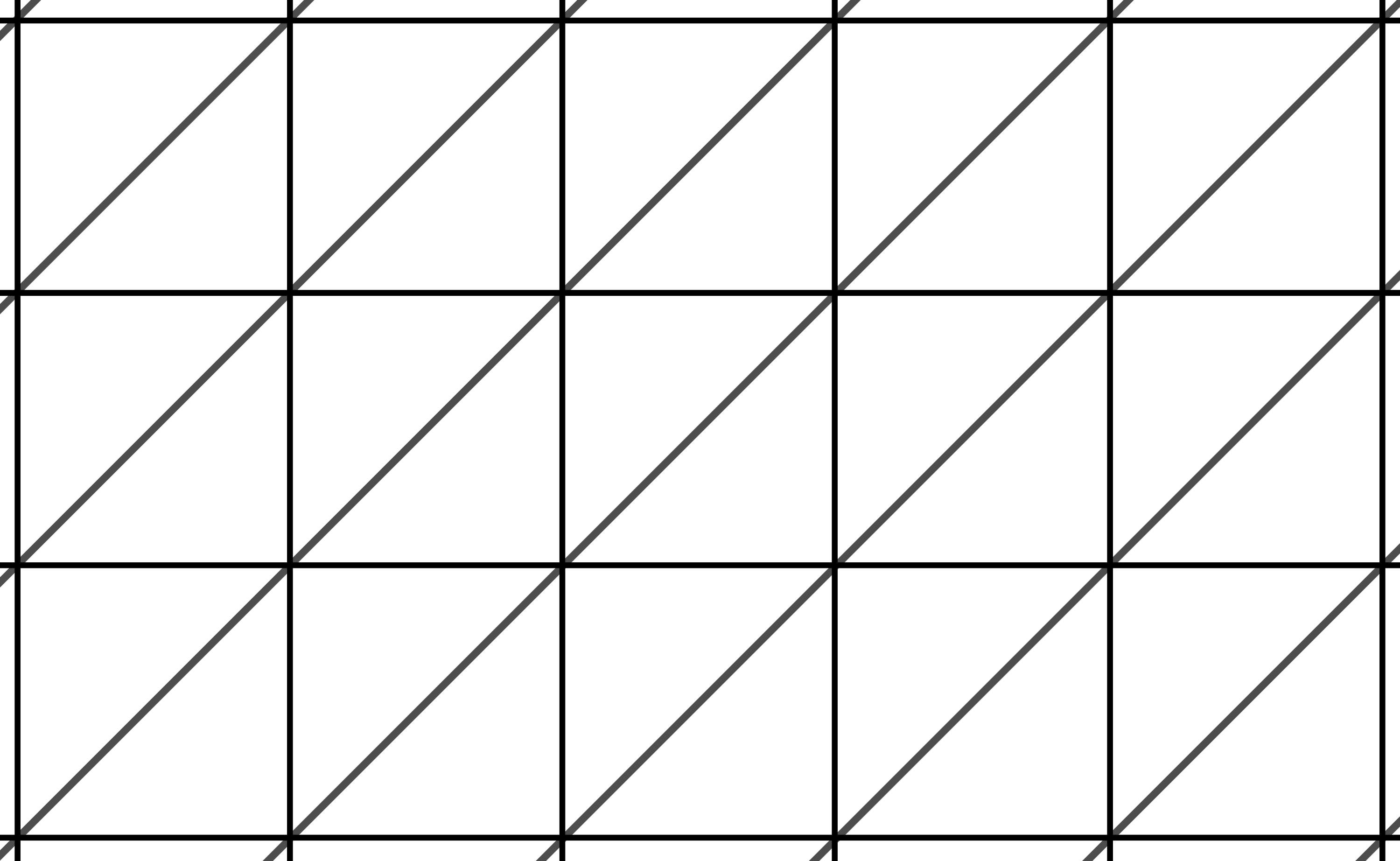} \label{fig: triangular_grid1}}
    \caption{The triangular grid.}
    \label{fig: triangular_grid}
\end{figure}

\section{Preliminaries}
We consider graphs that are undirected and without loops. Let $G$ be a graph. A set $C\subseteq V(G)$ is a \emph{clique} of $G$ if $xy\in E(G)$ for all $x,y\in C$. A set $S\subseteq V(G)$ is an \emph{independent set} of $G$ if $xy\notin E(G)$ for all $x,y\in S$. Let $v\in V(G)$. The \emph{neighborhood of $v$} is denoted by $N(v)=\{u\in V(G)\mid uv\in E(G)\}$ and the \emph{degree of $v$} is denoted by $d(v)=|N(v)|$. Let $H_1$ and $H_2$ be graphs. Denote by $H_1\cup H_2$ the graph $(V(H_1)\cup V(H_2), E(H_1)\cup E(H_2))$. Let $G$ be a graph and $x,y\in V(G)$. We denote by $P_{x,y}$ a path of $G$ connecting $x$ and $y$. The \emph{edge subdivision operation} for an edge $uv\in E(G)$ is the deletion of $uv$ from $E(G)$, the addition of a new vertex $w$ to $V(G)$ and the addition of the edges $uw$ and $vw$ to $E(G)$. A graph which has been derived from $G$ by a sequence of edge subdivision operations is called a \emph{subdivision of $G$}. Let $P$ and $Q$ be paths on a triangular grid $\mathcal{G}$. We denote by $P\cap_v Q$ when referring to the intersection between $P$ and $Q$ that includes both vertices and edges, and $P\cap_e Q$ when referring to the edge intersection between $P$ and $Q$.

Let $\mathcal{G}$ be a rectangular grid. Let $\mathcal{P}$ be a collection of nontrivial simple paths on $\mathcal{G}$. The edge intersection graph EPG$(\mathcal{P})$ of $\mathcal{P}$ to have vertices which correspond to the members of $\mathcal{P}$, so that two vertices are adjacent in EPG$(\mathcal{P})$ if and only if the corresponding paths in $\mathcal{P}$ share at least one edge in $\mathcal{G}$. A graph $G$ is called an \emph{edge intersection graph of paths on a rectangular grid} (EPG) if $G=\text{EPG$(\mathcal{P})$}$ for some $\mathcal{P}$ and $\mathcal{G}$, and $\langle \mathcal{P},\mathcal{G}\rangle$ is an EPG representation of $G$. Figure~\ref{fig: graph-example-epg-b} illustrates an EPG representation of the graph in Figure~\ref{fig: graph-example-epgt-a}. On those representations, for clearness of the drawings, the grid lines are omitted, and parts of a path that are parallel and mutually close are considered to belong to a same grid line. A turn of a path at a grid point is called a \emph{bend}. An EPG representation is \emph{B$_k$-EPG} if each path has at most $k$ bends. A graph that has a B$_k$-EPG representation is called \emph{B$_k$-EPG}. The \emph{rectangular bend-number} of a graph $G$ is the least $k$ such that $G$ is B$_k$-EPG.

We define the \emph{edge intersection graph of paths} EPG\textsubscript{t}$(\mathcal{P})$ of a collection of paths $\mathcal{P}$ to have vertices which correspond to the members of $\mathcal{P}$, such that two vertices are adjacent in EPG\textsubscript{t}$(\mathcal{P})$ if and only if the corresponding paths in $\mathcal{P}$ share at least one edge on a triangular grid $\mathcal{G}$. A graph $G$ is called an \emph{edge intersection graph of paths on a triangular grid} (EPG\textsubscript{t}) if $G=\text{EPG\textsubscript{t}$(\mathcal{P})$}$ for some $\mathcal{P}$ and $\mathcal{G}$, and $\langle \mathcal{P},\mathcal{G}\rangle$ is an EPG\textsubscript{t} representation of $G$. The graph in Figure~\ref{fig: graph-example-epgt-a} is B$_1$-EPG\textsubscript{t}, as the representation in Figure~\ref{fig: graph-example-epgt-c} shows. Similarly to the EPG graphs, a turn of a path at a grid point is called a \emph{bend}. A path is a \emph{B$_k$-path} if it contains at most $k$ bends. An EPG\textsubscript{t} representation is \emph{B$_k$-EPG\textsubscript{t}} if each path has at most $k$ bends. A graph that has a B$_k$-EPG\textsubscript{t} representation is called \emph{B$_k$-EPG\textsubscript{t}}. The \emph{triangular bend-number} of a graph $G$ is the least $k$ such that $G$ is B$_k$-EPG\textsubscript{t}.

A set of edges of a grid is \emph{co-linear} if all edges of the set belong to the same line of the grid, horizontal or vertical. The set of edges is called \emph{parallel} if all its edges lie on parallel lines of the grid, but no two of them are co-linear. A \emph{segment} of a path is a maximal subpath of the path with no bends. Therefore, a $0$-bend path has only one segment (the path itself), whereas a $1$-bend path has two segments. A $1$-bend path can be referred to as \emph{narrow} (Figure~\ref{fig: 45_angle}), \emph{normal} (Figure~\ref{fig: 90_angle}) or \emph{wide} (Figure~\ref{fig: 135_angle}), depending on the angle formed by its two segments. Note in Figure~\ref{fig: graph-example-epgt-c} that $P_a$ and $P_b$ are wide paths, $P_d$ and $P_e$ are normal paths and $P_c$ is a narrow path.

\begin{figure}[htb]
    \centering
    \subfigure[][]{\includegraphics[scale=0.7]{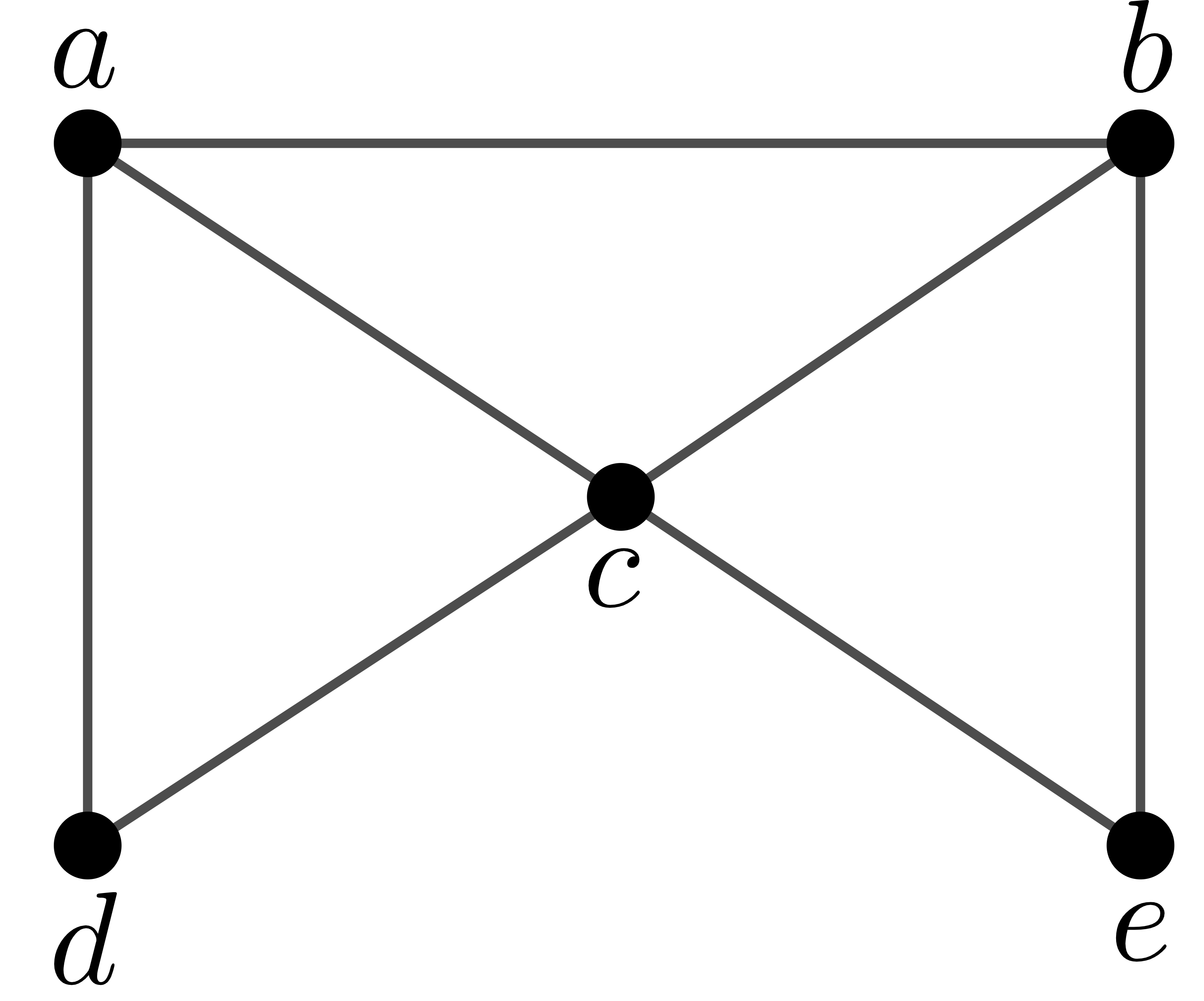} \label{fig: graph-example-epgt-a}}
    \qquad
    \subfigure[][]{\includegraphics[scale=0.7]{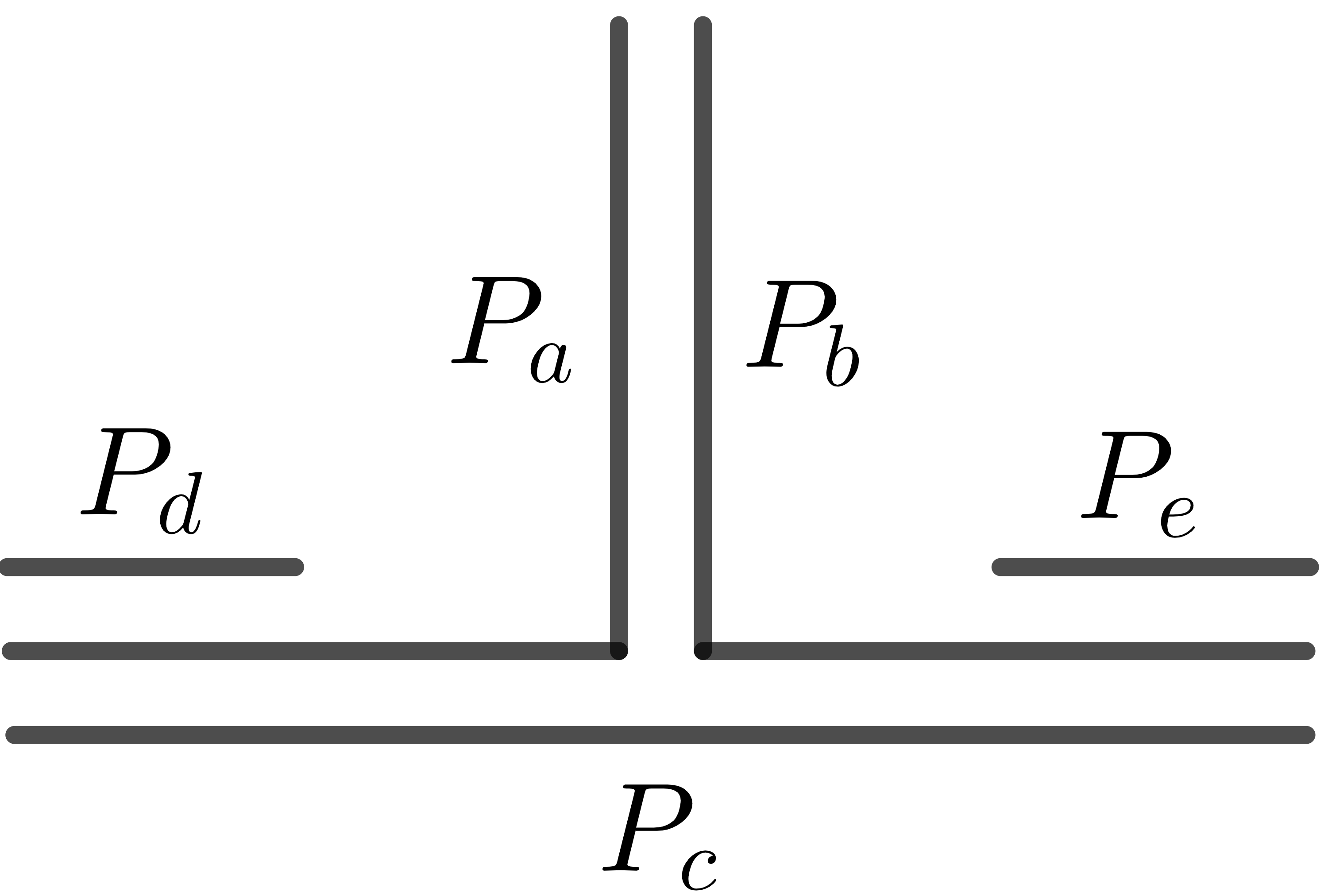} \label{fig: graph-example-epg-b}}
    \qquad
    \subfigure[][]{\includegraphics[scale=0.7]{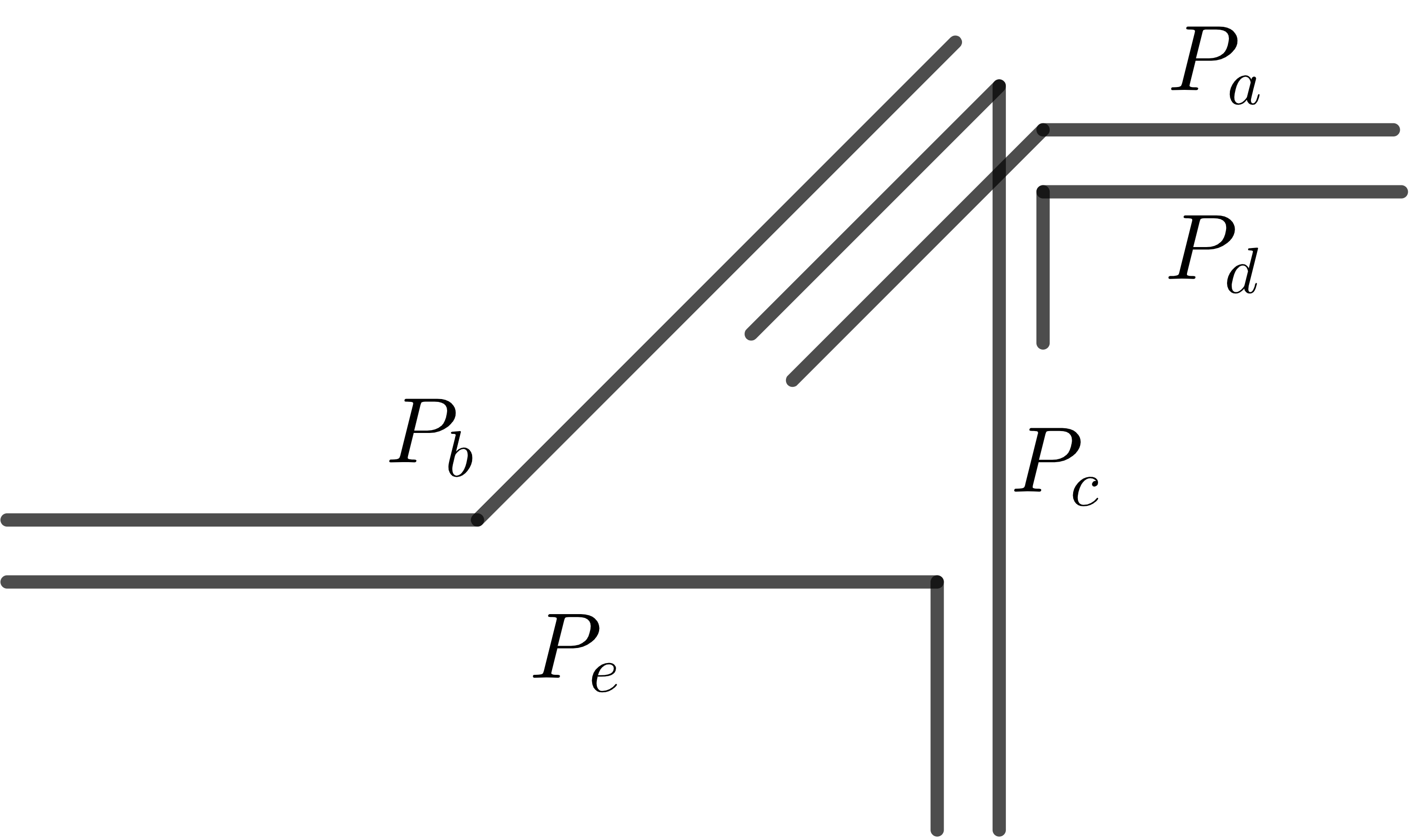} \label{fig: graph-example-epgt-c}}
    \caption{(a) A graph, (b) a B$_1$-EPG representation and (c) a B$_1$-EPG\textsubscript{t} representation.}
    \label{fig: graph-example-epgt}
\end{figure}

Let $G$ be a graph and $\langle \mathcal{P},\mathcal{G}\rangle$ a B$_1$-EPG\textsubscript{t} representation of $G$ on a triangular grid $\mathcal{G}$, where $\mathcal{P}=\{P_i\mid 1\leq i\leq |V(G)|\}$. We define $U(\mathcal{P})\subset\mathcal{G}$ as the \emph{underlying grid hosting the paths of $\mathcal{P}$}, such that 
\begin{eqnarray*}
U(\mathcal{P})&=&\{s\mid \text{$s$ is a segment of a path $P_i$ such that $P_i\cap_e P_j\neq\emptyset$}\\
& & \text{for some $1\leq i, j\leq|V(G)|$ with $i\neq j$}\}
\end{eqnarray*}
that is, $U(\mathcal{P})$ is the set of all segments of paths which intersect any other path. The collection of paths $\mathcal{P}=\{P_1, P_2, P_3, P_4, P_5\}$ in Figure~\ref{fig: underlying_graph_example_b} is a B$_1$-EPG\textsubscript{t} representation of the graph in Figure~\ref{fig: underlying_graph_example_a}, and $U(\mathcal{P})$ is formed by the vertices and edges of the segments marked in blue. Note that all segments of $P_2$, $P_3$ and $P_4$ edge intersect with some $P_i\in\mathcal{P}$, only one segment of $P_1$ edge intersect with some $P_i\in\mathcal{P}$ and none of the segments of $P_5$ edge intersect with any $P_i\in\mathcal{P}$. The graph of $U(\mathcal{P})$ is defined as the graph $G$ in which $V(G)$ is the set of grid points of $U(\mathcal{P})$ and $E(G)$ is the set of grid edges of $U(\mathcal{P})$. For convenience, we may refer to the subgrid $U(\mathcal{P})$ and the graph of $U(\mathcal{P})$ interchangeably, as long as no ambiguity arises.

\begin{figure}[htb]
    \centering
    \subfigure[][]{\includegraphics[scale=0.7]{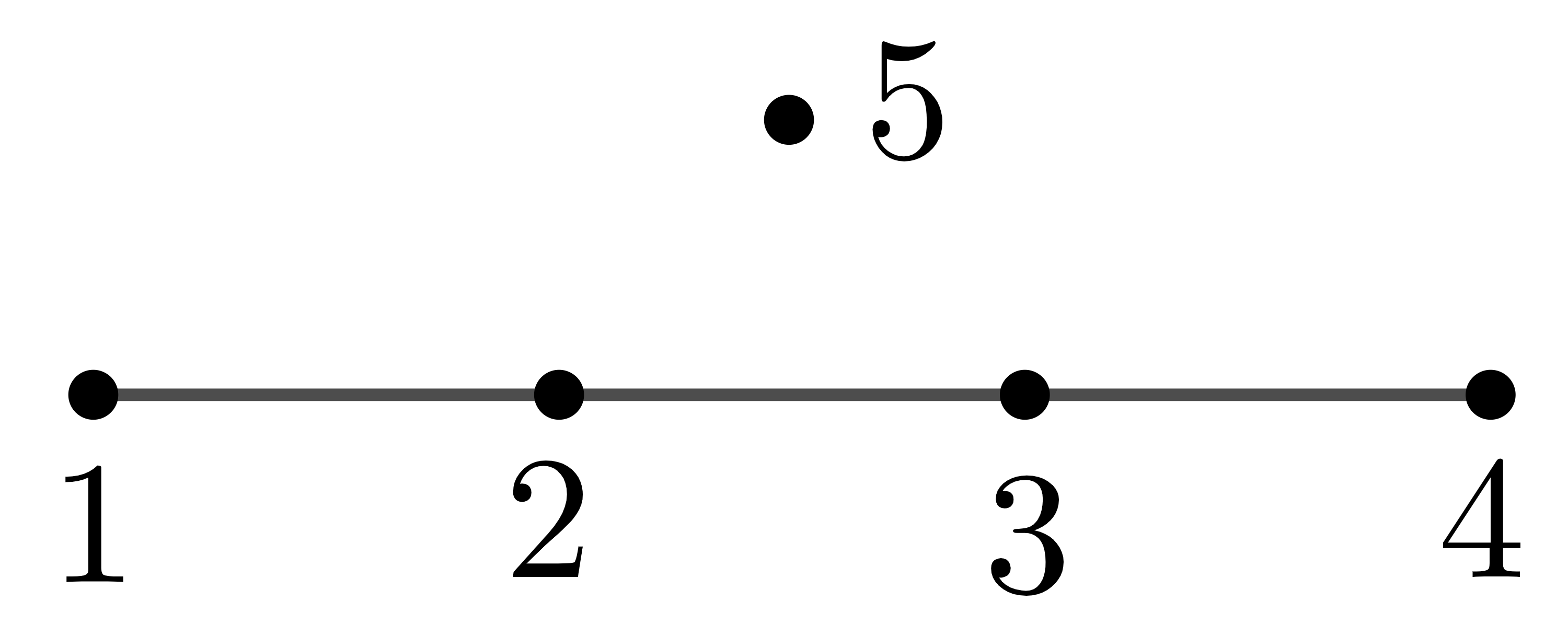} \label{fig: underlying_graph_example_a}}
    \qquad
    \subfigure[][]{\includegraphics[scale=0.7]{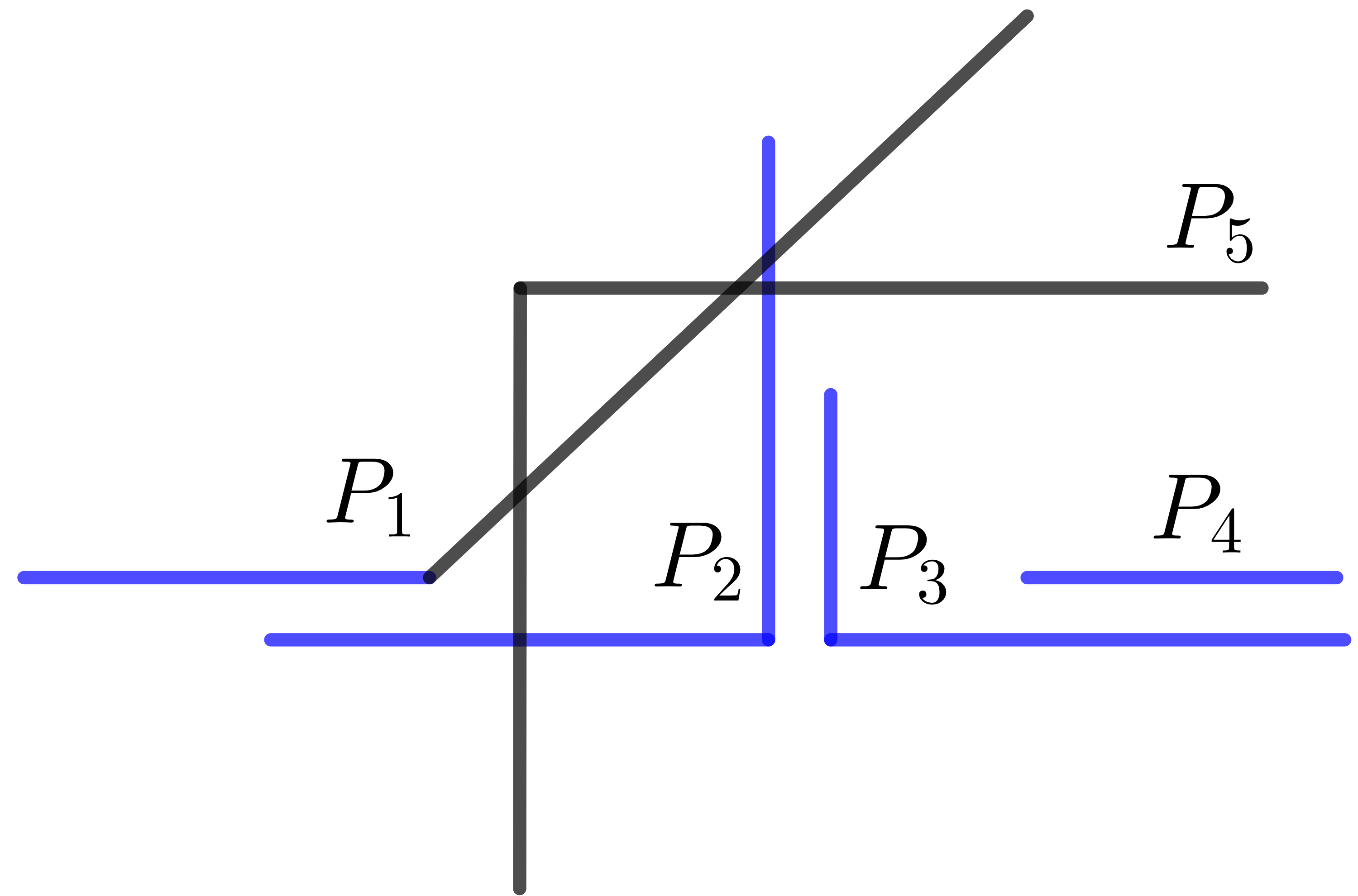} \label{fig: underlying_graph_example_b}}
    \caption{A graph and a collection of paths $\mathcal{P}=\{P_1, P_2, P_3, P_4, P_5\}$ forming a B$_1$-EPG\textsubscript{t} representation. The segments marked in blue are the ones forming $U(\mathcal{P})$.}
    \label{fig: underlying_graph_example}
\end{figure}

The \emph{$k$-cycle} graph $C_k$, $k\geq3$, has $k$ vertices, $v_1,\ldots,v_k$, and edges $v_iv_{i+1}$ for all $1\leq i\leq k$, where addition is assumed to be modulo $k$. The \emph{$k$-sun} graph $S_k$, $k\geq3$, has $2k$ vertices, consisting of a cycle $x_1y_1,y_1x_2,x_2y_2,y_2x_3,\ldots,x_ky_k,y_kx_1$, called the \emph{outer cycle}, such that $\{y_1,\ldots,y_k\}$ is a clique, called the \emph{inner clique}. A \emph{complete bipartite graph} $K_{m,n}$ is such that $V(K_{m,n})=V_1\cup V_2$, where $V_1$ and $V_2$ are independent sets, such that $|V_1|=m$, $|V_2|=n$ and for every two vertices $v_1\in V_1$ and $v_2\in V_2$, $v_1v_2\in E(K_{m,n})$. The complete bipartite graph $K_{1,k}$ is called a \emph{$k$-star}, $k\geq3$. The $3$-star is called a \emph{claw}.

\begin{figure}[htb]
    \center
    \subfigure[]{\includegraphics[scale=0.8]{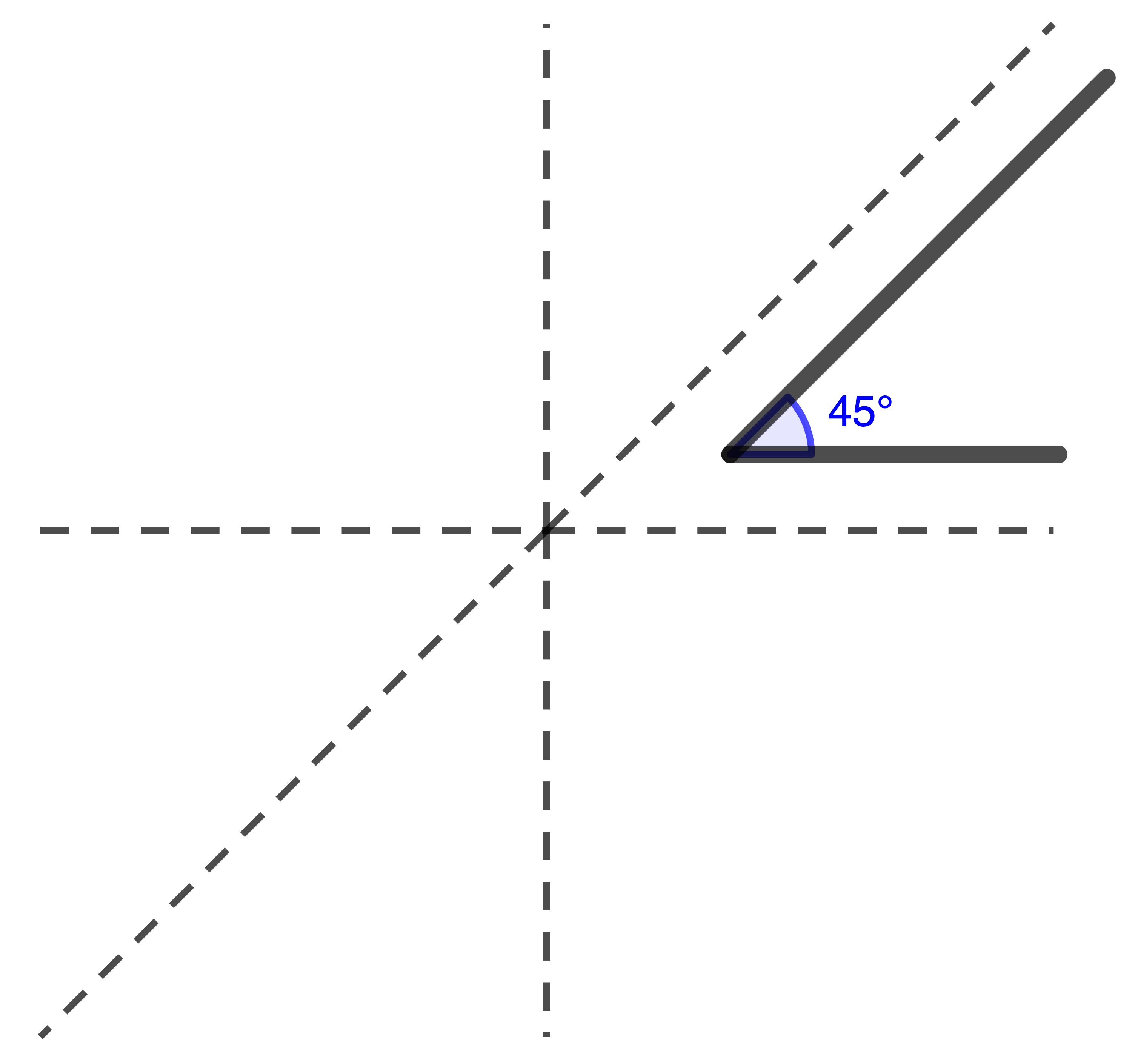} \label{fig: 45_angle}}
    \qquad
    \subfigure[]{\includegraphics[scale=0.8]{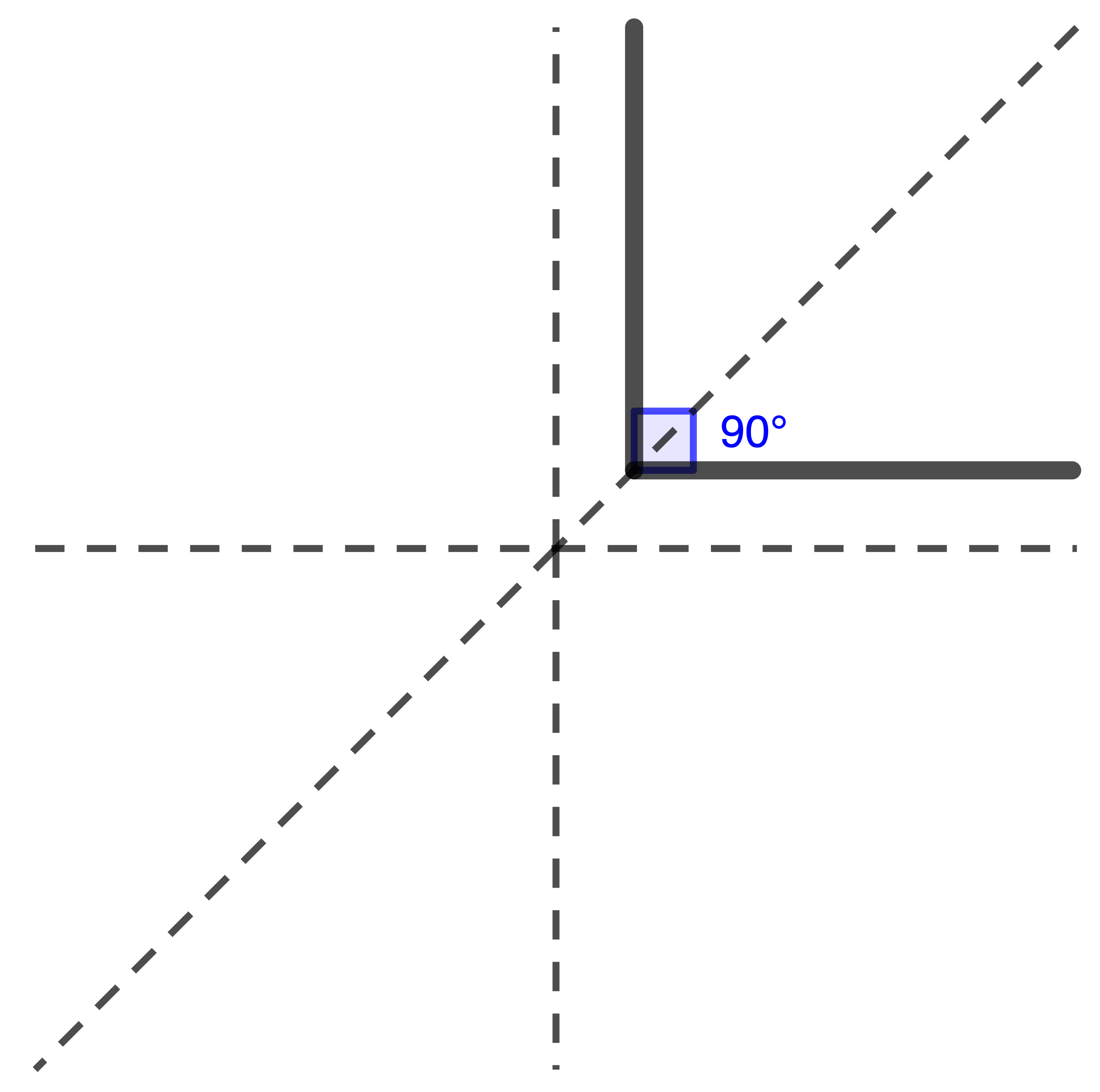} \label{fig: 90_angle}}
    \qquad
    \subfigure[]{\includegraphics[scale=0.8]{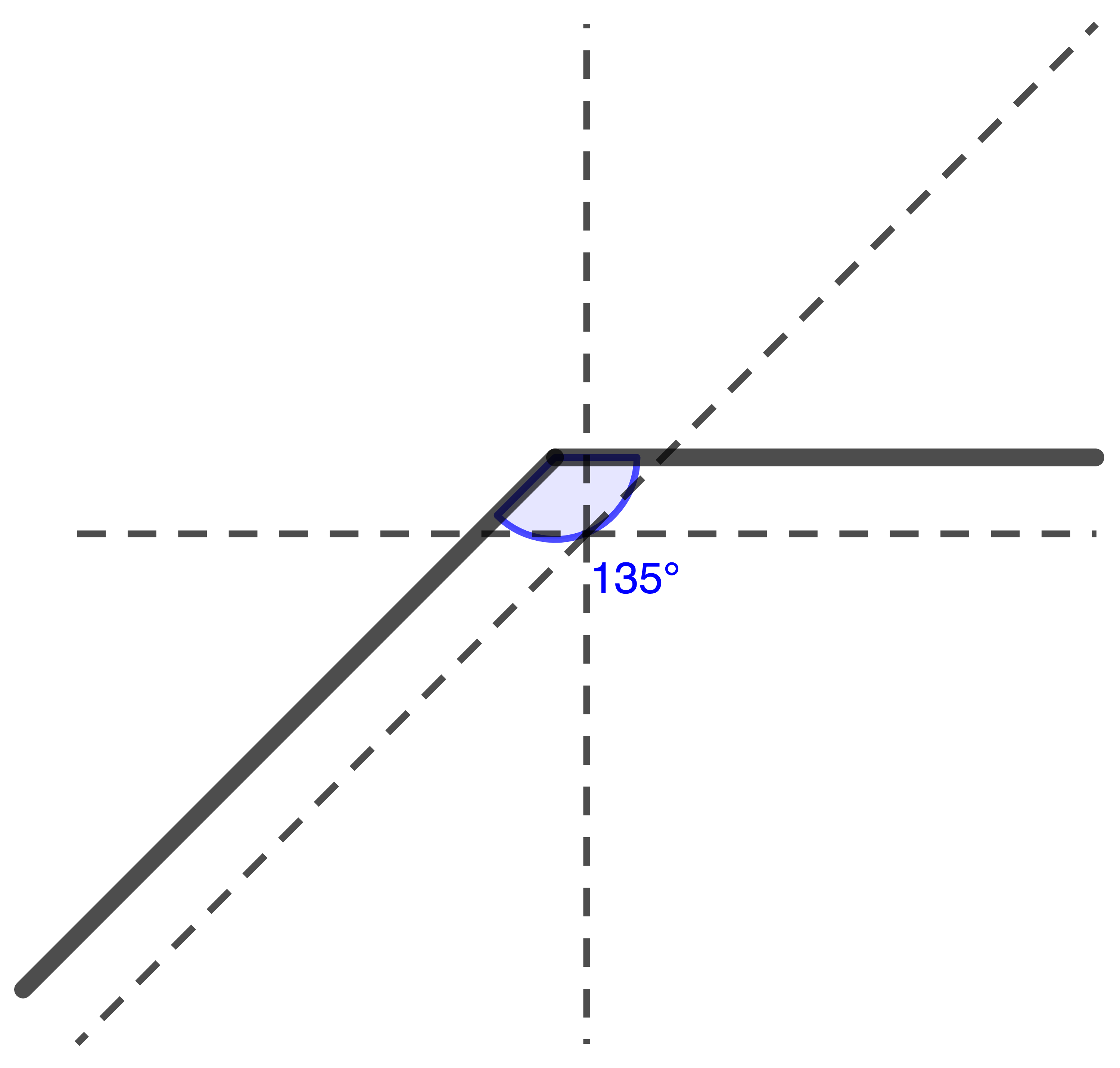} \label{fig: 135_angle}}
    \caption{The types of $1$-bend paths on the triangular grid: (a) narrow, (b) normal and (c) wide.}
    \label{fig: grid_angles}
\end{figure}
The following remarks are straight forward for a grid $\mathcal{G}$.
\begin{rmk} \label{1bend_int}
Two $1$-bend paths intersect in at most three segments. To see this, let $P_1=s_1\cup s_2$ and $P_2=s_3\cup s_4$ be two $1$-bend paths, such that $s_i$, $1\leq i\leq4$, are the segments of $\mathcal{G}$ forming $P_1$ and $P_2$. Since $P_1$ is a $1$-bend path, $s_1$ and $s_2$ must be in different directions. Since there are only three possible directions on the triangular grid, either $s_3$ or $s_4$ is parallel to $s_1$ or $s_2$.
\end{rmk}
\begin{rmk} \label{1bend_int_1}
Let $P_1=s_1\cup s_2$ and $P_2=s_3\cup s_4$ be $1$-bend paths, such that $s_i$, $1\leq i\leq4$, are the segments of $\mathcal{G}$ forming $P_1$ and $P_2$. If $P_1\cap_e P_2\neq\emptyset$, then either $P_1$ and $P_2$ have the same bend point and are of the same type, or $P_1\cap_e P_2\subseteq s_i$ for some $1\leq i\leq4$.
\end{rmk}
\begin{rmk} \label{1bend_int_2}
Let $P_1=s_1\cup s_2$ and $P_2=s_3\cup s_4$ be two $1$-bend paths, such that $s_i$, $1\leq i\leq4$, are the segments of $\mathcal{G}$ forming $P_1$ and $P_2$. Suppose $s_i$ and $s_j$ are on the same grid line $l$, for some $i\in\{1,2\}$ and $j\in\{3,4\}$, and $s_i\cap_v s_j=\{b\}$. If $P_1\cap_e P_2\neq\emptyset$, then $b$ is the bend point of $P_1$ and $P_2$, $b\in x$, for some $x\in P_1\cap_e P_2$ and $P_1\cap_e P_2\not\subset l$.
\end{rmk}
\begin{rmk} \label{1bend_int_3}
Let $l_1$ and $l_2$ be two grid lines and $P$ a path in $\mathcal{G}$.
\begin{itemize}
    \item If $l_1\cap_v l_2=\emptyset$, $P\cap_e l_1\neq\emptyset$ and $P\cap_e l_2\neq\emptyset$, then $P$ must have at least two bends. 
    \item If $l_1\cap_v l_2=\{b\}$, $P\cap_e l_1\neq\emptyset$, $P\cap_e l_2\neq\emptyset$ and $P$ has a single bend, then $b$ is the bend point of $P$ and one of the segments of $P$ is contained in $l_1$ and the other is contained in $l_2$.
\end{itemize}
\end{rmk}

\section{Helly and Strong Helly Numbers of B$_1$-EPG\textsubscript{\lowercase{t}} Graphs}
Let $\cal {F}$ be a family of subsets of some universal set $U$, and $h\geq 1$ be an integer.  We say that $\cal{F}$ is $h$-\emph{intersecting} if every subfamily of $h$ sets of $\cal {F}$ intersect. The \emph{core} of $\cal {F}$ is the intersection of all sets of $\cal {F}$, denoted by $core(\cal F)$. Note that $\cal {F}$ can be $h$-intersecting even if $core(\cal F) = \emptyset$.

The family $\cal{F}$ is $h$-\emph{Helly} if every $h$-intersecting subfamily $\cal{F'}$ of it satisfies $core(\cal{F'}) \neq \emptyset$~\cite{duchet1978propriete}. On the other hand, if for every subfamily $\cal{F'}$ of $\cal{F}$, there are $h$ subsets whose core equals the core of  $\cal {F'}$, then $\cal {F}$ is said to be \emph{strong} $h$-\emph{Helly}. Clearly, if $\cal {F}$ is $h$-Helly then it is $h'$-Helly, for all $h' \geq h$. Similarly, 
 if ${\cal F}$ is strong $h$-Helly then it is strong $h'$-Helly, for all $h' \geq h$.

The \emph{Helly number} of the family $\cal{F}$ is the least integer $h$, such that $\cal{F}$ is $h$-Helly. Similarly, the \emph{strong Helly number} of $\cal{F}$ is the least $h$, for which  $\cal{F}$ is strong $h$-Helly. It also follows that the strong Helly number of $\cal{F}$ is at least equal to its  Helly number.

A \emph{class} $\cal {C}$ of families $\cal {F}$  of subsets of some universal set $U$ is a subcollection  of the families  $\cal {F}$ of $U$. We say that $\cal C$ is a \emph{hereditary} class when it is closed under inclusion. The \emph{Helly number} of a class $\cal{C}$ of families $\cal{F}$ of subsets is the largest Helly number among the families $\cal {F}$. Similarly, the \emph{strong Helly number} of a class $\cal {C}$ is the largest strong Helly number of the families of $\cal {C}$.

If $\cal F$ is a family of subsets and $\cal C$ is a class of families, denote by $H(\cal F)$ and 
$H(\cal C)$,  the Helly numbers of $\cal F$ and $\cal C$, respectively, while  $sH({\cal F})$ and $sH({\cal C})$  represent the strong Helly numbers of $\cal F$ and $\cal C$.

 The Helly property known in the literature is when a family of subsets has Helly number $2$. It is well-known~\cite{golumbic1985edge1} that any collection of paths on a tree has Helly number $2$, any collection of intervals on a line has strong Helly number $2$, and any collection of B$_1$-paths on a rectangular grid has Helly number $4$~\cite{golumbic2013single}.

In \cite{santos2021helly}, the authors solved the problem of determining both the Helly and strong Helly numbers, for B$_k$-EPG, and B$_k$-VPG graphs, for each value $k$. Next, we use the same strategy to determine both the Helly and strong Helly numbers for B$_1$-EPG\textsubscript{t} graphs.

\subsection{The Helly Number of B$_1$-EPG\textsubscript{\lowercase{t}} Graphs}
The following theorem characterizes $h$-Helly families of subsets.

\begin{thm}[\cite{BD75}]\label{thm:BD}
A family $\cal{F}$  of subsets of the universal set $U$ is $h$-Helly if and only if for every subset $U' \subseteq U$, $|U'|= h+1$, the subfamily $\cal{F'}$ of $\cal{F}$, formed by the subsets containing at least $h$ of the $h+1$ elements of $U'$, has a non-empty core. 
\end{thm}

The next theorem is central to our results.

\begin{thm}[\cite{santos2021helly}]\label{thm:minimal}
Let ${\cal C}$ be a hereditary class of families ${\cal F}$ of subsets of the universal set $U$, whose Helly number $H({\cal C})$ equals $h$. Then there exists a family ${\cal F'} \in {\cal C}$ with exactly $h$ subsets, satisfying the following condition: 

For each subset $P_i \in \cal {F'}$, there is exactly one distinct element $u_i \in U$, such that \\
$$u_i \not \in P_i,$$ 
but $u_i$ is contained in all  subsets 
$$P_j \in {\cal F'} \setminus P_i.$$
\end{thm}

Let $\cal{ F'}$ be as in the previous theorem. It is simple to conclude that the removal of any subset from $\cal {F'}$ makes it an $(h-1)$-Helly family. Therefore we  call $\cal {F'}$ a \emph{minimal non}-$(h-1)$-\emph{Helly family}. Moreover, the element $u_i \not \in P_i$, contained in all subsets $P_j \in {\cal{F'}} \setminus P_i$, except $P_i$, is the \emph{$h$-non-representative} of $P_i$.

We can apply this notion of minimal families of subsets for the B$_1$-EPG\textsubscript{t} representations. Note that B$_1$-EPG\textsubscript{t} is a hereditary class.

\begin{lem} \label{lemma:3colin}
Let $\cal {F}$ be a minimal non-$(h-1)$-Helly family of paths on a triangular grid containing three co-linear non-representative edges. Then, $\cal{F}$ must contain paths with at least three bends.
\end{lem}
\begin{proof}
Let $u_i$ be the middle one of the three co-linear non-representative edges. It corresponds to the path $P_i$ of $\cal {F}$, not containing $u_i$.
Then $P_i$ must go through the other two non-representative edges, but it cannot include the middle edge. Therefore path $P_i$ must leave the common line of the grid, containing those three representatives edges, and return to that same line, thus requiring at least three bends.
\end{proof}

\begin{lem} \label{lemma:Lwit}
Let $\cal{F}$ be a minimal non-$(h-1)$-Helly family of paths on a grid with  Helly number $H({\cal F}) \geq 4$. If ${\cal F}$ contains three non-representative edges that lie on a common B$_1$-subpath $P_i$, then $\cal {F}$ must have some path with at least two bends.
\end{lem}
\begin{proof}
Since $\cal{F}$ is a minimal $(h-1)$-family having Helly number bigger or equal than $4$, it contains at least four paths. Without loss of generality, let  $u_1, u_2, u_3$ be the three non-representative edges contained in $P_4$ and such that $u_2$ lies between $u_1$ and $u_3$ in $P_4$. Then path $P_2$ must contain $u_1$ and $u_3$, but avoid $u_2$, thus requiring at least two bends.
\end{proof}

\begin{thm} \label{thm:helly_b1epgt}
H(B$_1$-EPG\textsubscript{t}) = 3.
\end{thm}
\begin{proof}
Let $\cal{F}$ be a family of three $1$-bend paths that pairwise intersect but which have no common edge, as depicted in Figure~\ref{fig:claw-example-epgt}. Then $\cal{F}$ is a $2$-intersecting B$_1$-EPG\textsubscript{t} family of three paths, having an empty core. 
Furthermore, removing any of the paths from $\cal{F}$ makes its core become non-empty. Therefore, $\cal{F}$ is a minimal non-$2$-Helly family and $H(\text{B$_1$-EPG\textsubscript{t}})\geq3$.

Assume by contradiction that the Helly number of B$_1$-paths is $h > 3$. In this case, consider a minimal non-$(h-1)$-Helly family of $\cal F$ of B$_1$-paths. Then $\cal F$ contains at least $h$ paths. Any path $P_1 \in \cal{F}$ must contain $h-1$ non-representative edges corresponding to the $h-1$ distinct paths of $\cal F$ other than $P_1$. Since $h-1 \geq 3$, $P_1$ contains at least three distinct non-representative edges $u_2, u_3, u_4 \in P_i$, with $u_3$ lying  between $u_2$ and $u_4$ in the path.

If $u_2$, $u_3$ and $u_4$ are co-linear, then by Lemma~\ref{lemma:3colin}, $P_3 \in \cal{F}$ must contain at least three bends. Otherwise, the edges must lie on $P_1$ which has a single bend. Thus, it follows from Lemma~\ref{lemma:Lwit} that $P_3$ has two bends. In any situation, a contradiction arises, implying that $H({\cal F}) \leq 3$.

This concludes the proof of the theorem.
\end{proof}

\subsection{The Strong Helly Number of B$_1$-EPG\textsubscript{\lowercase{t}} Graphs}
In this section, we determine the strong Helly number of B$_1$-EPG\textsubscript{t} graphs. In~\cite{santos2021helly}, the authors obtained the following result.

\begin{thm}[\cite{santos2021helly}]\label{thm:minimal-strong}

Let ${\cal C}$ be a hereditary class of families $\cal F$ of subsets of the universal set $U$, whose strong Helly number $sH({\cal C})$ equals $h$. Then there exists a family ${\cal F'} \in {\cal C}$ with exactly $h$ subsets satisfying the following condition: 

For each subset $P_i \in \cal {F'}$, there is exactly one distinct element $u_i \in U$, such that \\
$$u_i \not \in P_i,$$ 
but $u_i$ is contained in all  subsets 
$$P_j \in {\cal F'} \setminus P_i.$$
\end{thm}

Again, if we consider the family $\cal F'$ described in the theorem above it is simple to conclude that the removal of any subset from $\cal {F'}$ turns it $(h-1)$-strong Helly. Then call $\cal {F'}$ a \emph{minimal} non-$(h-1)$-strong Helly family. Moreover, the element $u_i \not \in P_i$, contained in all subsets $P_j \in {\cal{F'}} \setminus P_i$, except $P_i$, is the \emph{$h$-non-representative} of $P_i$.

As before, we employ the above minimal families of subsets, applied to paths on a triangular grid.
\begin{thm}
sH(B$_1$-EPG\textsubscript{t}) = 3.
\end{thm}
\begin{proof}
Recall that $sH(\text{B$_1$-EPG\textsubscript{t}})\geq H(\text{B$_1$-EPG\textsubscript{t}})$. Thus, by Theorem~\ref{thm:helly_b1epgt}, $sH(\text{B$_1$-EPG\textsubscript{t}})\geq3$.

Assume by contradiction that the strong Helly number of B$_1$-paths is $h > 3$. In this case, consider a minimal non-$(h-1)$-strong Helly family of $\cal F$ of B$_1$-paths. Then $\cal F$ contains at least $h$ paths. Any path $P_1 \in \cal{F}$ must contain $h-1$ non-representative edges corresponding to the $h-1$ distinct paths of $\cal F$ other than $P_1$. Since $h-1 \geq 3$, $P_1$ contains at least three distinct non-representative edges $u_2, u_3, u_4 \in P_i$, with $u_3$ lying  between $u_2$ and $u_4$ in the path.

If $u_2$, $u_3$ and $u_4$ are co-linear, then by Lemma~\ref{lemma:3colin}, $P_3 \in \cal{F}$ must contain at least three bends. Otherwise, the edges must lie on $P_1$ which has a single bend. Thus, it follows from Lemma~\ref{lemma:Lwit} that $P_3$ has two bends. In any situation, a contradiction arises, implying that $sH({\cal F}) \leq 3$.

This concludes the proof of the theorem.
\end{proof}

\section{Cliques on B$_1$-EPG\textsubscript{\lowercase{t}} Representations}
In this section, we characterize the B$_1$-EPG\textsubscript{t} representations of cliques with three vertices.

Let $T$ be a right triangle on the triangular grid $\mathcal{G}$. Let $u$ be a corner of $T$ and $(u, x_1)$ and $(u, x_2)$ the grid edges incident to $u$, such that $(u, x_1) \in E(T)$ and $(u, x_2) \in E(T)$. If $Q$ is a path that bends at $u$, then if it contains both $(u, x_1)$ and $(u, x_2)$, we say that $Q$ is an \emph{inside path relative to $u$}. If $Q$ contains either $(u, x_1)$ or $(u, x_2)$, we say that $Q$ is a \emph{midway path relative to $u$}. Otherwise, $Q$ is an \emph{outside path relative to $u$}.

Let $\langle \mathcal{P},\mathcal{G}\rangle$ be a B$_1$-EPG\textsubscript{t} representation of a graph $G$ on a triangular grid $\mathcal{G}$. Let $C$ be a maximal clique of $G$ and $\mathcal{P}_C\subseteq\mathcal{P}$ the set of paths representing the vertices of $C$. If $\bigcap_e P \neq \emptyset$, for all $P \in \mathcal{P}_C$, then $C$ is called an \emph{edge-clique}. If $\bigcap_e P = \emptyset$ and $\bigcap_v P = \{b\}$, for all $P \in \mathcal{P}_C$, then $C$ is called a \emph{claw-clique}. If $U(\mathcal{P}_C)$ has a right triangle $T$ as a subgraph, then $C$ is called a \emph{triangular-clique}. Let $T \subset U(\mathcal{P}_C)$, such that every corner, $x$, $y$ and $z$, of $T$ is the bend point of at least one path in $\mathcal{P}_C$. Assume every path of $\mathcal{P}_C$ bends at $x$, $y$ or $z$, and let $\mathcal{P}_C^x$, $\mathcal{P}_C^y$ and $\mathcal{P}_C^z$ be the paths of $\mathcal{P}_C$ that bend at $x$, $y$ and $z$, respectively. Consider the following cases:
\begin{itemize}
    \item If every path of $\mathcal{P}_C$ is an inside path relative to $x$, $y$ or $z$, then $C$ is called a \emph{flag-clique}.
    \item If at least one path of $\mathcal{P}_C^x$ is a midway path relative to $x$, and every path of $\mathcal{P}_C^y$ and every path of $\mathcal{P}_C^z$ are inside paths relative to $y$ and $z$, then $C$ is called a \emph{paw-clique}.
    \item If every path of $\mathcal{P}_C^x$ and every path of $\mathcal{P}_C^y$ are inside paths relative to $x$ and $y$, and at least one path of $\mathcal{P}_C^z$ is an outside path relative to $z$, then $C$ is called a \emph{cricket-clique}.
    \item If at least one path of $\mathcal{P}_C^x$ and at least one path of $\mathcal{P}_C^y$ are midway paths relative to $x$ and $y$, and every path of $\mathcal{P}_C^z$ is an inside path relative to $z$, then $C$ is called a \emph{bull-clique}.
    \item If every path of $\mathcal{P}_C^x$ is an inside path relative to $x$, at least one path of $\mathcal{P}_C^y$ is a midway path relative to $y$, and at least one path of $\mathcal{P}_C^z$ is an outside path relative to $z$, then $C$ is called an \emph{extended-bull-clique}.
    \item If at least one path of $\mathcal{P}_C^x$ is a midway path relative to $x$, at least one path of $\mathcal{P}_C^y$ is a midway path relative to $y$ and at least one path of $\mathcal{P}_C^z$ is a midway path relative to $z$, then $C$ is called a \emph{net-clique}.
\end{itemize}

These subtypes of the triangular-clique get their names after the respective layout of their paths on the grid (see Table~\ref{Tab:table-cliques}). Note that the existence of a third direction on the grid allows the arising of a new type of clique. See in Figure~\ref{fig: triangular_clique_examples} some examples of a triangular-clique.

\begin{figure}[htb]
\center
\subfigure[][]{\includegraphics[scale=0.6]{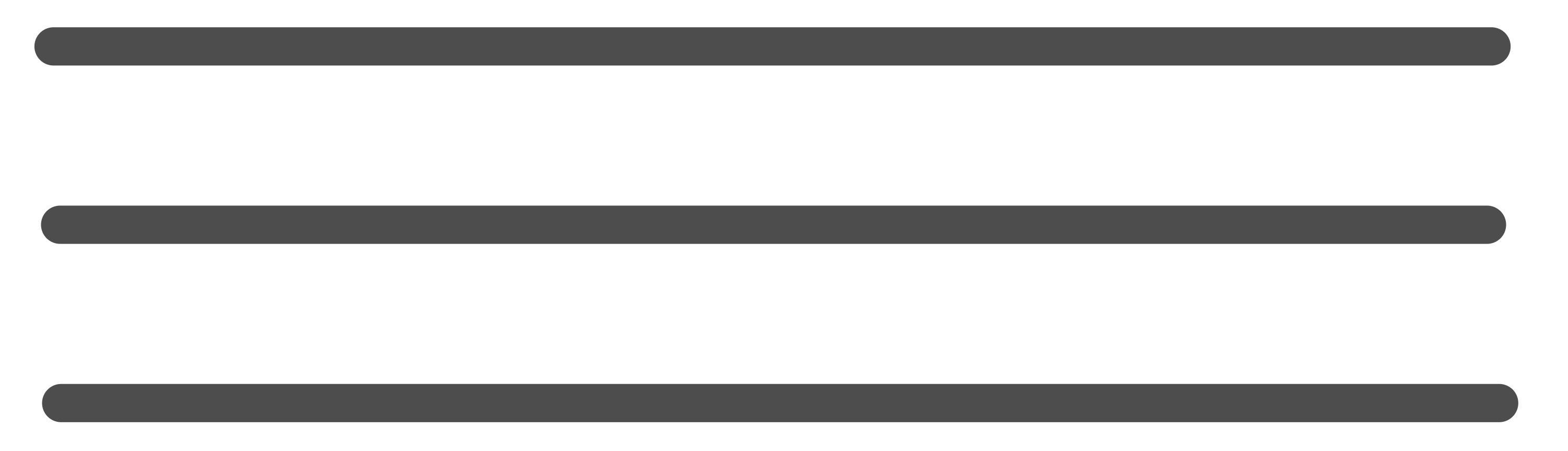}}
\qquad
\subfigure[][]{\includegraphics[scale=0.6]{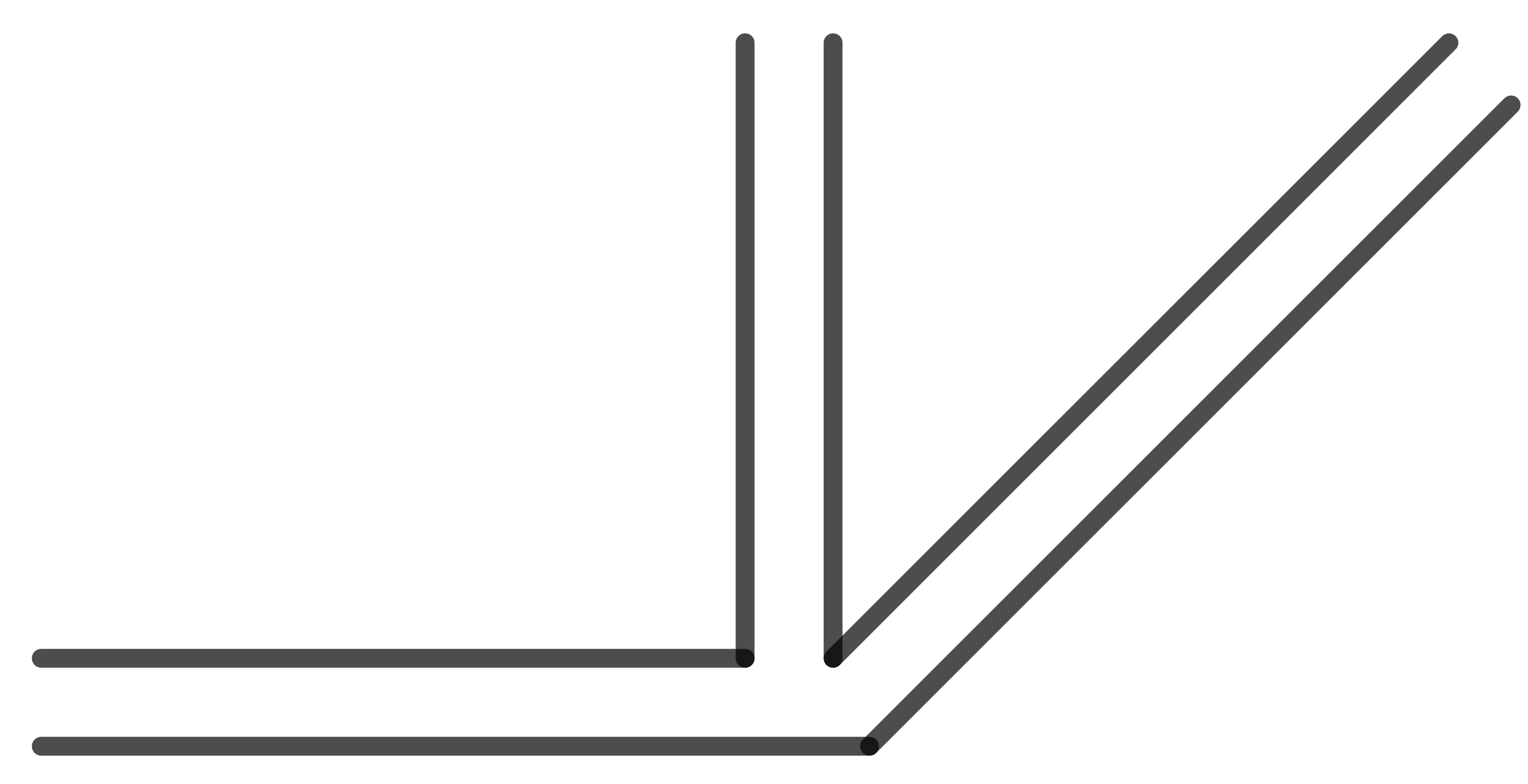} \label{fig:claw-example-epgt}}
\qquad
\subfigure[][]{\includegraphics[scale=0.6]{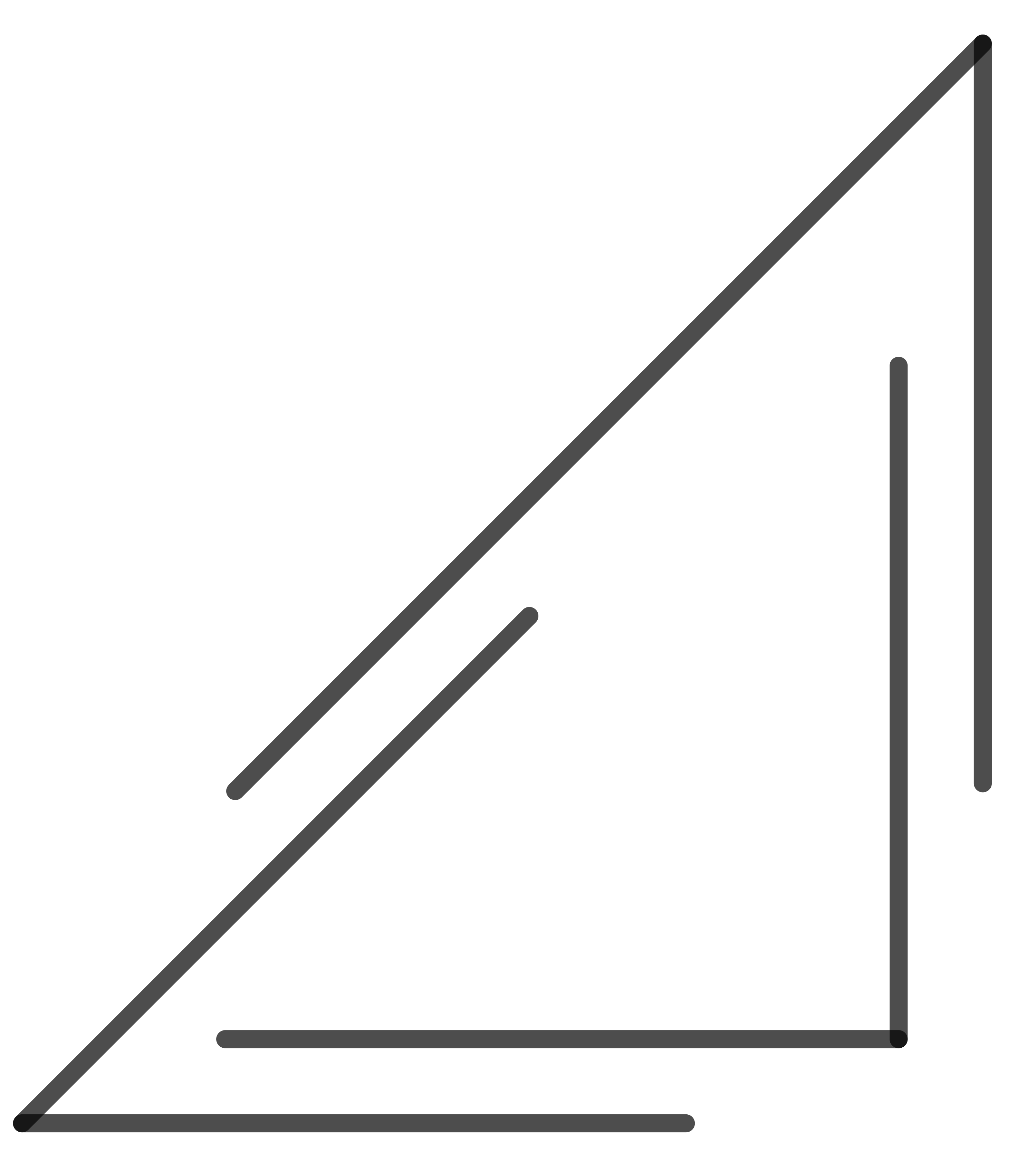}}
\qquad
\subfigure[][]{\includegraphics[scale=0.6]{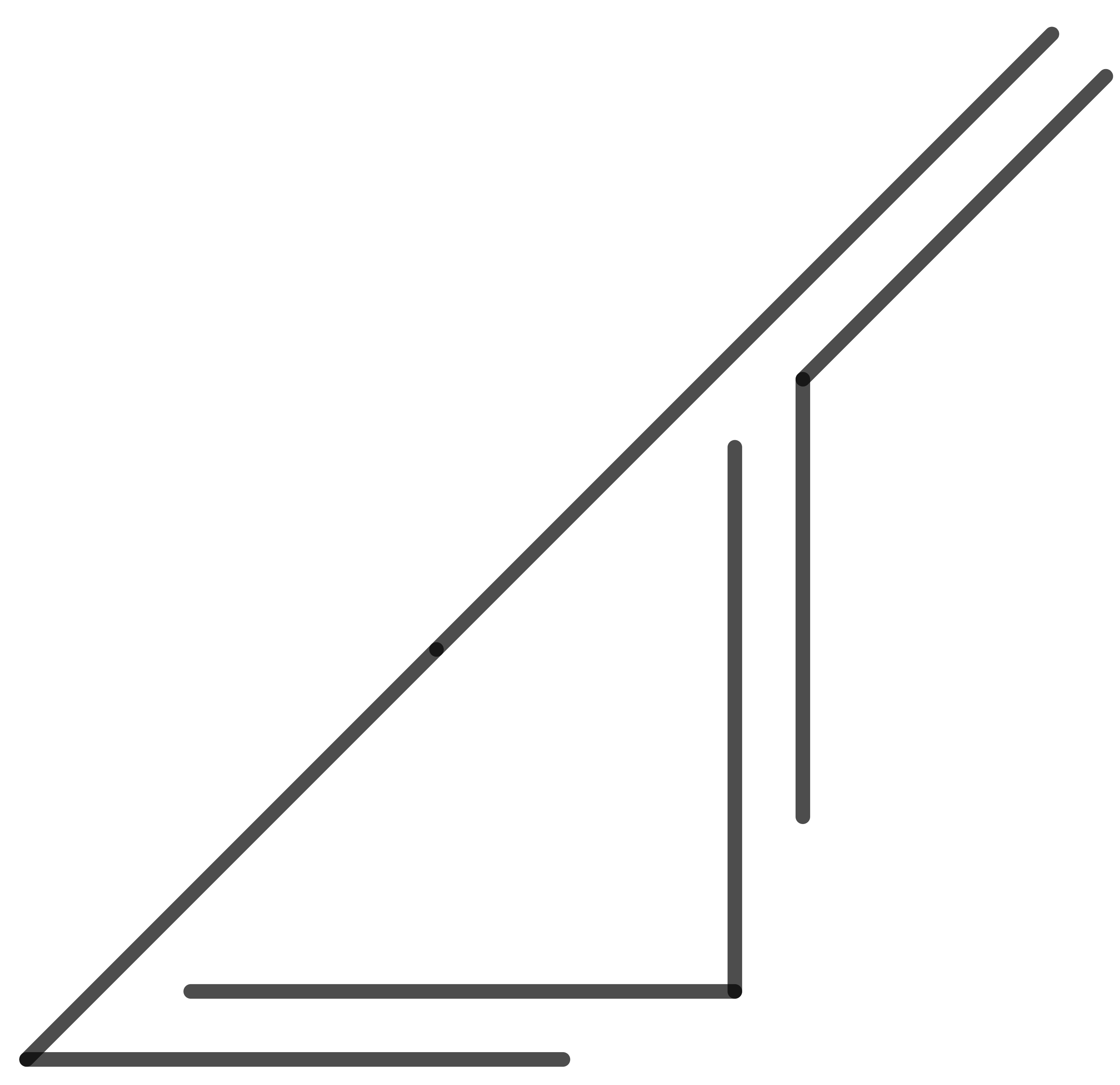}}
\qquad
\subfigure[][]{\includegraphics[scale=0.6]{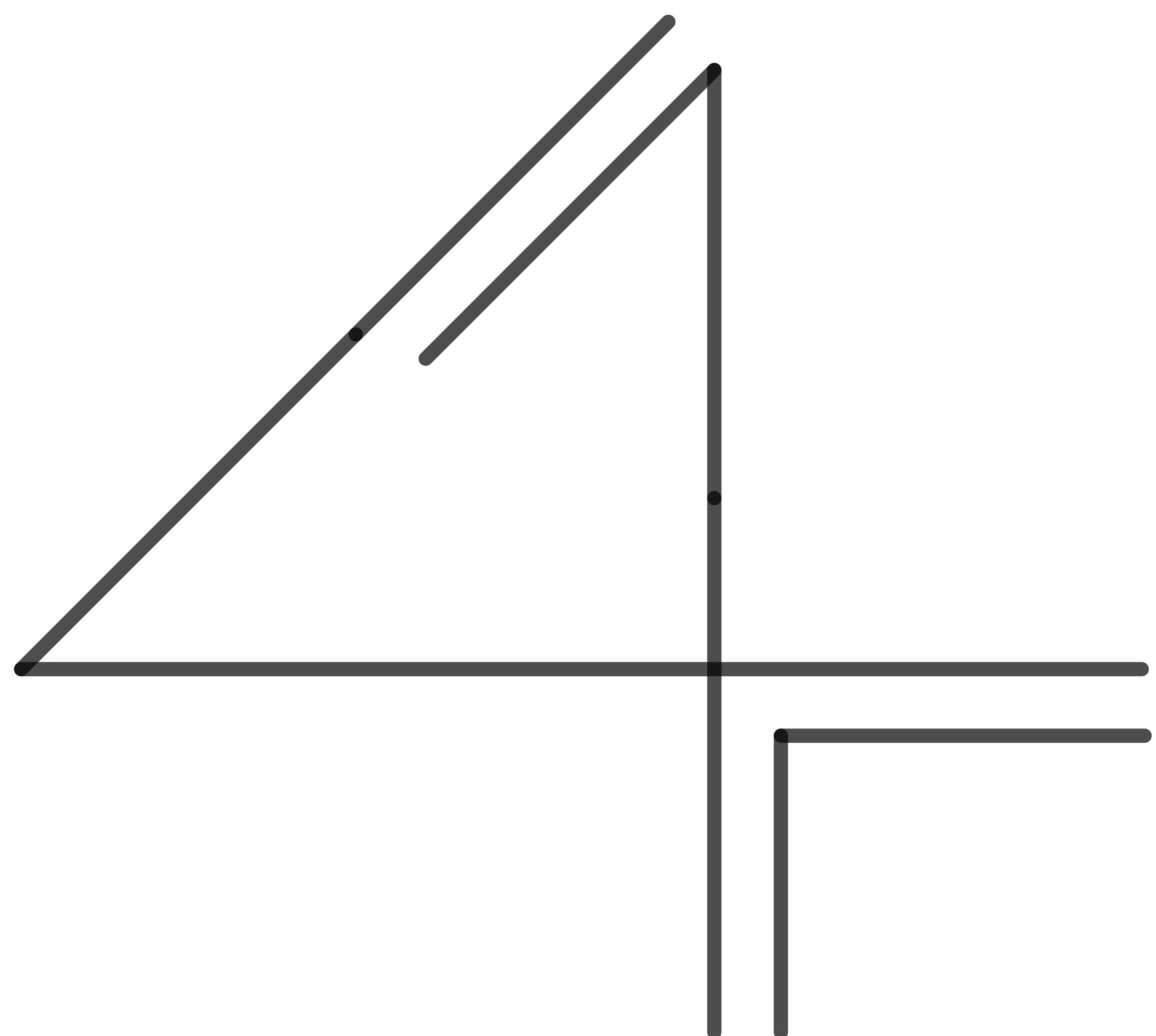}}
\qquad
\subfigure[][]{\includegraphics[scale=0.6]{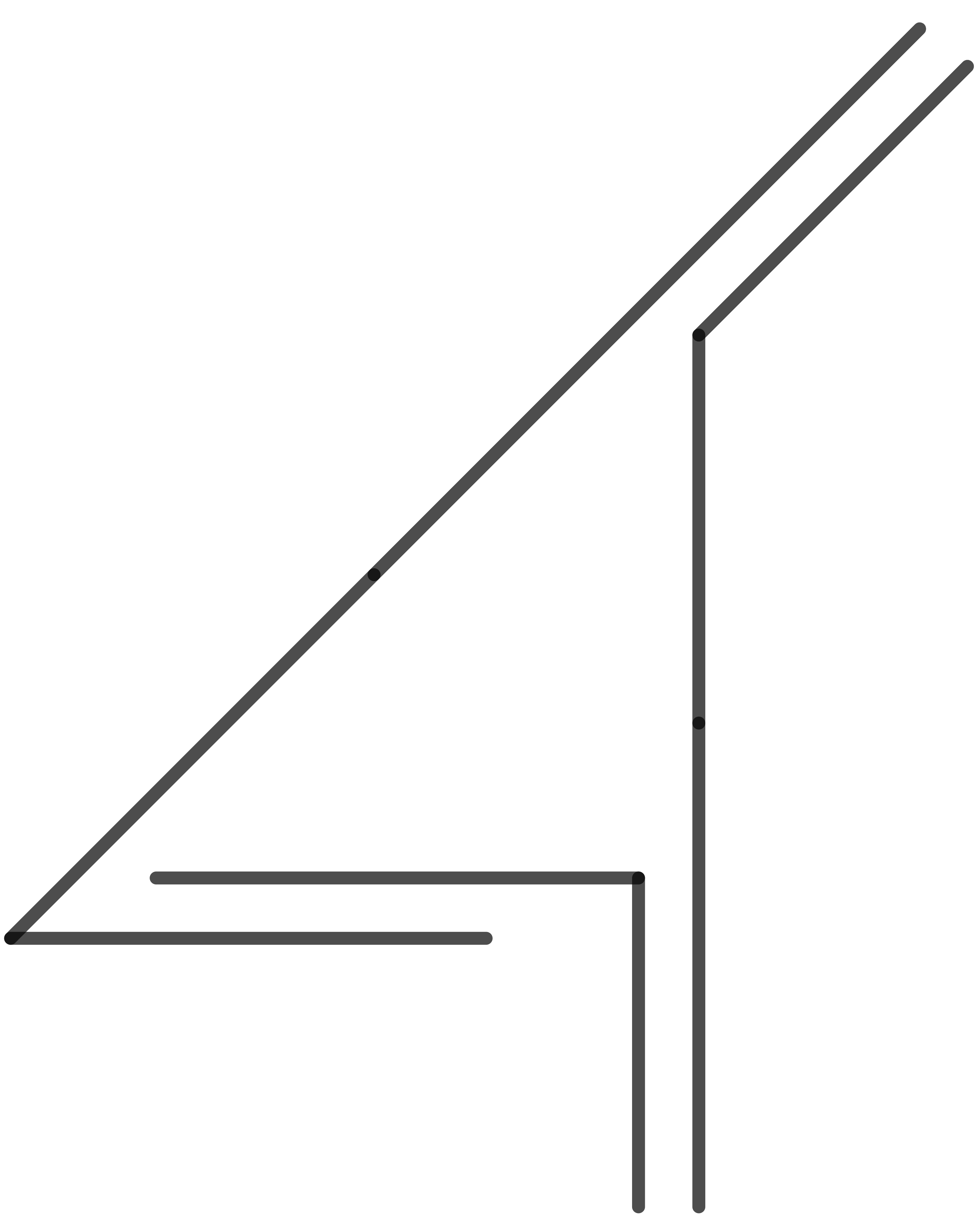}}
\qquad
\subfigure[][]{\includegraphics[scale=0.6]{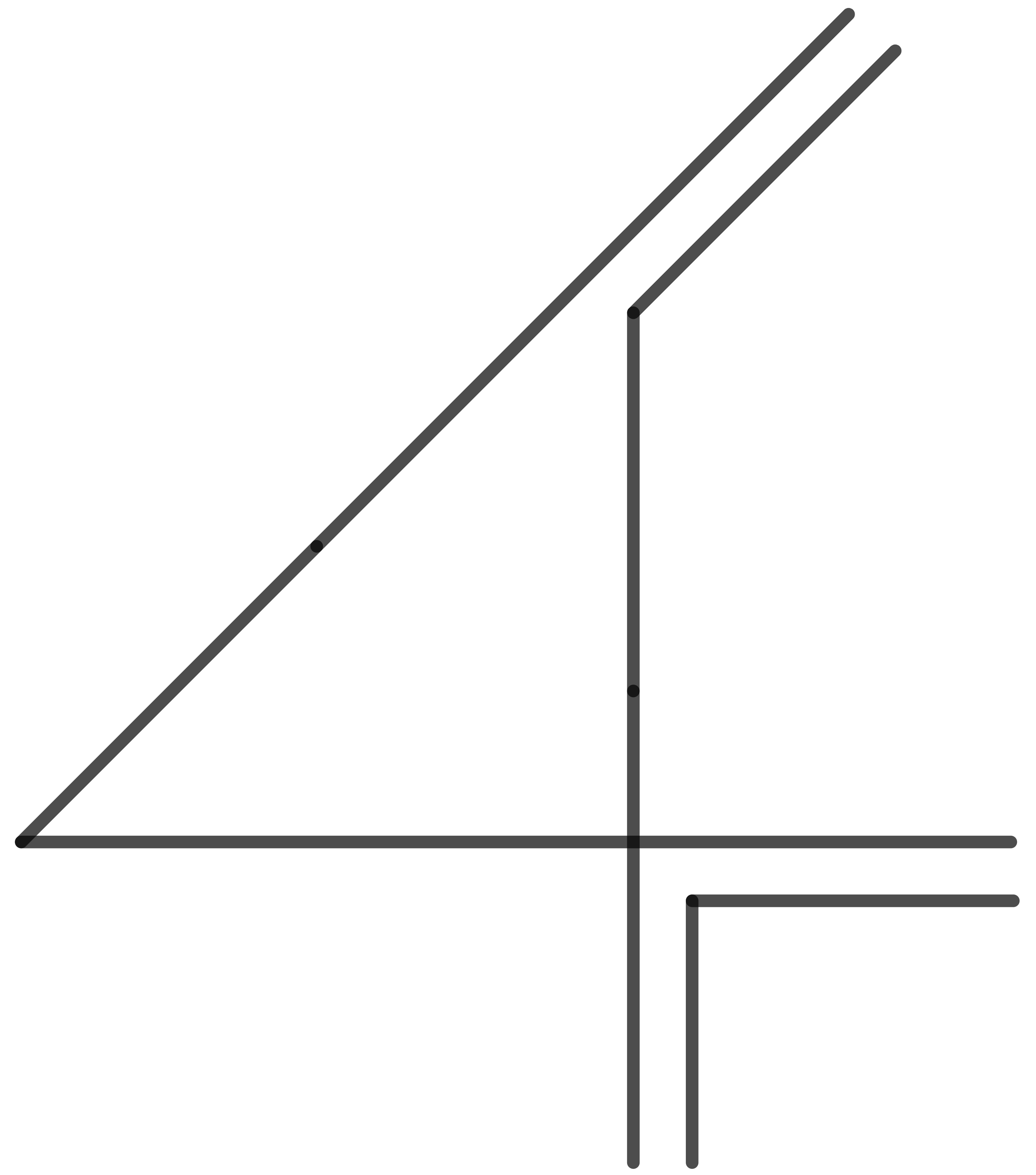}}
\qquad
\subfigure[][]{\includegraphics[scale=0.6]{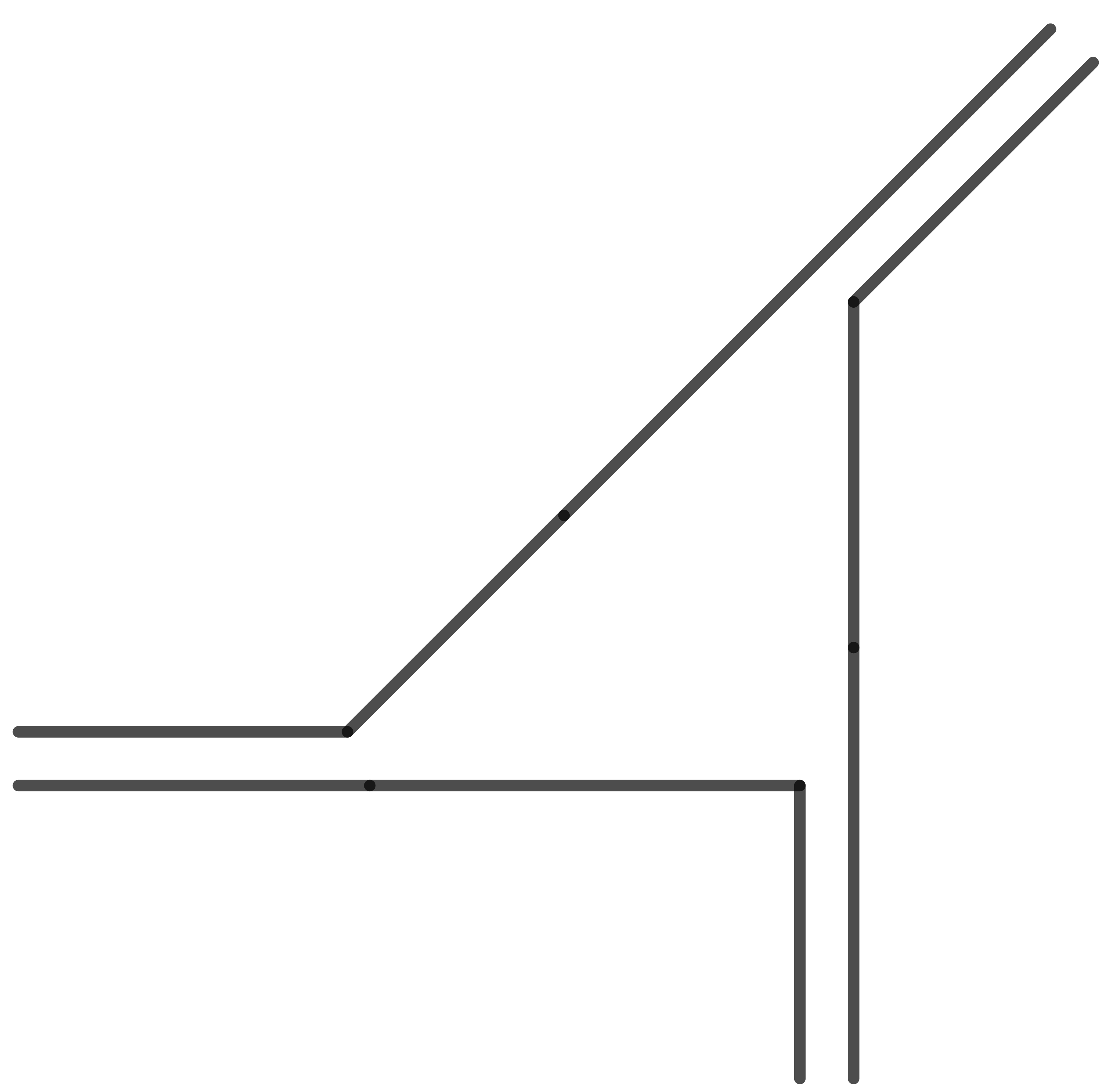}}
\caption{Examples of B$_1$-EPG\textsubscript{t} representations of a clique: (a) edge-clique, (b) claw-clique, (c) flag-clique, (d) paw-clique, (e) cricket-clique, (f) bull-clique, (g) extended-bull-clique and (h) net-clique.}
\label{fig: triangular_clique_examples}
\end{figure}

Now, let us turn our attention to triangular-cliques with three vertices. Let $C$ be a clique with three vertices and $\mathcal{P}=\{P_1, P_2, P_3\}$ the set of paths representing the vertices of $C$, such that $U(\mathcal{P})$ has a right triangle $T$ as a subgraph. In this case, there is exactly one path of $\mathcal{P}$ bending at each corner of $T$. Each $P_i\in \mathcal{P}$ is either an inside path (I), a midway path (M) or an outside path (O) relative to its respective corner. Let us use the notation $(x, y, z)$, where $x,y,z\in\{\text{I, M, O}\}$, in which the $i$-th coordinate indicates how $P_i$ bend relative to its respective corner. These are not ordered triples, that is, $(\text{I, M, O}) = (\text{I, O, M})$. Thus, there are ten different triples: $(\text{I, I, I})$, $(\text{M, M, M})$, $(\text{O, O, O})$, $(\text{I, I, M})$, $(\text{I, I, O})$, $(\text{M, M, I})$, $(\text{M, M, O})$, $(\text{O, O, I})$, $(\text{O, O, M})$ and $(\text{I, M, O})$. See in Table~\ref{Tab:table-cliques} how each one of the subtypes of the triangular-clique correspond to a different triple.

\begin{table}
\begin{center}
\begin{tabular}{| c | c | c |}
\hline
Paths & Layout & Name\\
\hline
$(\text{I, I, I})$ & \includegraphics[scale=0.6]{figuras/flag_example.png} & Flag\\
\hline
$(\text{I, I, M})$ & \includegraphics[scale=0.6]{figuras/paw_example.png} & Paw\\
\hline
$(\text{I, I, O})$ & \includegraphics[scale=0.6]{figuras/cricket_example.png} & Cricket\\
\hline
$(\text{M, M, I})$ & \includegraphics[scale=0.6]{figuras/bull_example.png} & Bull\\
\hline
$(\text{I, M, O})$ & \includegraphics[scale=0.6]{figuras/extended-bull_example.png} & Extended-bull\\
\hline
$(\text{M, M, M})$ & \includegraphics[scale=0.6]{figuras/net_example.png} & Net\\
\hline
\end{tabular}
\end{center}
\caption{The types of triangular-clique according to how the paths bend relative to the corners of the triangle. The letters I, M and O stand for inside path, midway path and outside path respectively.}
\label{Tab:table-cliques}
\end{table}

Note that $(\text{M, M, O})$, $(\text{O, O, I})$, $(\text{O, O, M})$ and $(\text{O, O, O})$ are not associated with any of the subtypes. However, it does not imply that these triples form types of triangular-clique different from those already considered. It is easy to see that if there are at least two outside paths, they will not edge intersect one another and, therefore, it will not be a clique. In the case of $(\text{M, M, O})$, let $b_i$ be the bend point of $P_i$ for $i\in\{1, 2, 3\}$, and $(x, b_3)$ and $(y, b_3)$ edges of $P_3$ that do not belong to the right triangle. Note that, $P_1$ must contain $(x, b_3)$ and $P_2$ must contain $(y, b_3)$ or vice-versa. Since $P_1$ and $P_2$ must bend at $b_1$ and $b_2$ respectively, the only they would intersect each other is either if both of them are inside paths, or one of them is an inside path and the other is a midway path.

\begin{thm} \label{thm cliques}
Let $\langle \mathcal{P},\mathcal{G}\rangle$ be a B$_1$-EPG\textsubscript{t} representation of a graph $G$, and let $C=\{v_1,v_2,v_3\}$ be a maximal clique of $G$. Then, $C$ corresponds to either an edge-clique, a claw-clique or a triangular-clique.
\end{thm}
\begin{proof}
Let $P_1$, $P_2$ and $P_3$ be paths of $\mathcal{P}$ representing $v_1$, $v_2$ and $v_3$, and $b_1$, $b_2$ and $b_3$ their corresponding bend points, respectively.

Suppose one of the paths, say $P_1$, is such that $s_1=P_1\cap_e P_2$ and $s_2=P_1\cap_e P_3$ are both on a same segment of $P_1$, and let $l_1$ be the grid line that contains such a segment. If $s_1\cap_e s_2\neq\emptyset$, then $C$ is an edge clique. If $s_1\cap_e s_2=\emptyset$, note that $P_2$ and $P_3$ cannot share edges on $l_1$, since this would force them to have more than one bend. Suppose $P_2$ and $P_3$ bend on different grid points, that is $b_2\neq b_3$. Let $l_2$ and $l_3$ be the grid lines (other than $l_1$) that contain $b_2$ and $b_3$ respectively, such that $l_2$ contains a segment of $P_2$ and $l_3$ contains a segment of $P_3$. In this case, note that $P_2$ has all of its edges on grid lines $l_1$ and $l_2$, and $P_3$ has all of its edges on grid lines $l_1$ and $l_3$, since both of them are $1$-bend paths. But, since $P_2$ and $P_3$ don't share edges on $l_1$, this scenario is not possible because the rest of their edges are on different grid lines ($l_2$ and $l_3$). Thus, $P_2$ and $P_3$ must bend at the same grid point, and in this case $C$ is a claw clique.

Now, suppose $s_1=P_1\cap_e P_2$ and $s_2=P_1\cap_e P_3$ are on different segments of $P_1$. Then, $C$ clearly cannot be an edge clique. Let $P_1'$ and $P_1''$ be the segments forming $P_1$, and $l_1'$ and $l_1''$ the grid lines containing $P_1'$ and $P_1''$ respectively. Without loss of generality, assume $s_1\subseteq P_1'$. Let $(b_1, x_i)$, for all $1\leq i\leq4$, be the edges of $\mathcal{G}$ that contain $b_1$, such that $(b_1, x_i)\notin E(P_1)$. Note that, if $P_2$ and $P_3$ contain the same grid edge $(b_1, x_i)$ for some $i$, then $C$ is a claw clique. Assume $P_2$ and $P_3$ do not contain the same grid edge $(b_1, x_i)$ for some $i$. Thus, $b_2$ is on $l_1'$, $b_3$ is on $l_1''$ and $b_2\neq b_3$. Also, $s_3=P_2\cap_e P_3$ is not on grid lines $l_1'$ and $l_1''$. Therefore, there must exist a grid line $l$ such that $s_3\subset l$, $b_2\in l$ and $b_3\in l$. Moreover, $l$, $l_1'$ and $l_1''$ are in different directions. The path $P_1$ is either narrow, normal or wide. Without loss of generality, assume $P_1'$ is horizontal. Let us analyze each case separately:
\begin{itemize}
    \item[(i)] If $P_1$ is narrow, then $l$ must be a vertical grid line. Thus, $C$ must be a triangular-clique. See in Figure~\ref{fig: narrow-case} that $C$ is either a flag-clique, a paw-clique, a bull-clique, an extended-bull-clique or a cricket-clique.
    \item[(ii)] If $P_1$ is normal, then $l$ must be a diagonal grid line. Since all cliques with three vertices, having at least one of them represented by a narrow path, were analyzed in the previous case, we may assume here that $P_2$ and $P_3$ are not narrow. Note that $P_1$ can have one of the following shapes: $\llcorner, \lrcorner, \ulcorner$ and $\urcorner$. Note also that, since $\mathcal{G}$ has only one diagonal direction, we cannot flip horizontally (or vertically) an entire representation, since this could imply that $/$-shaped paths would be turned into $\backslash$-shaped ones. However, it is possible to perform two (horizontal or vertical) flip operations in a row and obtain an isomorphic representation to the one we started with. Thus, without loss of generality, we may assume that $P_1$ is either $\llcorner$-shaped or $\lrcorner$-shaped. In both cases $C$ is a triangular-clique. If $P_1$ is $\llcorner$-shaped, $C$ is either a net-clique (see Figure~\ref{fig: regular-1}) or a bull-clique (see Figure~\ref{fig: regular-1_1}). If $P_1$ is $\lrcorner$-shaped, $C$ is either a bull-clique (see Figure~\ref{fig: regular-2}) or a cricket-clique (see Figure~\ref{fig: regular-2_1}).
    \item[(iii)] If $P_1$ is wide, then $l$ must be a vertical grid line. Since all cliques with three vertices, having at least one of them represented by a narrow or a normal path, were analyzed in the previous cases, we may assume here that $P_2$ and $P_3$ are both wide as well. Note, however, that we are assuming $P_1'$ is horizontal and $s_1\subseteq P_1'$. Therefore, $P_2$ must be a normal path (contradiction!). Thus, $P_1, P_2$ and $P_3$ cannot all be wide.
\end{itemize}
\begin{figure}[htb]
    \centering
    \subfigure[][]{\includegraphics[scale=0.5]{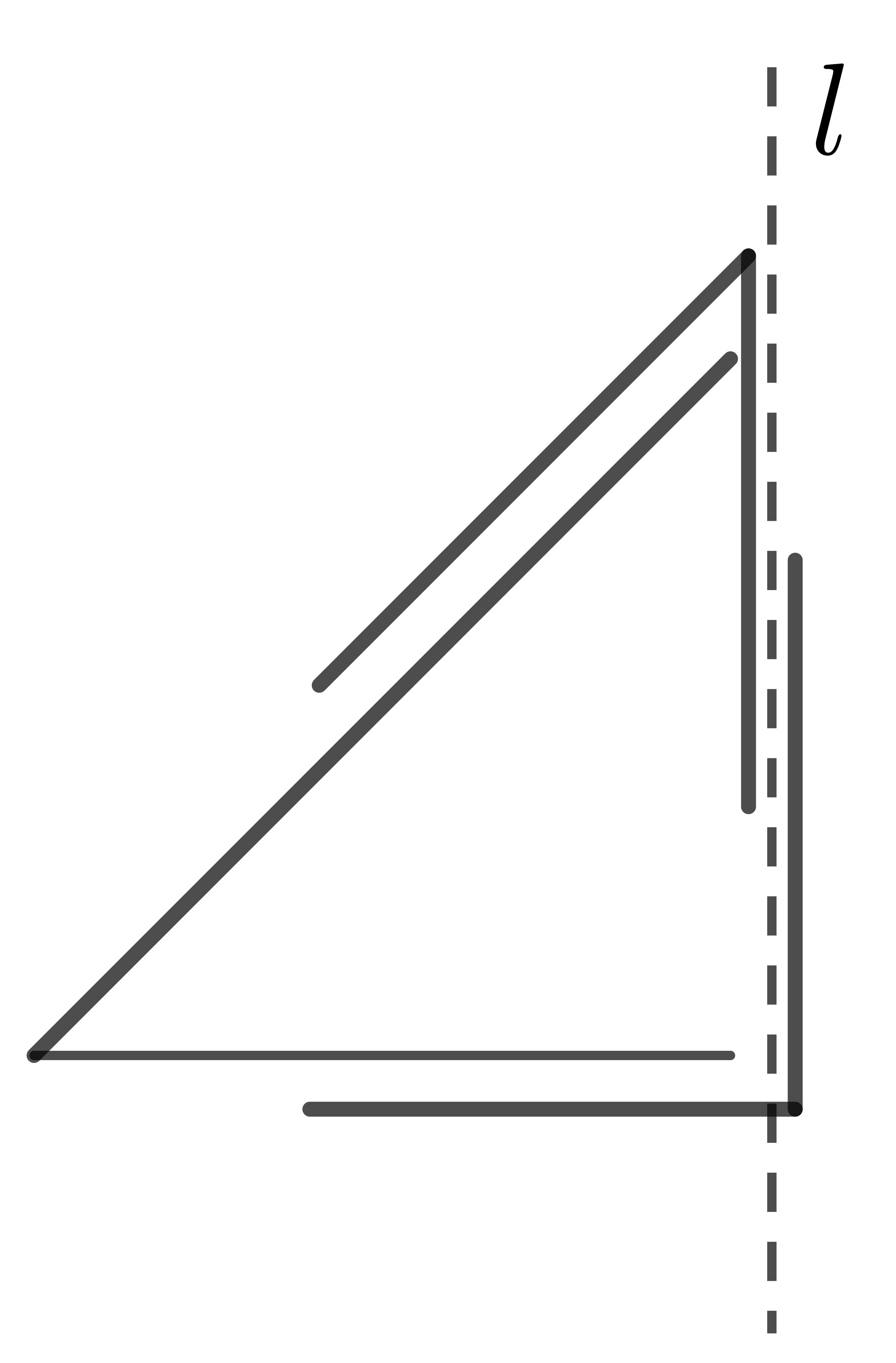} \label{fig: narrow-1}}
    \qquad
    \subfigure[][]{\includegraphics[scale=0.5]{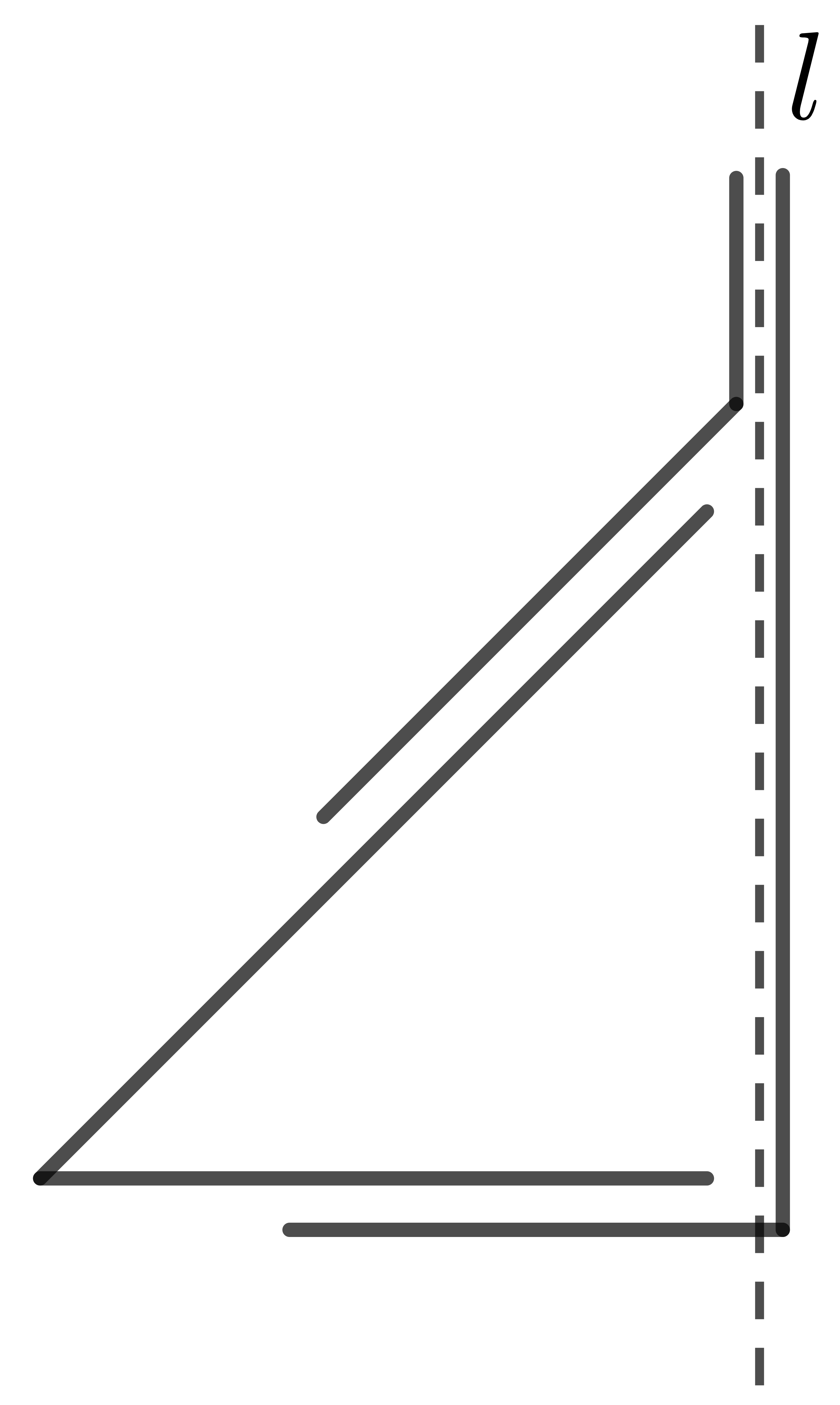} \label{fig: narrow-2}}
    \qquad
    \subfigure[][]{\includegraphics[scale=0.5]{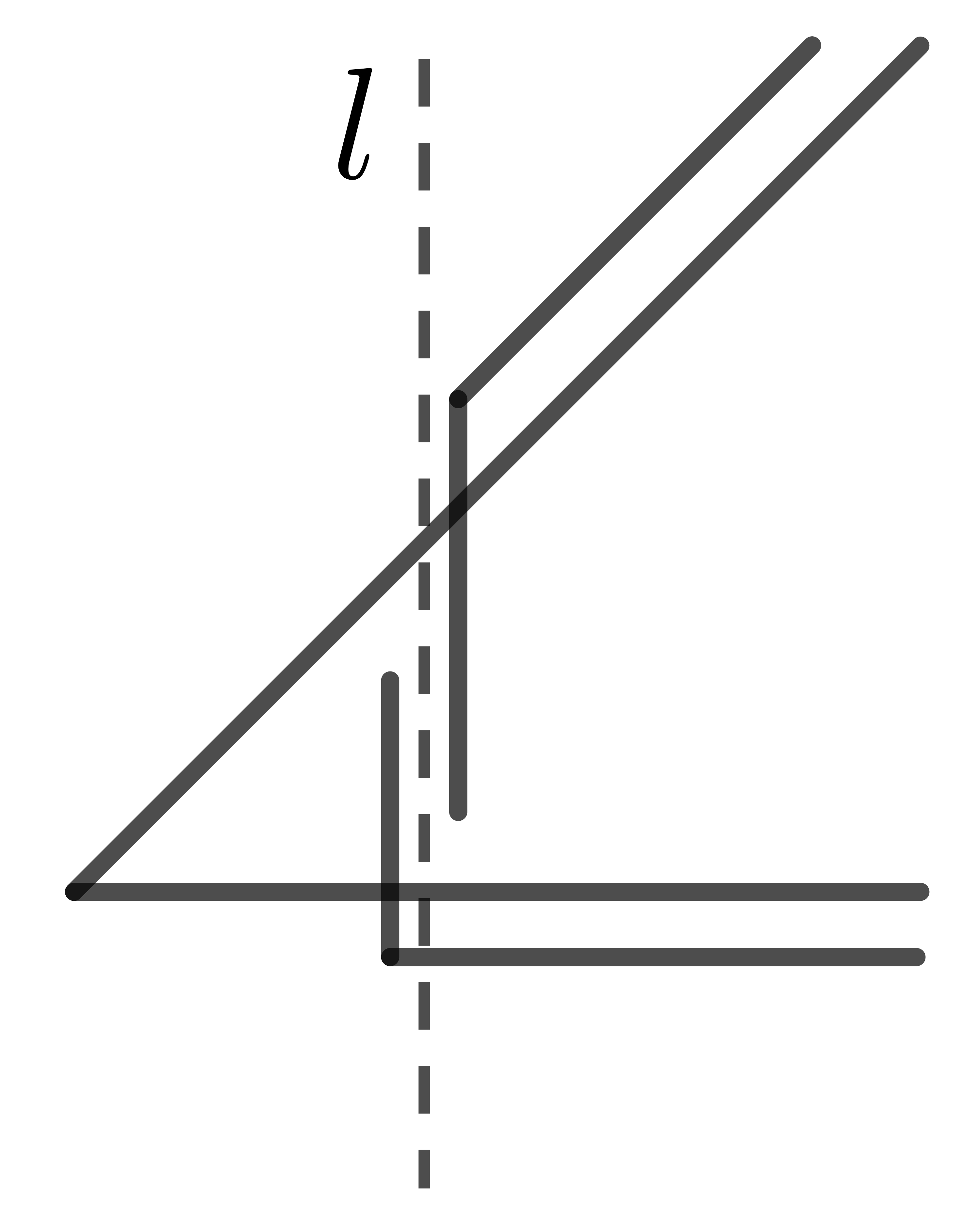} \label{fig: narrow-3}}
    \qquad
    \subfigure[][]{\includegraphics[scale=0.5]{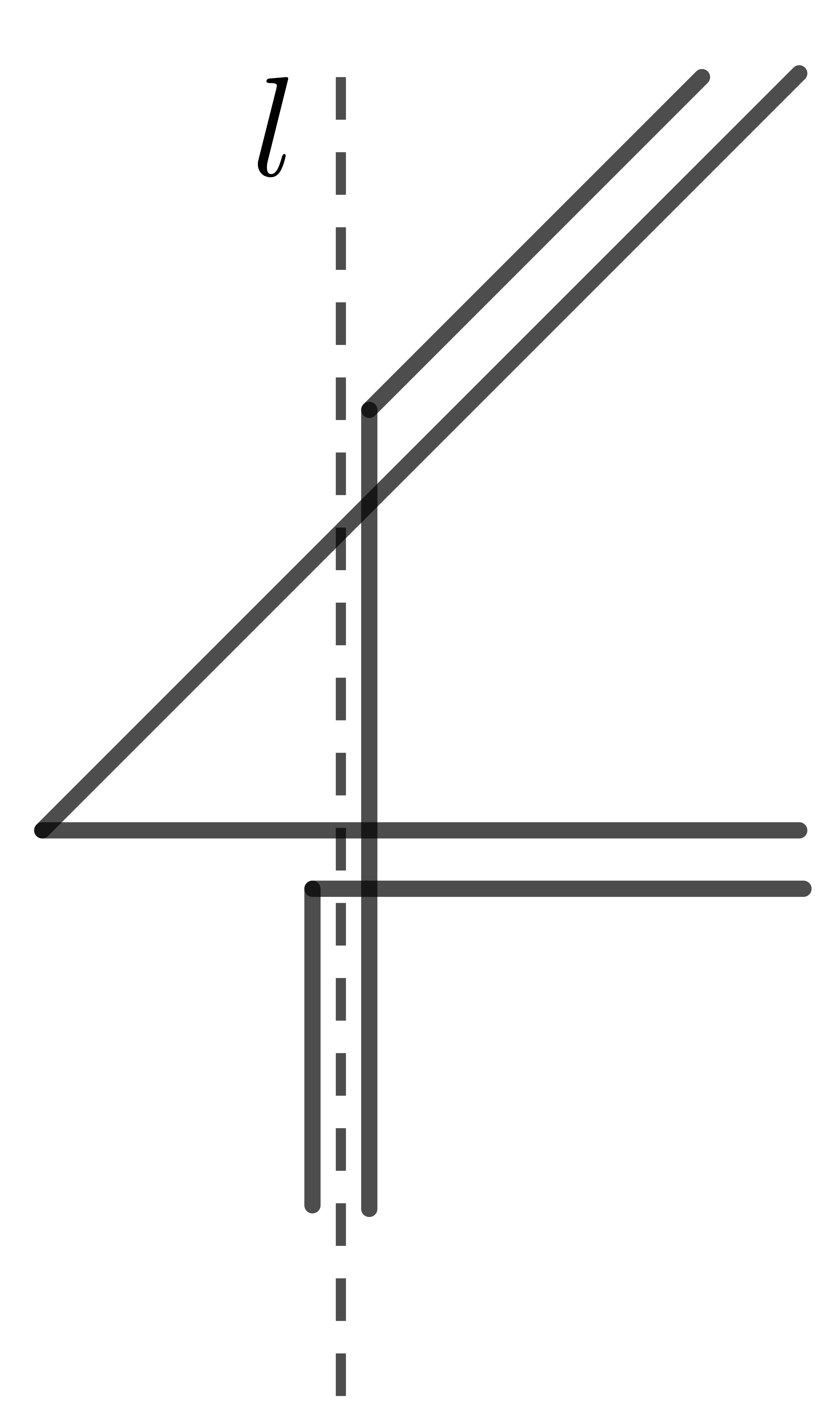} \label{fig: narrow-4}}
    \qquad
    \subfigure[][]{\includegraphics[scale=0.5]{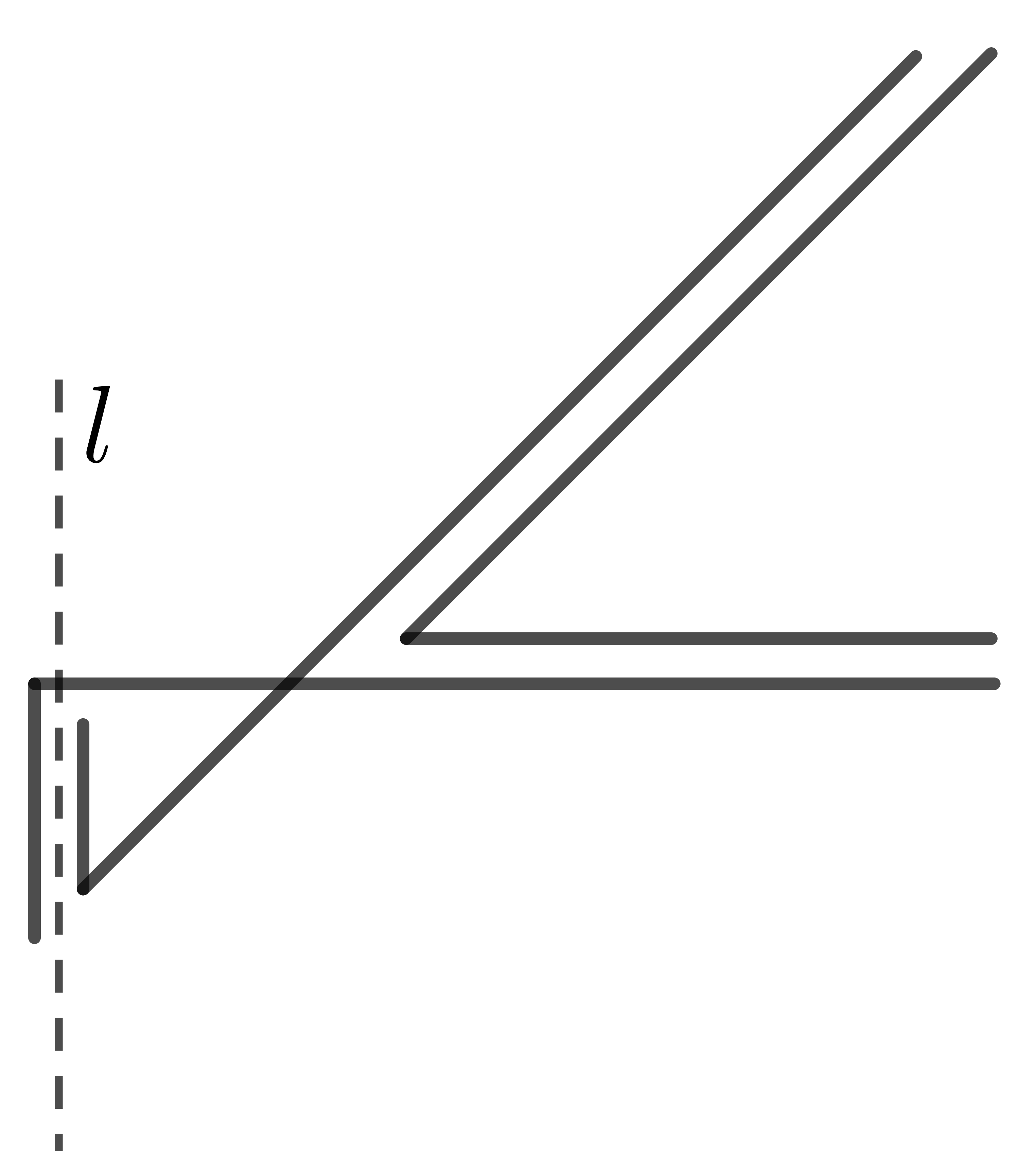} \label{fig: narrow-5}}
    \qquad
    \subfigure[][]{\includegraphics[scale=0.5]{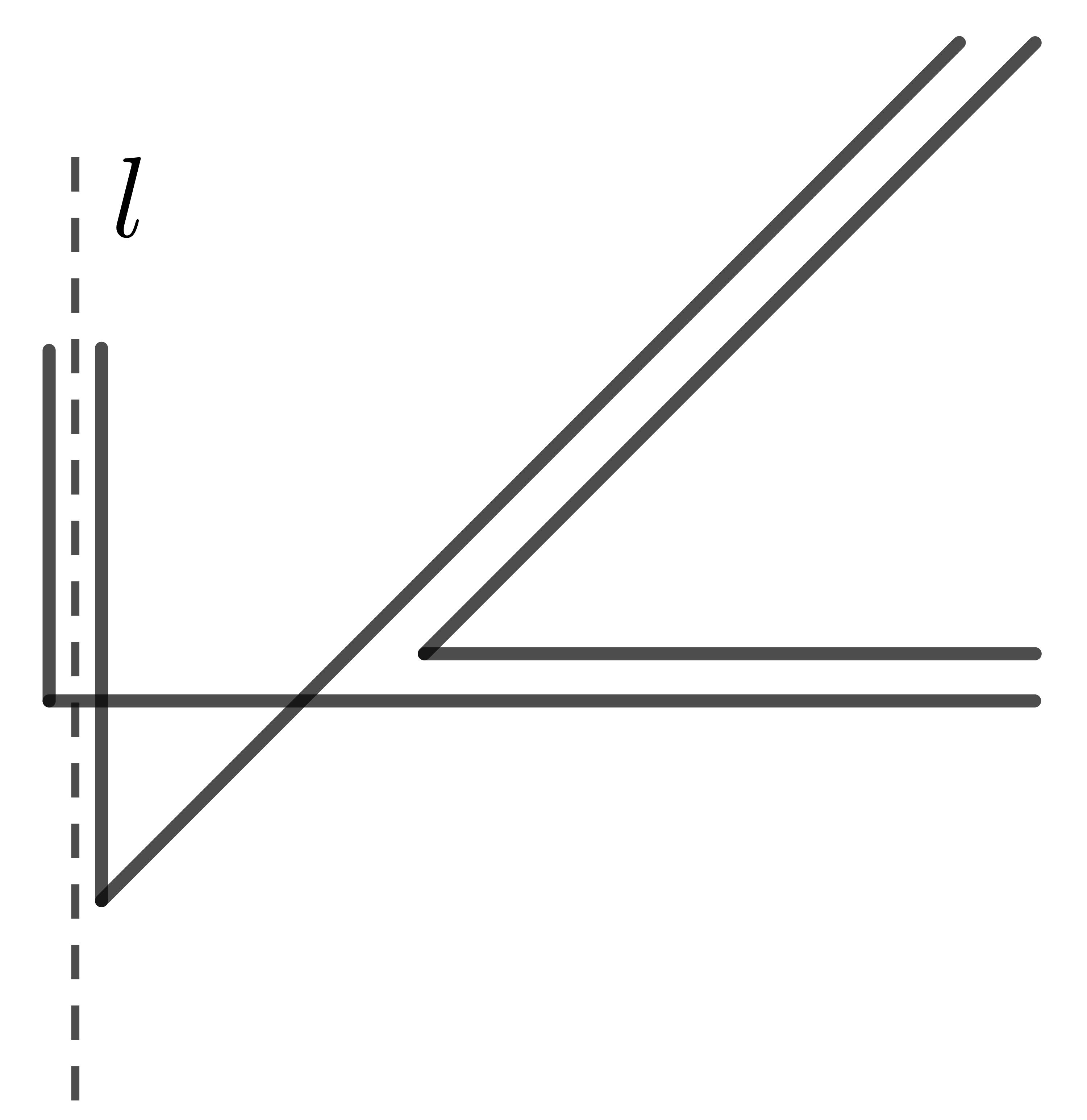} \label{fig: narrow-6}}
    \caption{The possible configurations derived from item (i) of the proof of Theorem~\ref{thm cliques}.}
    \label{fig: narrow-case}
\end{figure}
\begin{figure}[htb]
    \centering
    \subfigure[][]{\includegraphics[scale=0.5]{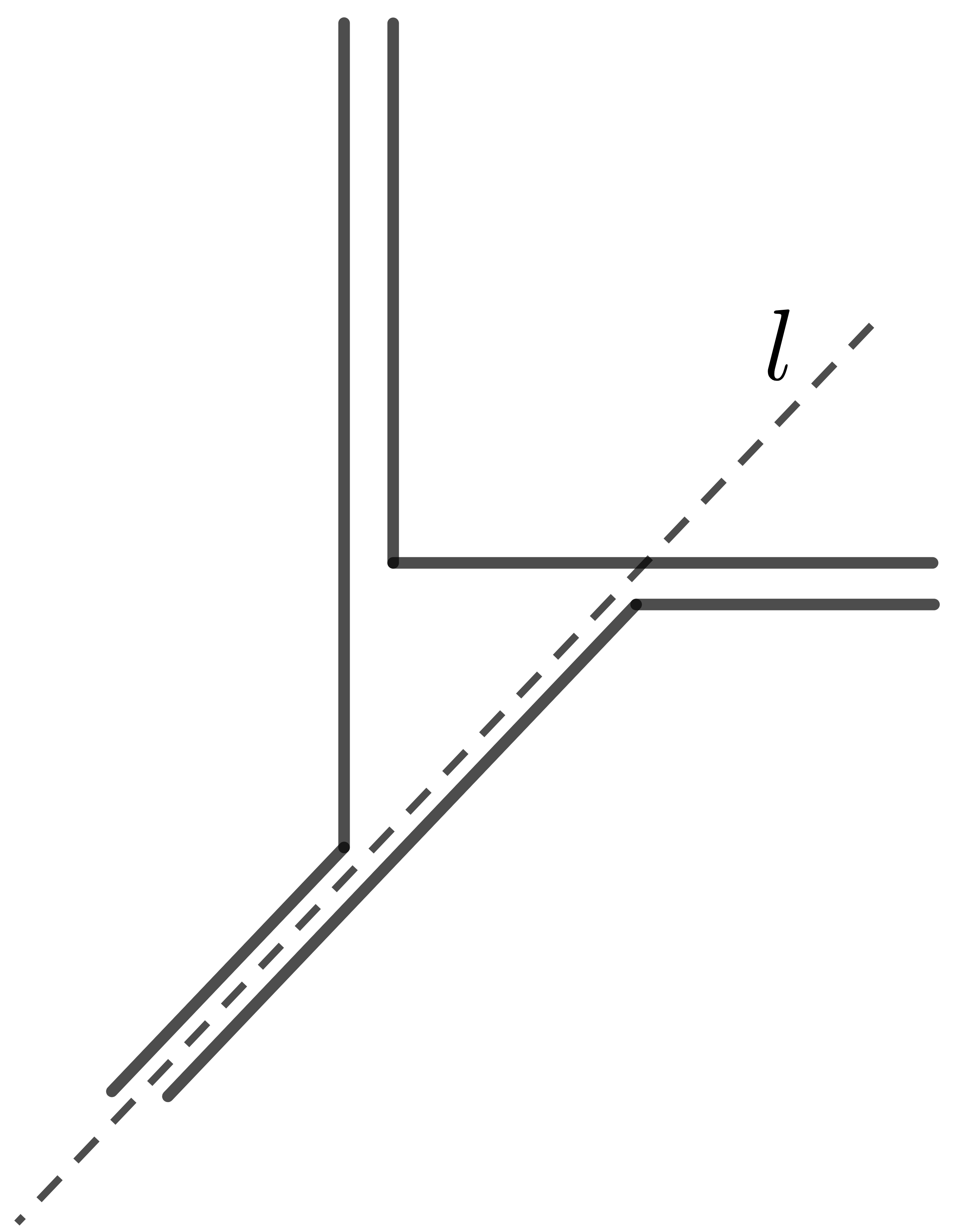} \label{fig: regular-1}}
    \qquad
    \subfigure[][]{\includegraphics[scale=0.4]{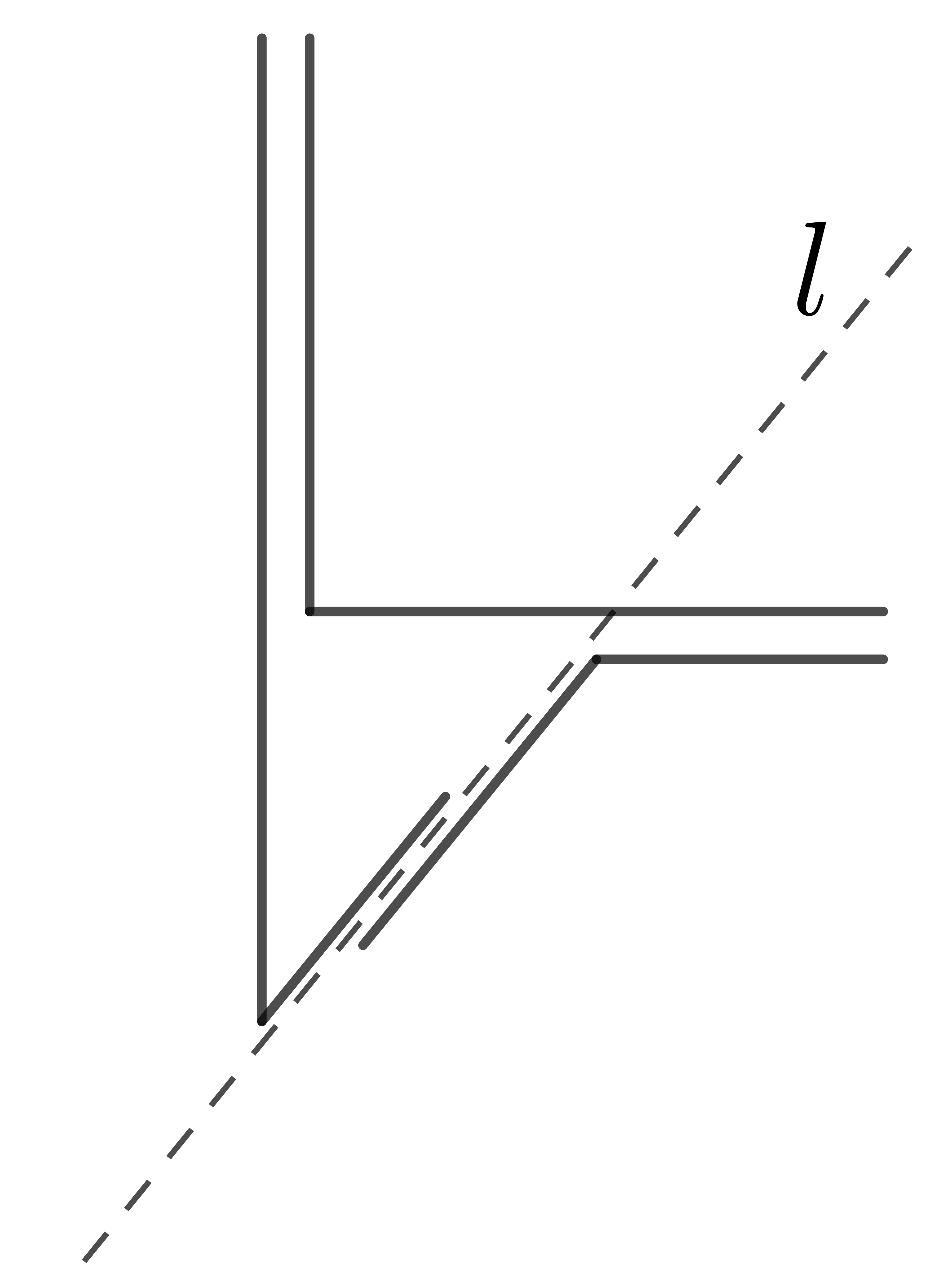} \label{fig: regular-1_1}}
    \qquad
    \subfigure[][]{\includegraphics[scale=0.5]{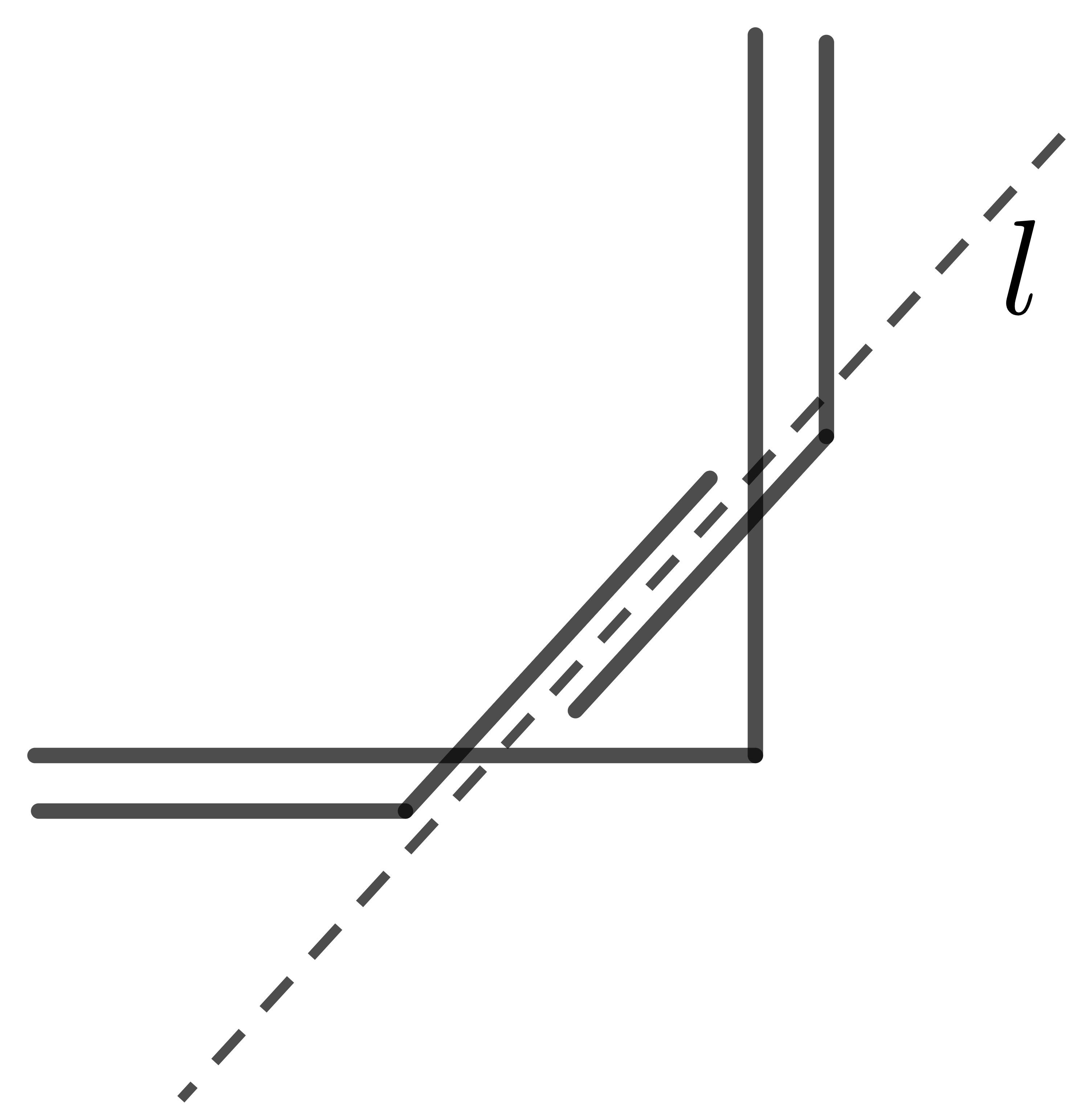} \label{fig: regular-2}}
    \qquad
    \subfigure[][]{\includegraphics[scale=0.4]{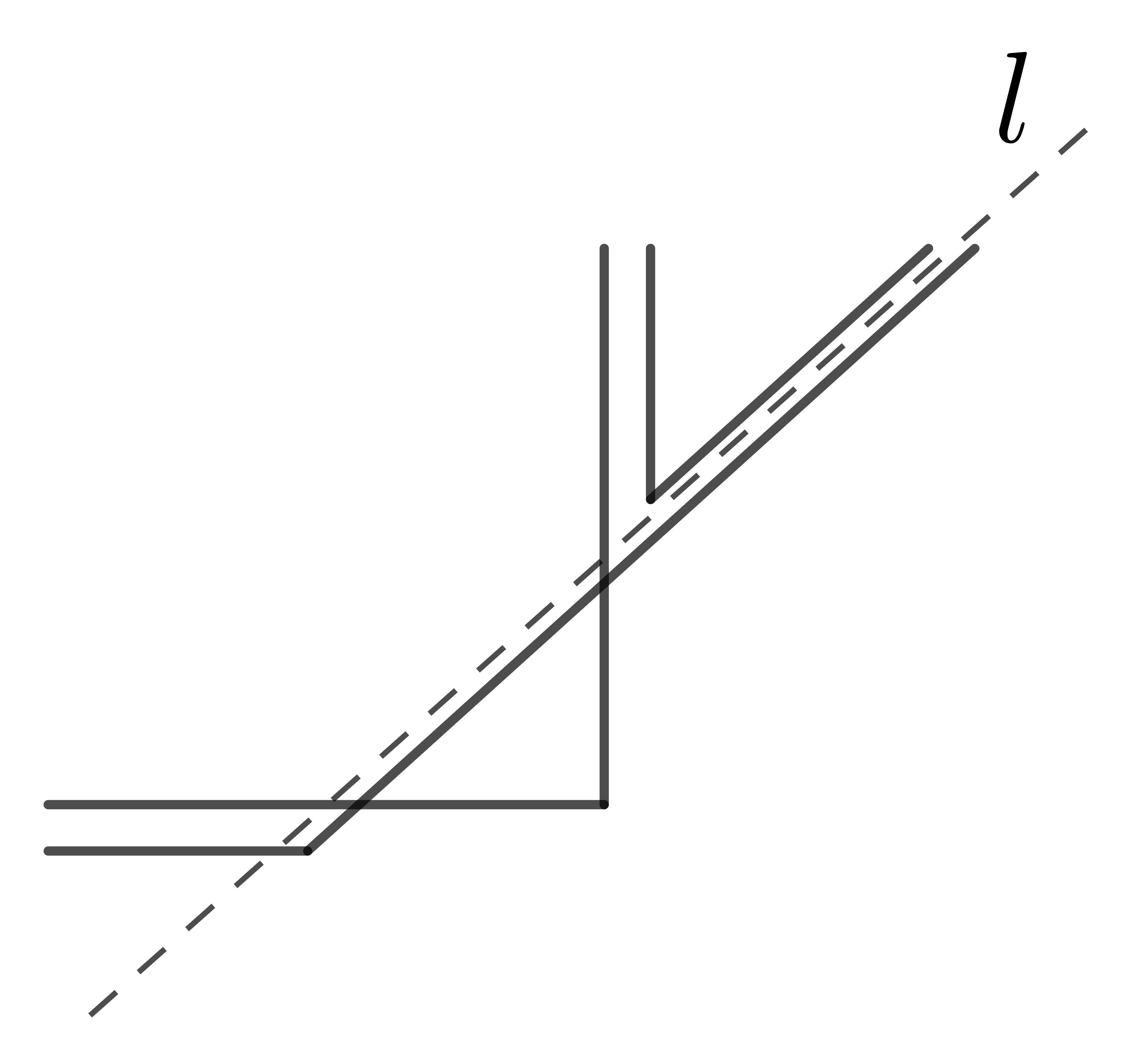} \label{fig: regular-2_1}}
    \caption{The possible configurations derived from item (ii) of the proof of Theorem~\ref{thm cliques}.}
    \label{fig: regular-case}
\end{figure}
This concludes the proof of the theorem.
\end{proof}

\section{Cycles on B$_1$-EPG\textsubscript{\lowercase{t}} Representations}
In this section, we characterize the B$_1$-EPG\textsubscript{t} representations of $4$-cycles.

Let $G$ be a chordless $4$-cycle and $\mathcal{P}=\{P_1, P_2, P_3, P_4\}$ the set of paths representing the vertices of $G$.

If $U(\mathcal{P})$ is a subdivision of a $4$-star, let $b$ be its central vertex and $a_1, a_2, a_3, a_4$ be the vertices of $U(\mathcal{P})$ such that $d(a_i)=1$, for all $1\leq i\leq4$. Consider the following cases:
\begin{itemize}
    \item If each $P_{a_i,b}\cup P_{a_{i+1},b}$ for $1\leq i\leq4$ is contained in a different member of $\mathcal{P}$, where addition is assumed to be modulo $4$, then $\mathcal{P}$ is called a \emph{true pie}. In a true pie, at least three of the four paths bend at $b$. See Figure~\ref{fig: b1_cycles_examples_truepie}.
    
    \item If each $P_{a_1, b}\cup P_{a_2, b}, P_{a_2, b}\cup P_{a_4, b}, P_{a_4, b}\cup P_{a_3, b}, P_{a_3, b}\cup P_{a_1, b}$ is contained in a different member of $\mathcal{P}$, then $\mathcal{P}$ is called a \emph{false pie}. In a false pie, at least two of the paths bend at $b$. See Figure~\ref{fig: b1_cycles_examples_falsepie}.
\end{itemize}

Let $Q$ be quadrilateral subgraph of $\mathcal{G}$ of any size, and let $s_1$, $s_2$, $s_3$ and $s_4$ be the segments of $\mathcal{G}$ forming the sides of $Q$, such that $s_i\cap_v s_{i+1}\neq\emptyset$ for $1\leq i\leq4$, where addition is assumed to be modulo $4$. We call $s_i\cap_v s_{i+1}$ the \emph{corners} of $Q$. If $Q$ is a subgraph of $U(\mathcal{P})$, each corner of $Q$ is the bend for a different member of $\mathcal{P}$, $P_2\cap_e P_3\neq\emptyset$, $P_3\cap_e P_4\neq\emptyset$, $P_4\cap_e P_1\neq\emptyset$, $P_2\cap_e P_4=\emptyset$ and $P_1\cap_e P_3=\emptyset$, then $\mathcal{P}$ is called a \emph{frame}. Consider the following cases:
\begin{itemize}
    \item If $Q$ is a rectangle, then $\mathcal{P}$ is called a \emph{rectangle frame} (or, simply, an \emph{r-frame}). See Figure~\ref{fig: b1_cycles_examples_rectangle}.
    
    \item If $Q$ is a trapezoid, then $\mathcal{P}$ is called a \emph{trapezoidal frame} (or, simply, a \emph{t-frame}). See Figure~\ref{fig: b1_cycles_examples_trapezoidal}.
    
    \item If $Q$ is a parallelogram, but not a rectangle, then $\mathcal{P}$ is called a \emph{parallelogram frame} (or, simply, a \emph{p-frame}). See Figure~\ref{fig: b1_cycles_examples_parallelogram}.
\end{itemize}

Let $T$ be a right triangle in $\mathcal{G}$. If $T\subseteq U(\mathcal{P}_C)$, each corner of $T$ is the bend for at most two different members of $\mathcal{P}$, $P_2\cap_e P_3\neq\emptyset$, $P_3\cap_e P_4\neq\emptyset$, $P_4\cap_e P_1\neq\emptyset$, $P_2\cap_e P_4=\emptyset$ and $P_1\cap_e P_3=\emptyset$, then $\mathcal{P}$ is called a \emph{flag}. See Figure~\ref{fig: b1_cycles_examples_flag}.

Let $G=T_1\cup T_2$, such that $T_1\cap_v T_2=\{v\}$ where $v$ is a corner of both $T_1$ and $T_2$. If $G\subseteq U(\mathcal{P}_C)$, each corner of $G$ is the bend for a different member of $\mathcal{P}$, $P_2\cap_e P_3\neq\emptyset$, $P_3\cap_e P_4\neq\emptyset$, $P_4\cap_e P_1\neq\emptyset$, $P_2\cap_e P_4=\emptyset$ and $P_1\cap_e P_3=\emptyset$, then $\mathcal{P}$ is called a \emph{butterfly}. See Figure~\ref{fig: b1_cycles_examples_butterfly}.

Note that the existence of a third direction on the grid, when compared to a rectangular grid, allows the arising of new representations of a $4$-cycle. See in Figure~\ref{fig: b1_cycles_examples} some examples of representations of a $4$-cycle on a triangular grid.

\begin{figure}[htb]
\center
\subfigure[][]{\includegraphics[scale=0.6]{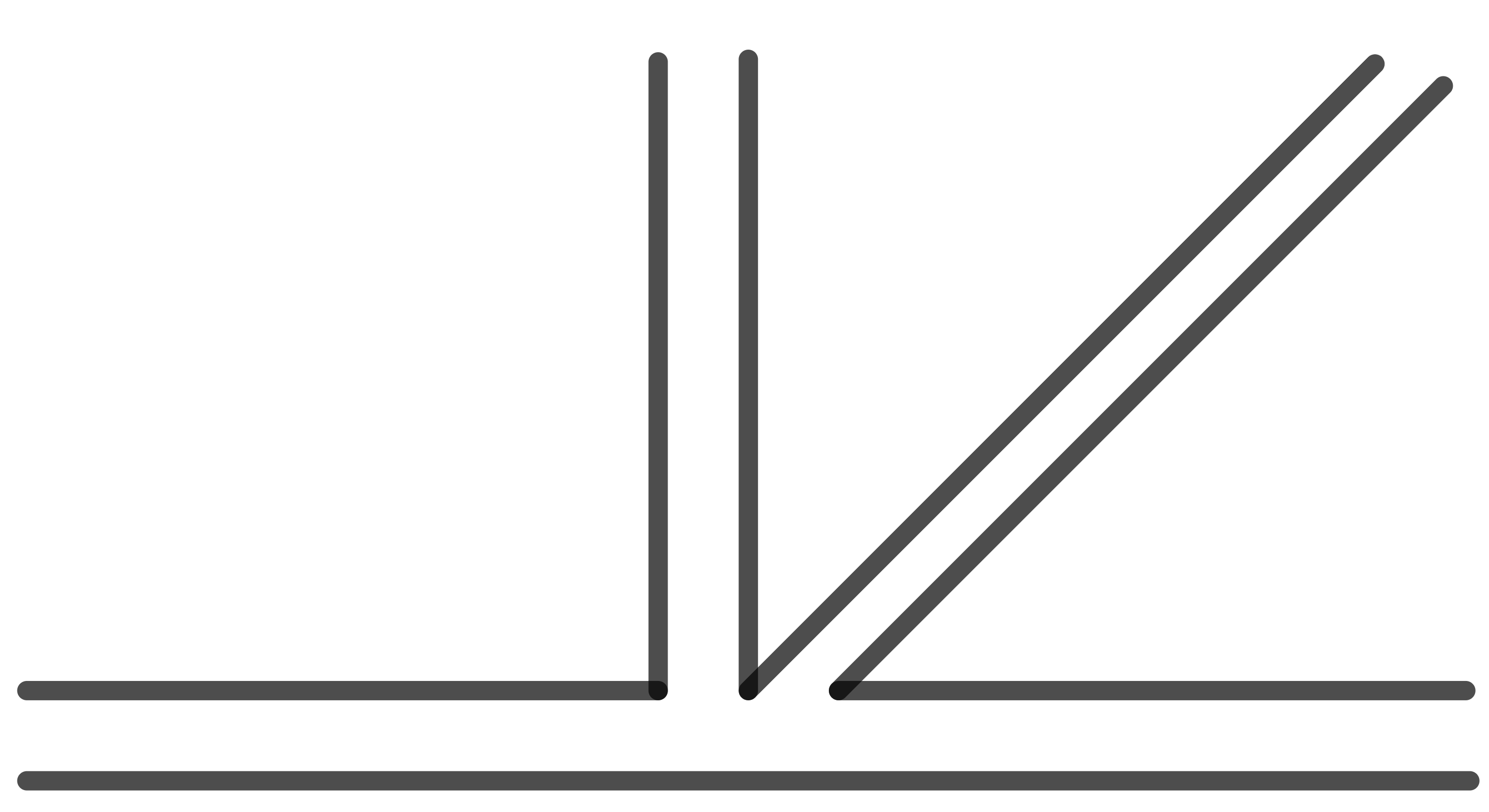} \label{fig: b1_cycles_examples_truepie}}
\qquad
\subfigure[][]{\includegraphics[scale=0.6]{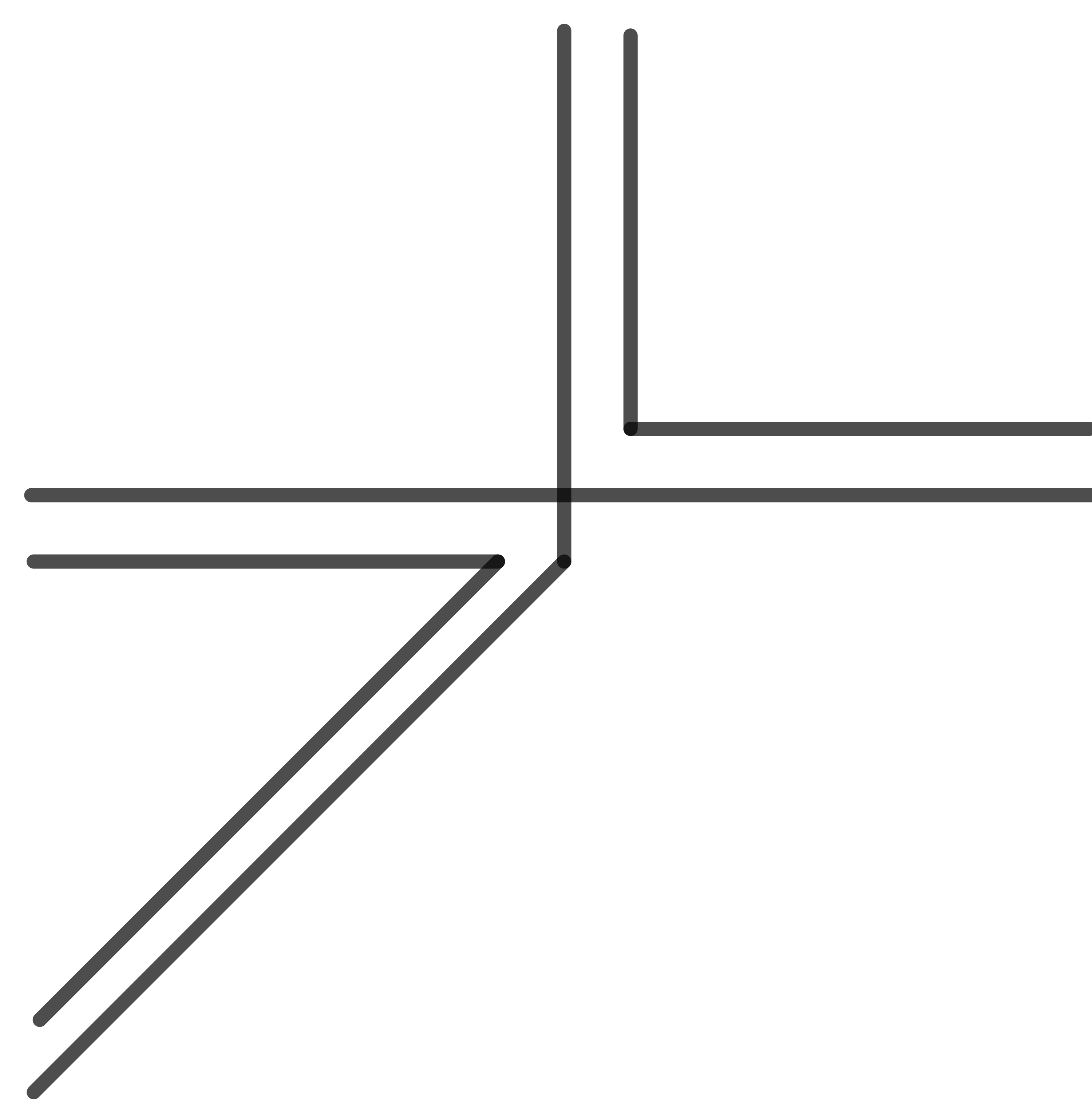} \label{fig: b1_cycles_examples_falsepie}}
\qquad
\subfigure[][]{\includegraphics[scale=0.5]{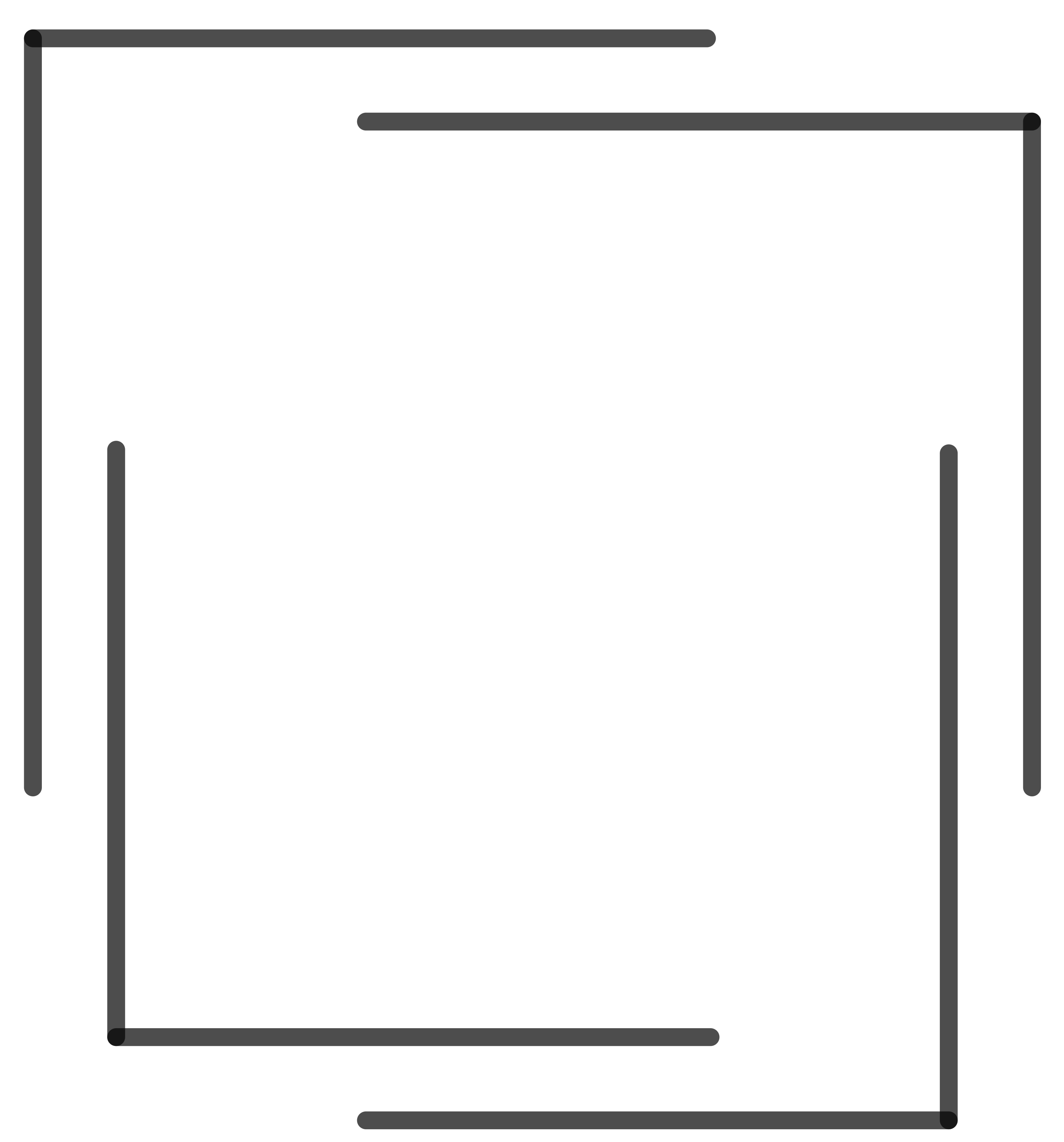} \label{fig: b1_cycles_examples_rectangle}}
\qquad
\subfigure[][]{\includegraphics[scale=0.5]{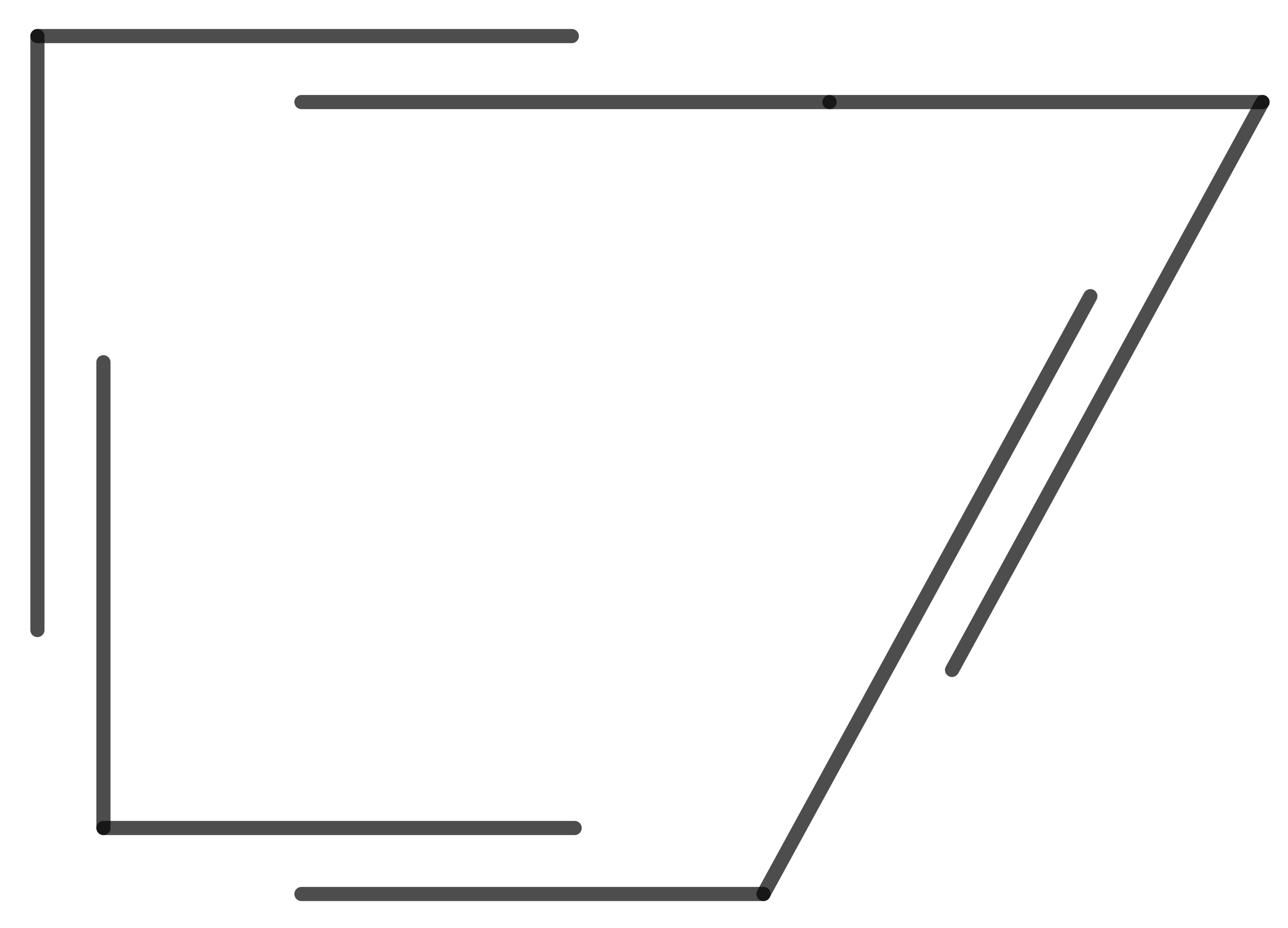} \label{fig: b1_cycles_examples_trapezoidal}}
\qquad
\subfigure[][]{\includegraphics[scale=0.5]{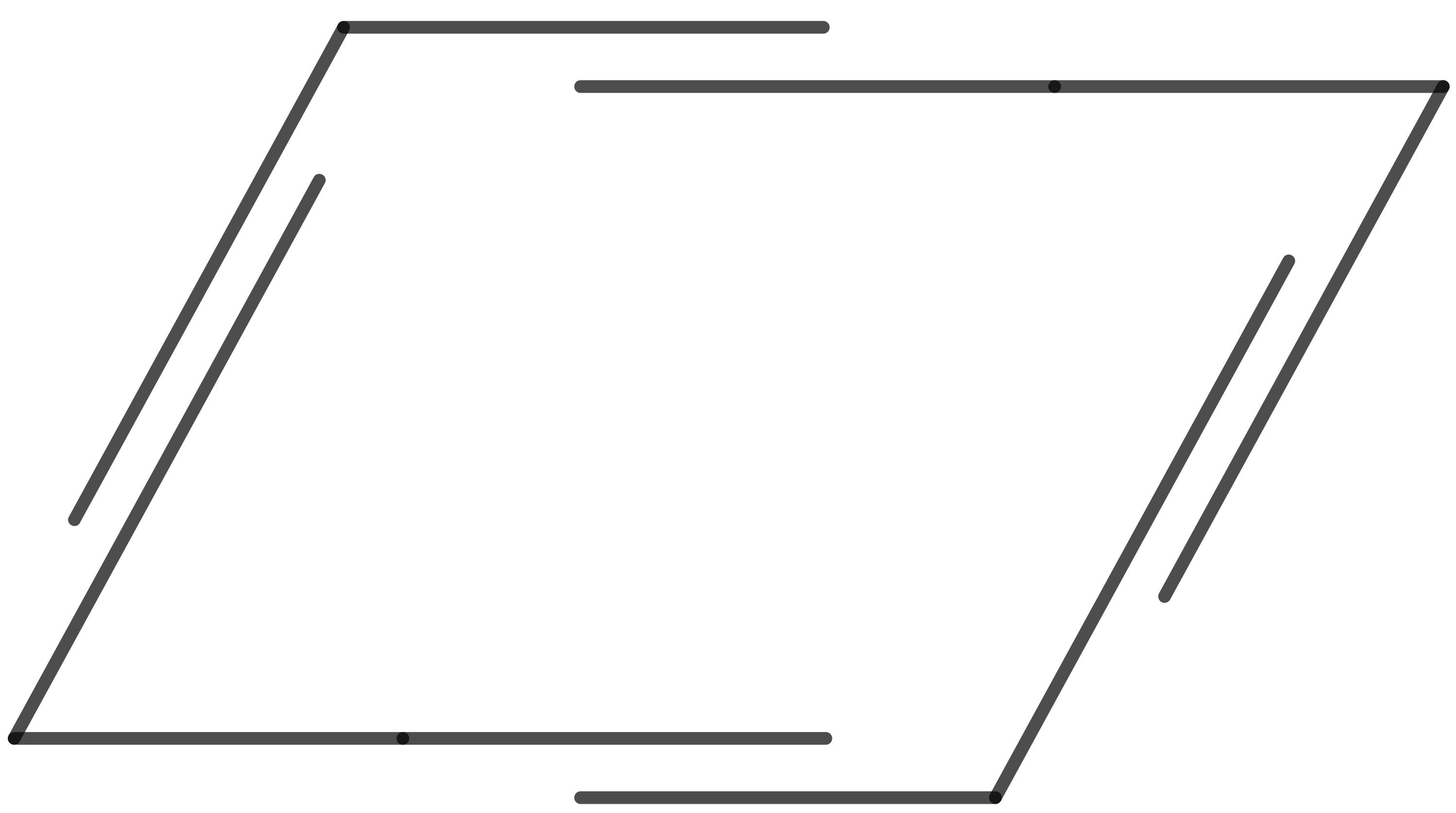} \label{fig: b1_cycles_examples_parallelogram}}
\qquad
\subfigure[][]{\includegraphics[scale=0.6]{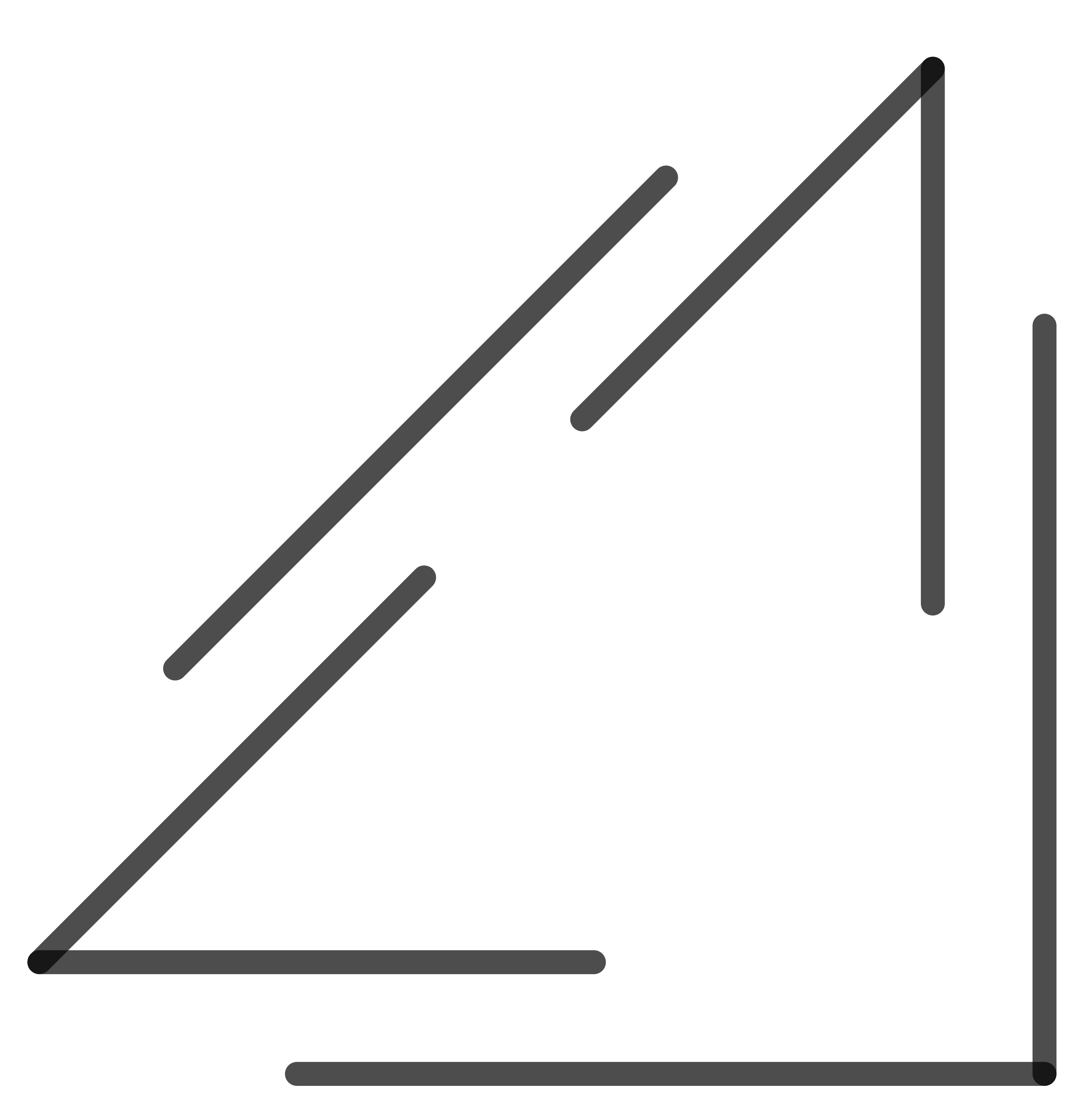} \label{fig: b1_cycles_examples_flag}}
\qquad
\subfigure[][]{\includegraphics[scale=0.5]{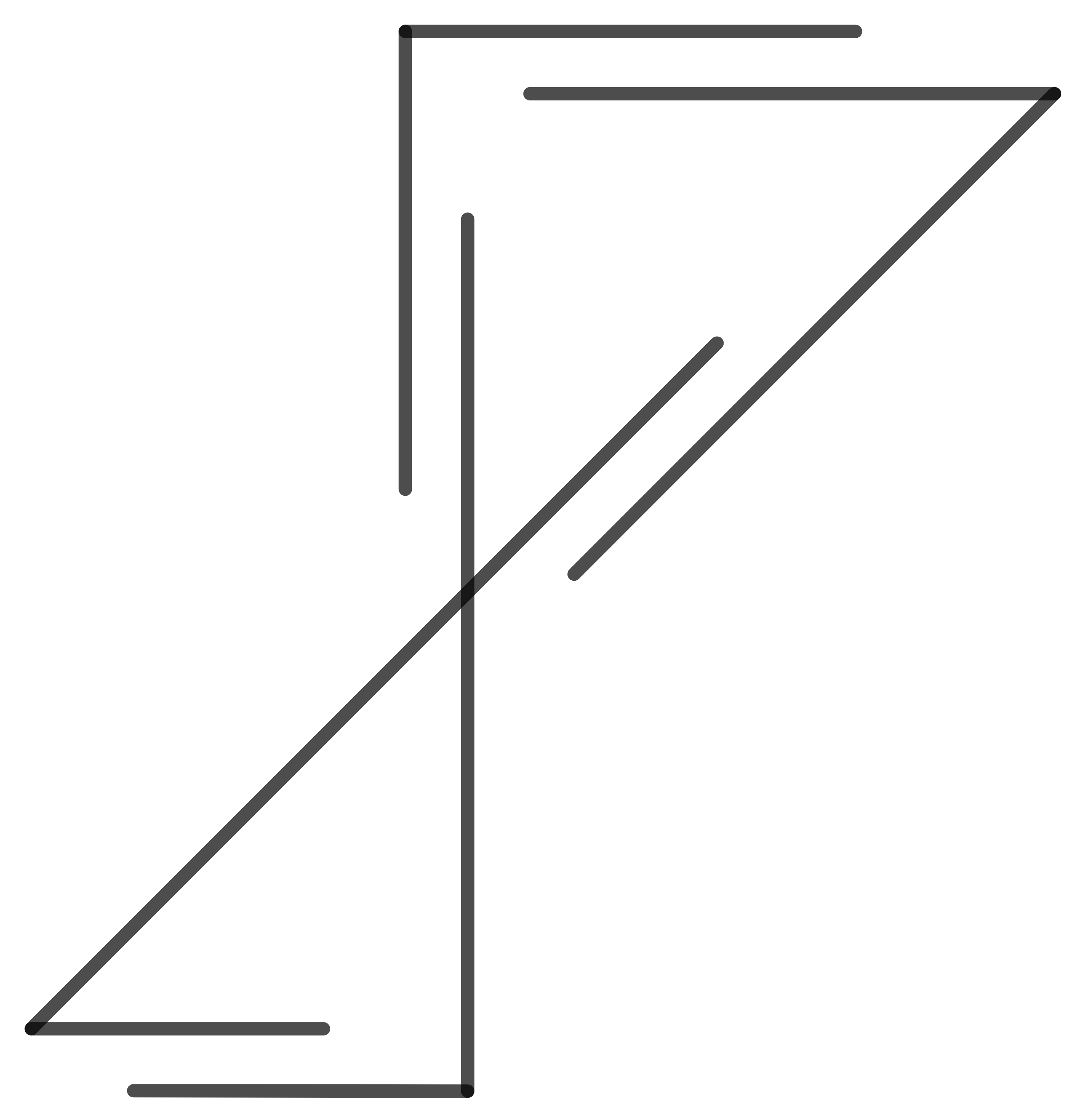} \label{fig: b1_cycles_examples_butterfly}}
\caption{Examples of B$_1$-EPG\textsubscript{t} representations of a $4$-cycle: (a) true-pie, (b) false-pie, (c) r-frame, (d) t-frame, (e) p-frame, (f) flag and (g) butterfly.}
\label{fig: b1_cycles_examples}
\end{figure}

\begin{thm} \label{thm_4-cycle_b1}
Let $\langle \mathcal{P},\mathcal{G}\rangle$ be a B$_1$-EPG\textsubscript{t} representation of a graph $G$. Then, every chordless $4$-cycle in $G$ corresponds to either a true pie, a false pie, a frame, a flag or a butterfly in $\mathcal{P}$.
\end{thm}
\begin{proof}
Let $C_4 = (v_1, v_2, v_3, v_4)$ be a chordless cycle in $G$. Let $\mathcal{P}_C=\{P_i\mid1\leq i\leq4\}$, where $P_i$ is the path in $\mathcal{G}$ corresponding to $v_i$. If $\mathcal{P}_C$ contains paths that use only two directions of $\mathcal{G}$, then $\mathcal{P}_C$ is an EPG representation. Thus, as proved in~\cite{golumbic2009edge}, $\mathcal{P}_C$ is either a true pie, a false pie or an r-frame. Therefore, from now on we are assuming $\mathcal{P}_C$ contains paths that use all three directions of $\mathcal{G}$.

Suppose $\bigcap_v P_i\neq\emptyset$, then clearly $\bigcap_v P_i=\{b\}$, for some grid point $b$. If each path $P_i$ contains exactly two grid edges with endpoint $b$, we obtain a star subgraph with center point $b$ and edges $(a_1, b)$, $(a_2, b)$, $(a_3, b)$, $(a_4, b)$.
Without loss of generality, $P_1$ contains the grid edges $(a_1, b)$, $(a_2, b)$. If $P_2$ contains the grid edges $(a_2, b)$, $(a_3, b)$ or $(a_4, b)$, $(a_1, b)$, then we obtain a true pie. Otherwise, $P_2$ contains the grid edges $(a_2, b)$, $(a_4, b)$ or $(a_1, b)$, $(a_3, b)$ and we obtain a false pie.

If at least one of $P_1,\ldots,P_4$ contains only one grid edge with endpoint $b$, let $P_1'$ and $P_1''$ be the segments forming $P_1$, and without loss of generality assume $P_1$ contains only one grid edge with endpoint $b$, $P_1'$ is horizontal and $b\in P_1'$. It is possible to assume that $P_1'$ is horizontal because the current assumption is that there are segments in all directions of $\mathcal{G}$, thus, without loss of generality $P_1$ is the path having a horizontal segment. Suppose $P_1'\cap_v P_2=\{b\}$ and $P_1'\cap_v P_4=\{b\}$. Let $P_2'$ and $P_2''$ be the segments forming $P_2$, $P_4'$ and $P_4''$ the segments forming $P_4$, and without loss of generality assume $b\in P_2'$ and $b\in P_4'$. Thus, $P_2'$ and $P_4'$ are either both vertical, or both diagonal, or one is vertical and the other is diagonal, or one is horizontal and the other is vertical/diagonal, and $P_1\cap_e P_2\subseteq P_1''$ and $P_1\cap_e P_4\subseteq P_1''$. Let us analyze each case separately, considering without loss of generality that $b$ is the leftmost point in $P_1'$. Note that, the cases in which $b$ is the rightmost point in $P_1'$ are isomorphic to the ones in which $b$ is the leftmost point up to two flip operations in a row, one horizontal and one vertical in any order.
\begin{itemize}
    \item[(i)] Suppose $P_2'$ and $P_4'$ are both vertical. Since $P_2\cap_e P_4=\emptyset$, $P_2'$ and $P_4'$ are in opposite directions. Without loss of generality, assume $b$ is the lowest point in $P_2'$. By Remark~\ref{1bend_int_1}, $P_1''$ and $P_2''$ must be on the same grid line $l$. Note that $l$ must be diagonal. Since $\mathcal{G}$ has only one diagonal direction, either $b$ is the bend point of both $P_1$ and $P_2$, or the bend point of $P_1$ is to the left of $b$. Then, in both cases $P_1$ has two grid edges with endpoint $b$ (contradiction!).
    \item[(ii)] Suppose $P_2'$ and $P_4'$ are both diagonal. Since $P_2\cap_e P_4=\emptyset$, $P_2'$ and $P_4'$ are in opposite directions. Without loss of generality, assume $b$ is the highest point in $P_2'$. By Remark~\ref{1bend_int_1}, $P_1''$ and $P_2''$ must be on the same grid line $l$. Note that $l$ must be vertical. Thus, either $b$ is the bend point of both $P_1$ and $P_2$, or the bend point of $P_1$ is to the left of $b$. Then, in both cases $P_1$ has two grid edges with endpoint $b$ (contradiction!).
    \item[(iii)] Suppose one of $P_2'$ and $P_4'$ is vertical, say $P_2'$, and the other is diagonal. By Remark~\ref{1bend_int_1}, $P_1''$ must be on the same grid line $l$ as both $P_2''$ and $P_4''$. This is a contradiction, because $P_2''$ and $P_4''$ are in different directions.
    \item[(iv)] If one of $P_2'$ and $P_4'$ is horizontal, by Remark~\ref{1bend_int_2}, $b$ must be the bend point of $P_1$. This is a contradiction, because $P_1$ contains only one grid edge with endpoint $b$.
\end{itemize}

Therefore, without loss of generality assume $P_1\cap_e P_2\subseteq P_1'$. Since $P_1$ contains only one grid edge with endpoint $b$, $P_2\cap_e P_4=\emptyset$ and $b\in P_4$, then $P_1\cap_e P_4\subseteq P_1''$. Note that, since $b\in P_1$, $b\in P_2$, $b\in P_4$, $P_1\cap_e P_2\subseteq P_1'$ and $P_1\cap_e P_4\subseteq P_1''$, then $U(\mathcal{P}_C\setminus\{P_3\})\cong T$, where $T$ is a right triangle on the grid. Note that $P_3\cap_e P_2\neq\emptyset$, $P_3\cap_e P_4\neq\emptyset$, $P_3\cap_e P_1=\emptyset$ and $b\in P_3$. Note also that, if $P_3\cap_e P_4\subset P_4''$, then $P_3\cap_e P_4'\neq\emptyset$, since $b\in P_3$ and $P_3\cap_e P_1=\emptyset$. Therefore, $P_3\cap_e P_4\subseteq P_4'$. Thus, if $P_3\cap_e P_2\subset P_2'$, $\mathcal{P}_C$ is a flag and if $P_3\cap_e P_2\subseteq P_2''$, $\mathcal{P}_C$ is either a flag or a butterfly. See Figure~\ref{fig: thm cycles}.
\begin{figure}[htb]
    \centering
    \subfigure[][]{\includegraphics[scale=0.5]{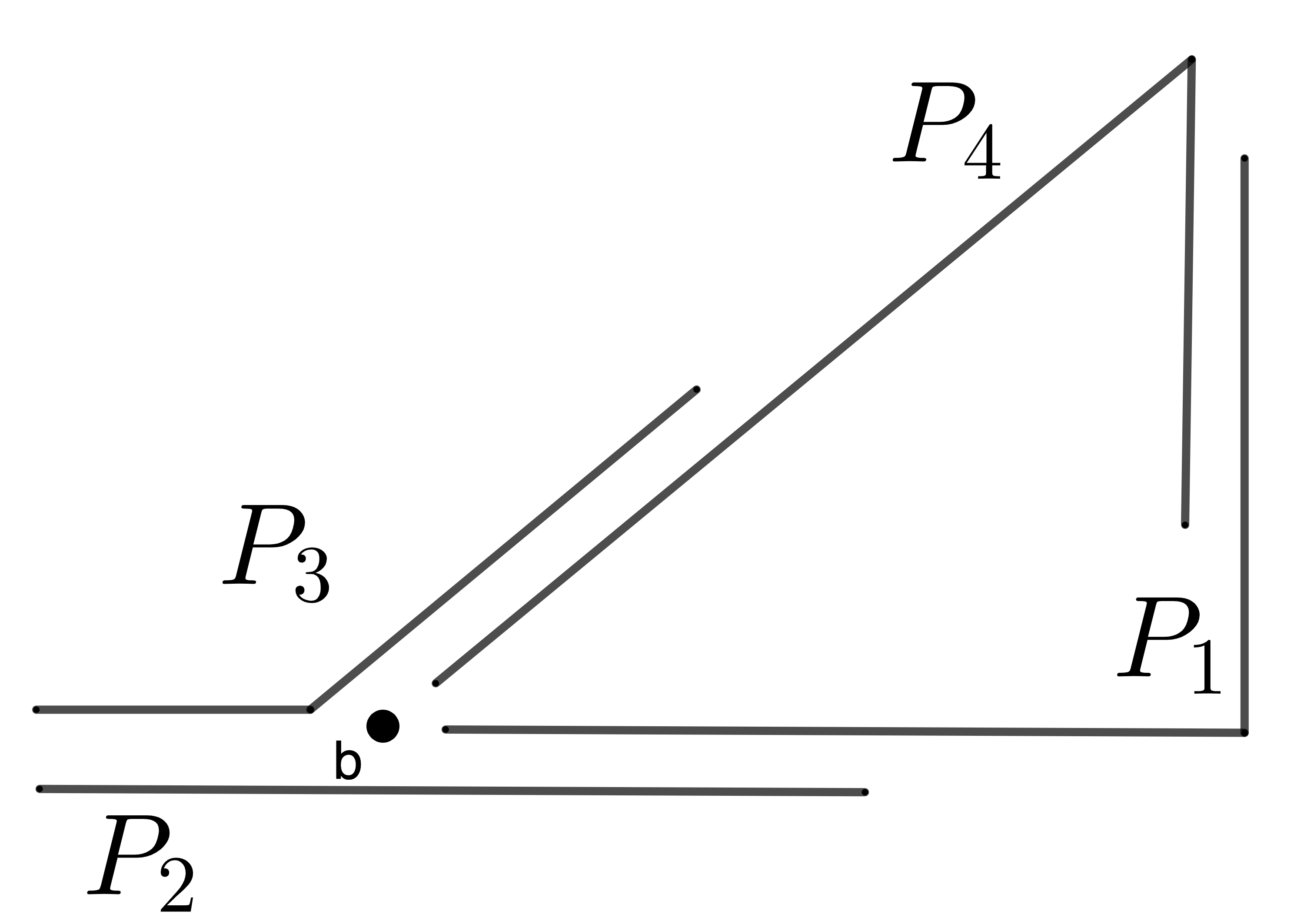}}
    \qquad
    \subfigure[][]{\includegraphics[scale=0.5]{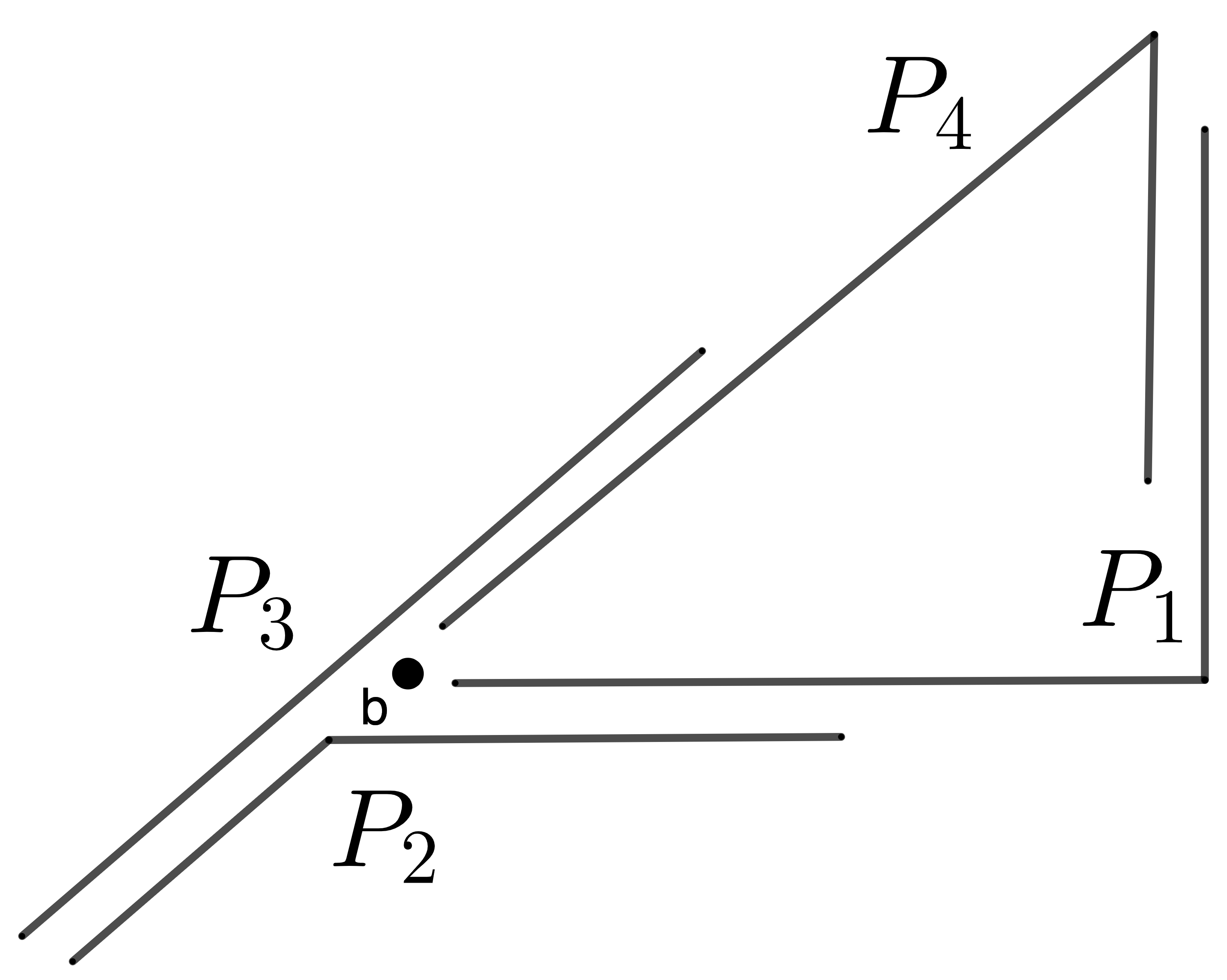}}
    \qquad
    \subfigure[][]{\includegraphics[scale=0.5]{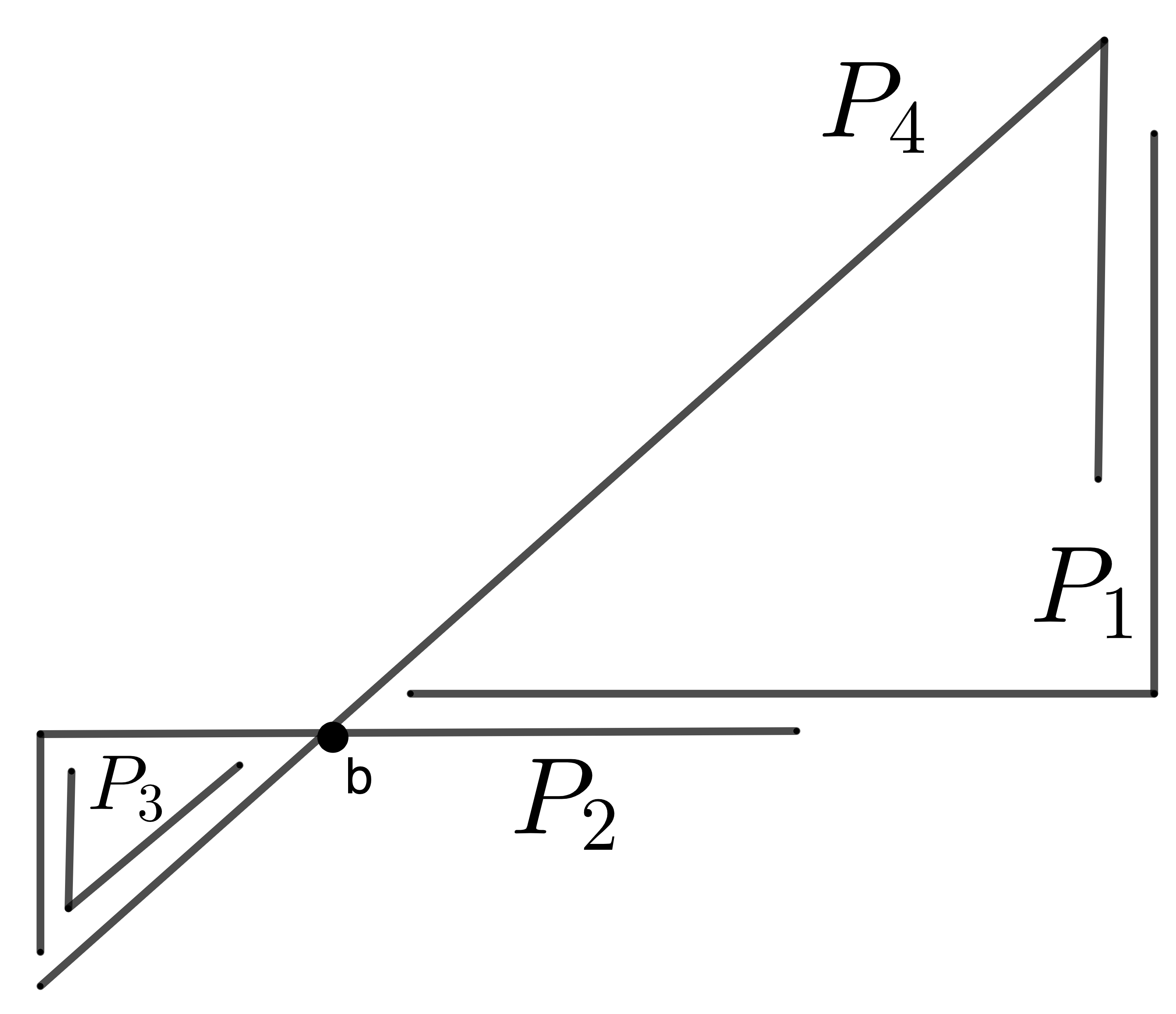}}
    \qquad
    \subfigure[][]{\includegraphics[scale=0.5]{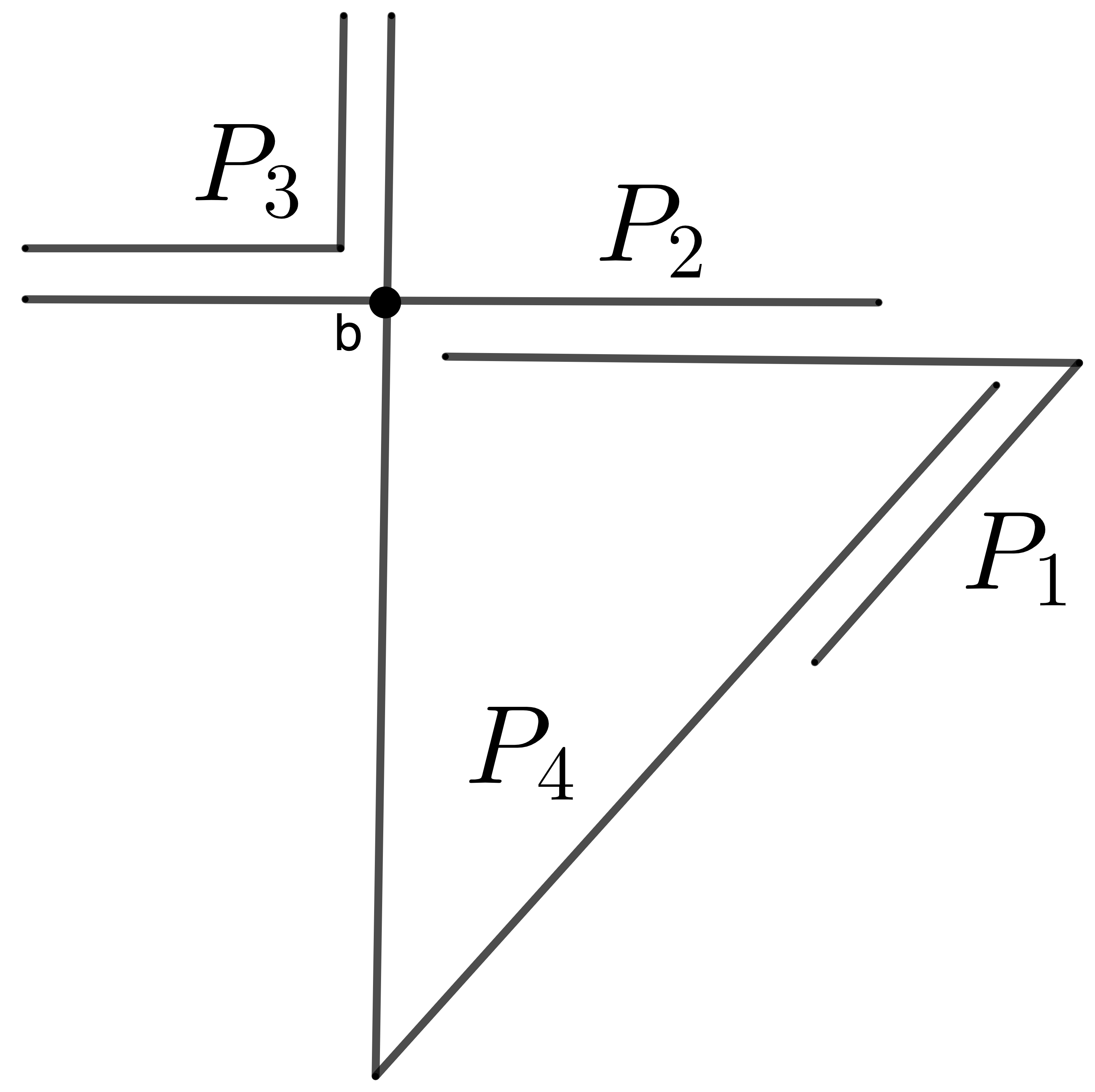}}
    \qquad
    \subfigure[][]{\includegraphics[scale=0.5]{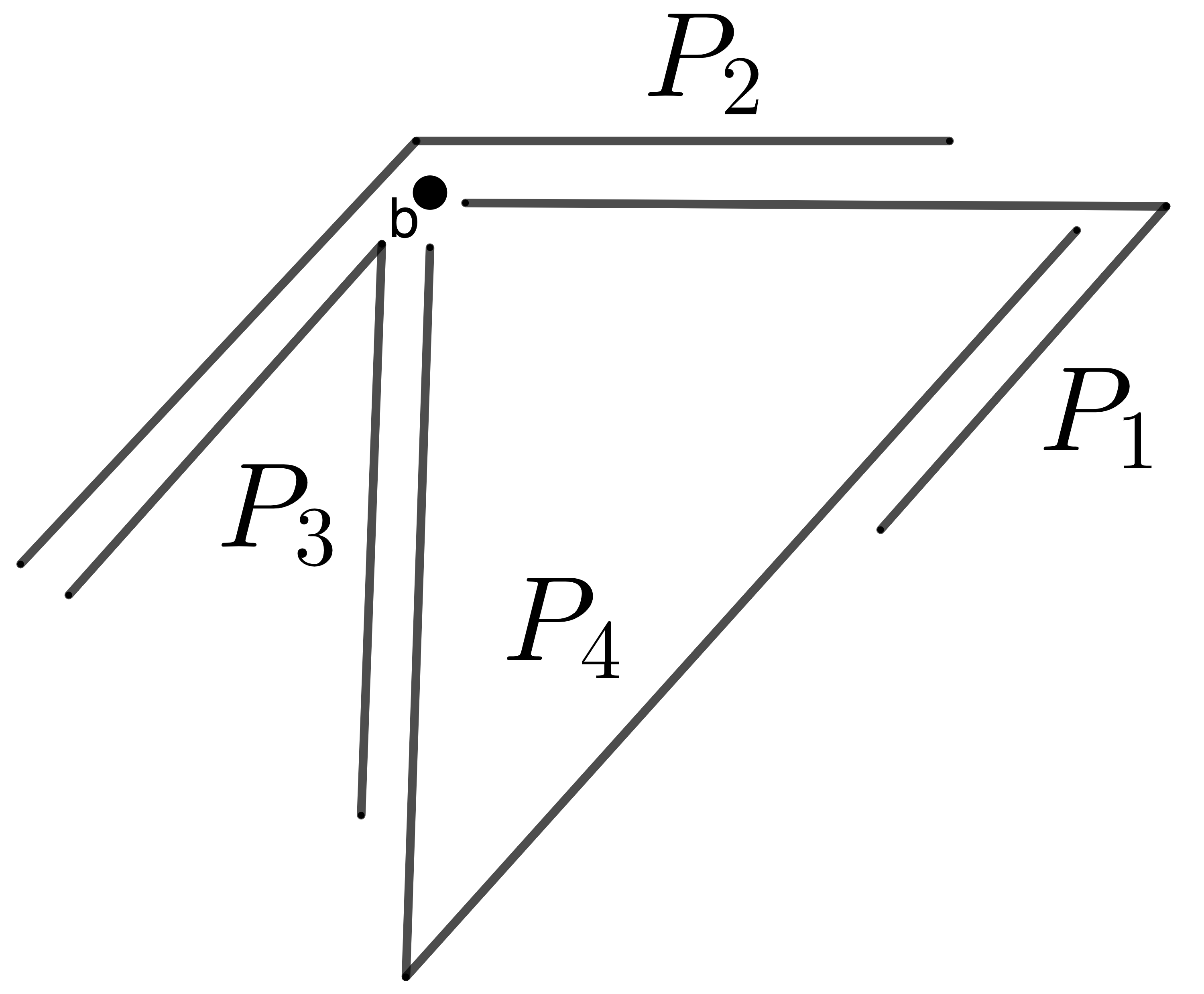}}
    \qquad
    \subfigure[][]{\includegraphics[scale=0.5]{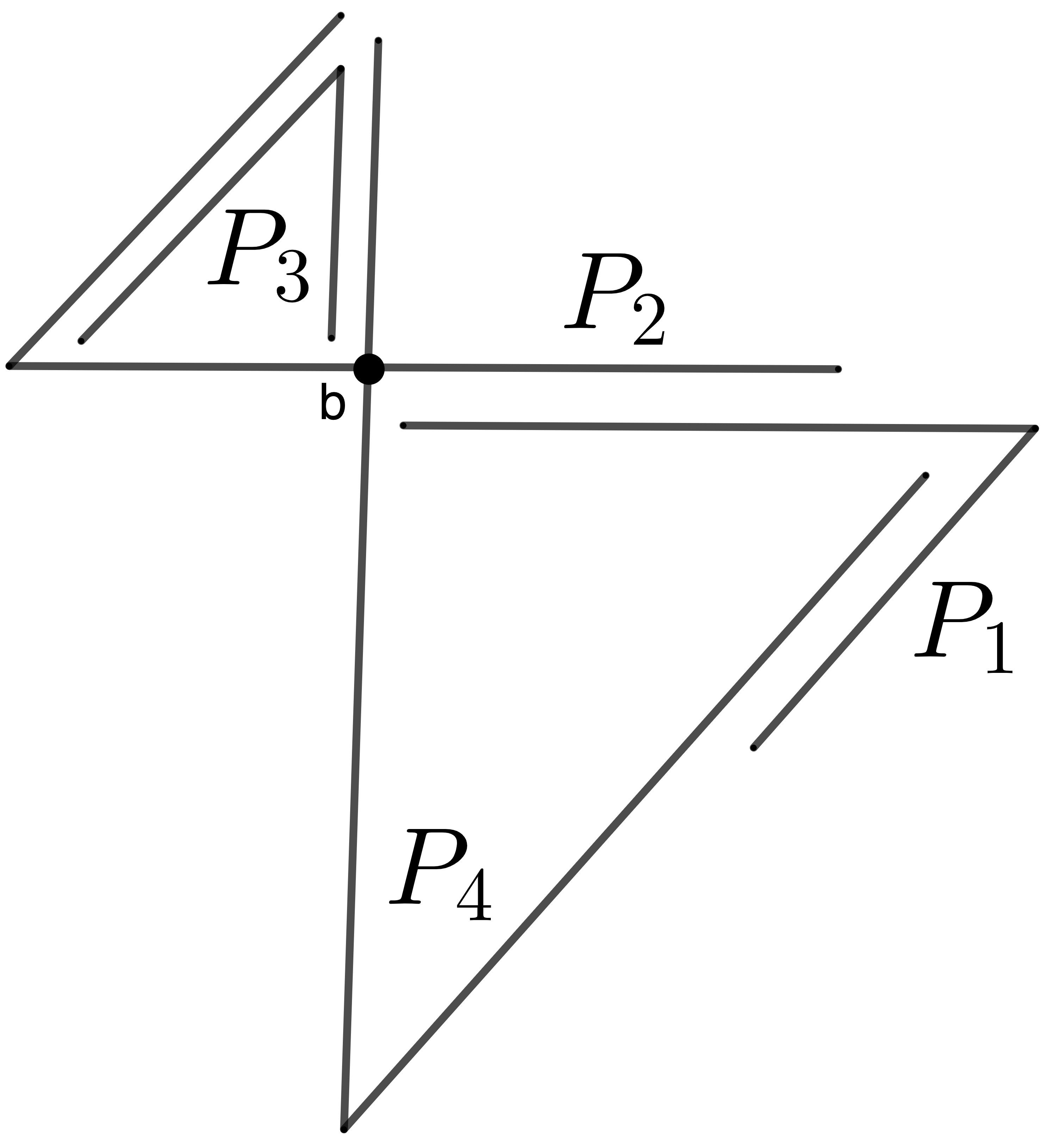}}
    \qquad
    \subfigure[][]{\includegraphics[scale=0.5]{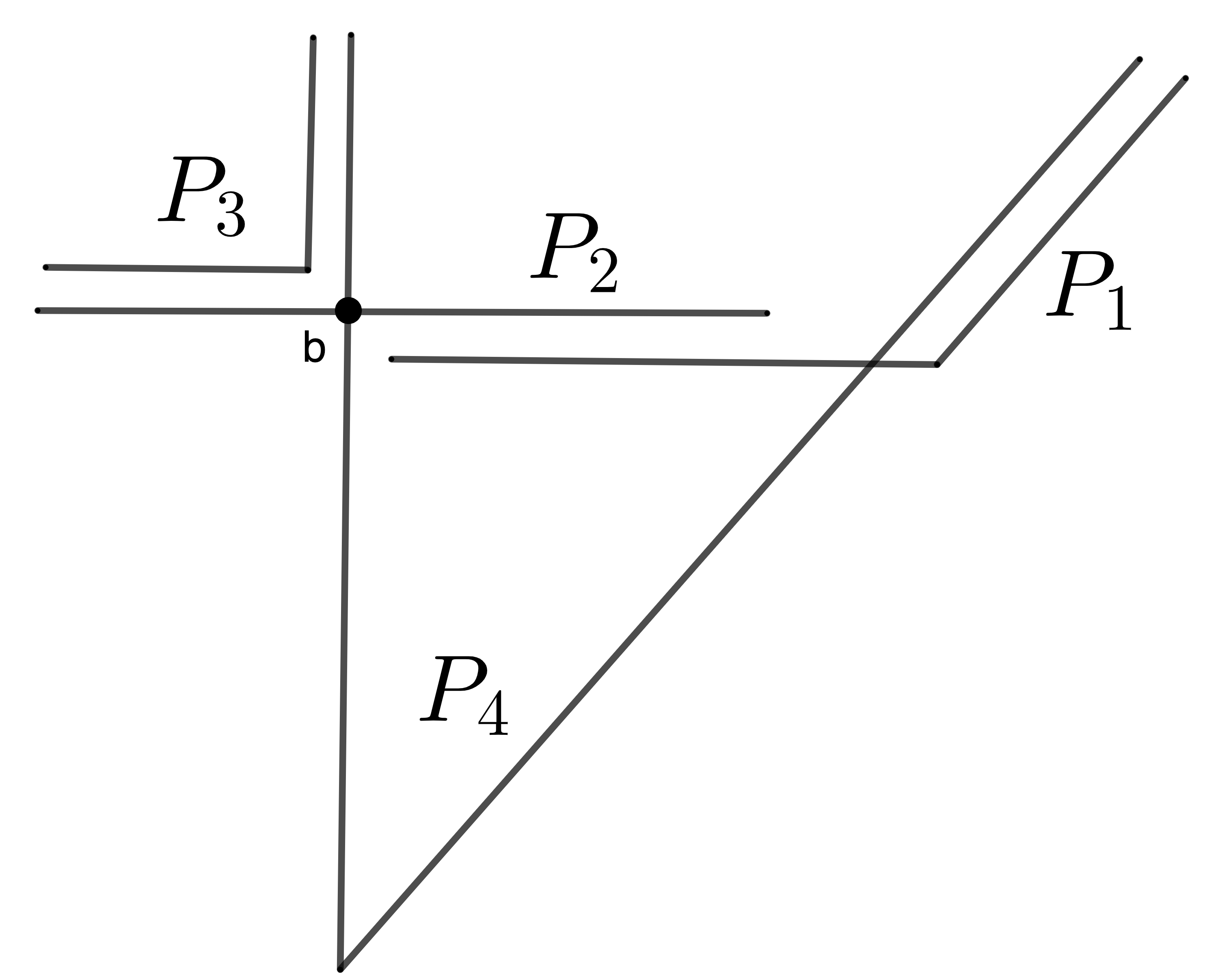}}
    \qquad
    \subfigure[][]{\includegraphics[scale=0.5]{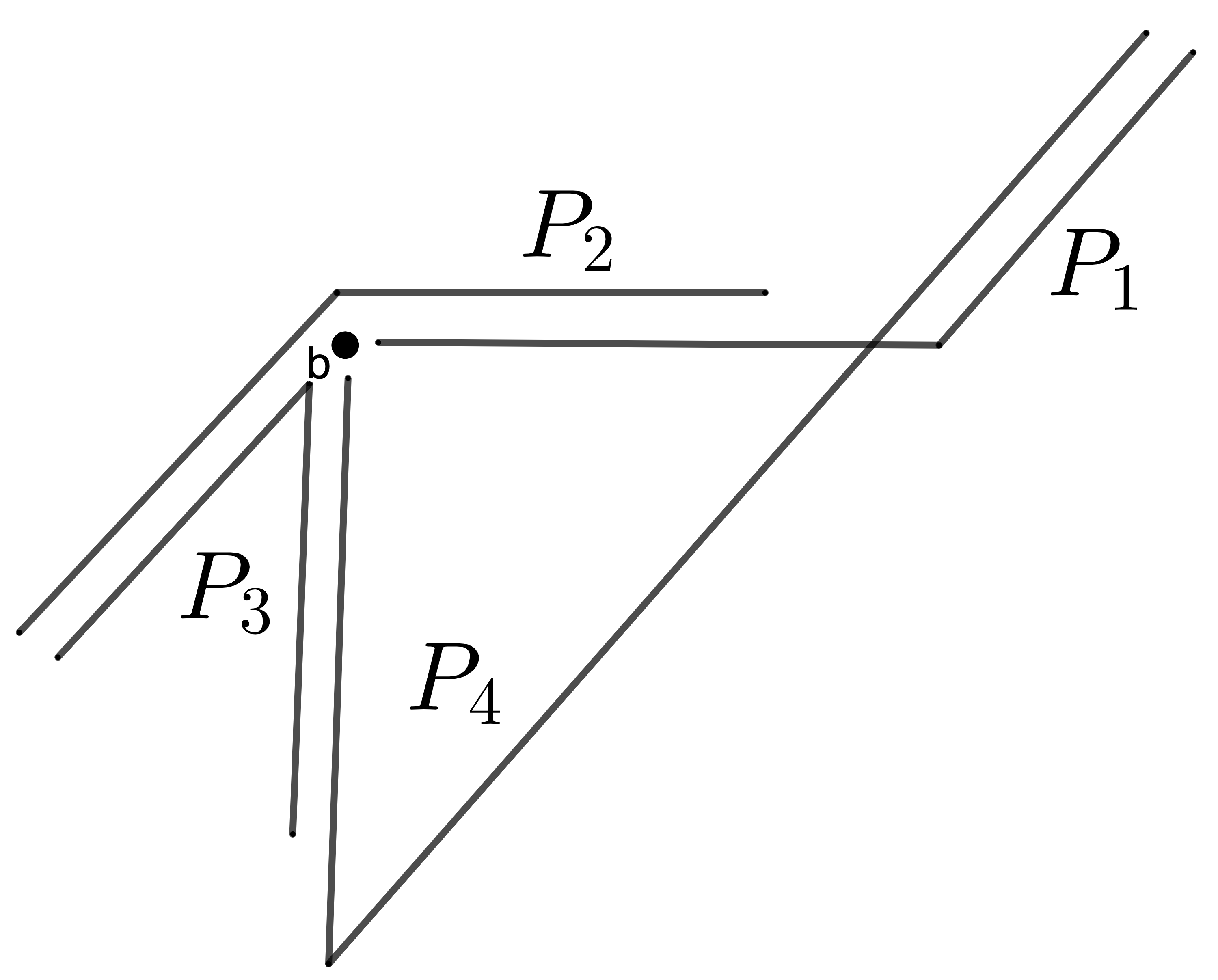}}
    \qquad
    \subfigure[][]{\includegraphics[scale=0.5]{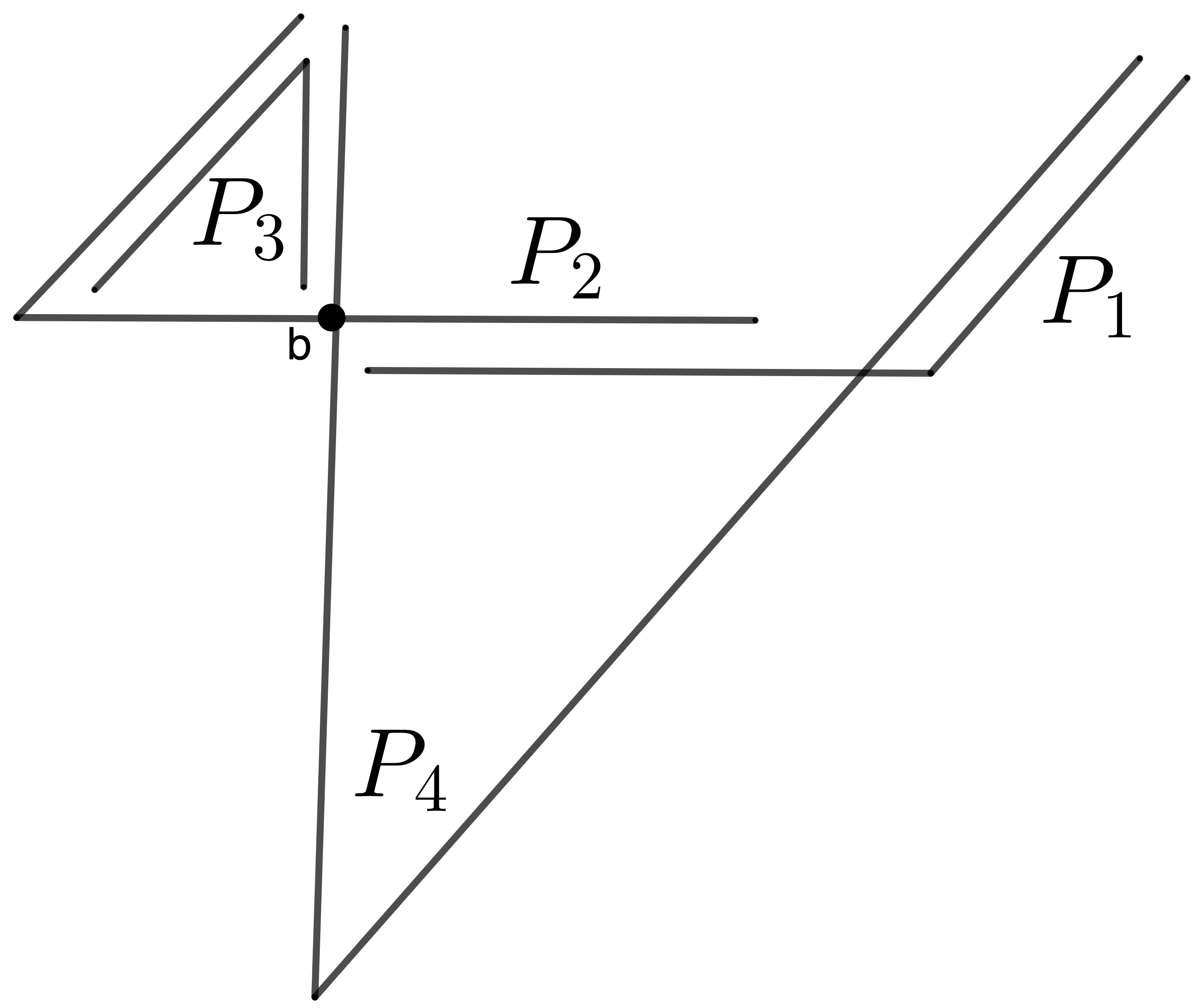}}
    \caption{Some possible configurations of a $C_4$ when $\bigcap_v P_i=\{b\}$ and at least one path contains only one grid edge with endpoint $b$.}
    \label{fig: thm cycles}
\end{figure}

Otherwise, $\bigcap_v P_i=\emptyset$. Assume there is at least one path with no bends and without loss of generality let $P_1$ be such a path. Note that, $P_1\cap_e P_2\neq\emptyset$, $P_1\cap_e P_4\neq\emptyset$ and $P_2\cap_e P_4=\emptyset$. If neither $P_2$ nor $P_4$ have a bend, then we obtain an interval representation of $C_4\setminus\{v_3\}$. However, in this case we cannot add the path $P_3$ with at most one bend. If both $P_2$ and $P_4$ have a single bend, consider the following cases: If $P_2$ and $P_4$ bend in the same direction, by Remark~\ref{1bend_int_3} the path $P_3$ cannot be added with one bend. Otherwise, let $b_2=(x_1, y_1)$ be the bend point of $P_2$, $b_4=(x_2, y_1)$ the bend point of $P_4$, $l_2$ the grid line containing a segment of $P_2$ such that $l_2\cap_v P_1=\{b_2\}$, $l_4$ the grid line containing a segment of $P_4$ such that $l_4\cap_v P_1=\{b_4\}$ and $l_2\cap_v l_4=\{b_3\}$. Then, $b_3=(x_i, y_2)$, for some $i\in\{1,2\}$, and $P_3$ can be a path that bends at $b_3$, in which case we obtain a flag.

Now, assume every $P_1,\ldots,P_4$ has a single bend and let $P_i'$ and $P_i''$ be the segments of $P_i$ and $b_i$ the bend point of $P_i$, for all $1\leq i\leq4$. Without loss of generality, we can assume that $P_i'\supseteq P_i\cap_e P_j$ and $P_i''\supseteq P_i\cap_e P_k$ for some $1\leq j,k\leq4$, such that $i\neq j$ and $i\neq k$. That is, a path $P_i$ edge intersect its neighbors on both segments. If that was not the case, $P_i$ could be treated as a path with no bends, since only one of its segments would be in fact relevant to the analysis, but this scenario was already considered previously. Consider at least two paths, $P_i$ and $P_j$, of $\mathcal{P}_C$ bend at the same point, that is, $b_i=b_j=b$.
\begin{itemize}
    \item If $P_i\cap_e P_j=\emptyset$, without loss of generality assume $P_i=P_1$, $P_j=P_3$, $P_1'\supseteq P_1\cap_e P_2$, $P_1''\supseteq P_1\cap_e P_4$, $P_3'\supseteq P_3\cap_e P_2$, $P_3''\supseteq P_3\cap_e P_4$. Thus, $b\in P_2$ and $b\in P_4$ (contradiction!).
    \item If $P_i\cap_e P_j\neq\emptyset$, without loss of generality assume $P_i=P_1$, $P_j=P_2$, $P_1'\supseteq P_1\cap_e P_2$, $P_2'\supseteq P_1\cap_e P_2$. Note that it is not possible that $P_1\subseteq P_2$ or $P_2\subseteq P_1$, since $P_1\cap_e P_4\neq\emptyset$, $P_2\cap_e P_4=\emptyset$, $P_2\cap_e P_3\neq\emptyset$ and $P_1\cap_e P_3=\emptyset$.
    
    Suppose $P_1''$ and $P_2''$ are on the same grid line and in the same direction. Thus, all segments of $P_1$ and $P_2$ are on the same grid lines. Note that $P_1\cap_e P_2$ is also contained in $P_1''$ and $P_2''$. Let $l_1$ be the grid line containing $P_1'$ and $P_2'$, and $l_2$ the grid line containing $P_1''$ and $P_2''$. Since $P_1\nsubseteq P_2$ and $P_2\nsubseteq P_1$, assume without loss of generality that $P_1'\supset P_2'$ and $P_2''\supset P_1''$. Thus, $P_1\cap_e P_4\subset P_1'$ and $P_2\cap_e P_3\subset P_2''$. Without loss of generality, assume $P_4''\supseteq P_1\cap_e P_4$ and $P_3''\supseteq P_2\cap_e P_3$. Note that $P_3''$ and $P_4''$ are in $l_2$ and $l_1$, respectively. Note also that, since $P_3'\supseteq P_3\cap_e P_4$ and $P_4'\supseteq P_3\cap_e P_4$, $P_3'$ and $P_4'$ must be on the same grid line, $l_3$, such that $l_3$ is not in the direction of $l_1$, since by Remark~\ref{1bend_int_3} this would imply that $P_4$ has at least two bends, and $l_2$, since by Remark~\ref{1bend_int_3} this would imply that $P_3$ has at least two bends. Thus, $l_3\cap_v l_2=\{b_3\}$ and $l_3\cap_v l_1=\{b_4\}$. This could imply that $P_1\cap_e P_3\neq\emptyset$ or $P_2\cap_e P_4\neq\emptyset$ (depending on the types and shapes of the paths), which would be a contradiction. Otherwise, $U(\mathcal{P}_C)$ contains a right triangle with corners $b=b_1=b_2$, $b_3$ and $b_4$, as a subgraph. Thus, $\mathcal{P}_C$ is a flag.
    
    Suppose $P_1''$ and $P_2''$ are on the same grid line and in opposite directions. Without loss of generality, assume $P_1''\supseteq P_1\cap_e P_4$, $P_2''\supseteq P_2\cap_e P_3$, $P_4''\supseteq P_1\cap_e P_4$ and $P_3''\supseteq P_2\cap_e P_3$. Thus, $P_1''$, $P_2''$, $P_3''$ and $P_4''$ are all on the same grid line. Moreover, $P_4'\supseteq P_3\cap_e P_4$ and $P_3'\supseteq P_3\cap_e P_4$. Therefore, $P_3'$ and $P_4'$ are on the same grid line and $b_3=b_4=b'$. Note that, if $b'\neq b$, either $P_1\cap_e P_3\neq\emptyset$ or $P_2\cap_e P_4\neq\emptyset$, which is a contradiction. But $b'=b$ is also a contradiction, since $\bigcap_v P_i=\emptyset$. Thus, $P_1''$ and $P_2''$ cannot be on the same grid line in opposite directions.
    
    Suppose $P_1''$ and $P_2''$ are not on the same grid line. Let $l_1$ be the grid line containing $P_1'$ and $P_2'$, and $l_2$ the grid line containing $P_1''$ and $l_3$ the grid line containing $P_2''$. Thus, $P_1''\supseteq P_1\cap_e P_4$ and $P_2''\supseteq P_2\cap_e P_3$. Without loss of generality, assume $P_4''\supseteq P_1\cap_e P_4$ and $P_3''\supseteq P_2\cap_e P_3$. Note that $P_3''$ and $P_4''$ are in $l_3$ and $l_2$, respectively. Note also that, since $P_3'\supseteq P_3\cap_e P_4$ and $P_4'\supseteq P_3\cap_e P_4$, $P_3'$ and $P_4'$ must be on the same grid line, $l_4$, such that $l_4$ is not in the direction of $l_2$, since by Remark~\ref{1bend_int_3} this would imply that $P_4$ has at least two bends, and $l_3$, since by Remark~\ref{1bend_int_3} this would imply that $P_3$ has at least two bends. Therefore, $l_4$ is in the same direction as $l_1$, but they are not coincident, otherwise $\bigcap_v P_i\neq\emptyset$ (contradiction!). Thus, $l_4\cap_v l_3=\{b_3\}$ and $l_4\cap_v l_2=\{b_4\}$. Thus, $U(\mathcal{P}_C)$ contains a right triangle with corners $b=b_1=b_2$, $b_3$ and $b_4$, as a subgraph. Thus, $\mathcal{P}_C$ is a flag.
\end{itemize}

Otherwise, all paths of $\mathcal{P}_C$ bend at different points. Let $l_1$ be the grid line containing $P_1\cap_e P_2$, $l_2$ the grid line containing $P_1\cap_e P_4$, $l_3$ the grid line containing $P_2\cap_e P_3$ and $l_4$ the grid line containing $P_3\cap_e P_4$. Since there are only three directions on the grid, by the pigeonhole principle two of the grid lines $l_1$, $l_2$, $l_3$ and $l_4$ must be parallel. Note that $l_2$ and $l_3$ cannot be parallel to $l_1$ or $l_4$. Without loss of generality, assume $l_2$ and $l_3$ are parallel. If $l_4$ and $l_1$ are also parallel, then $b_1$, $b_2$, $b_3$ and $b_4$ are the corners of a parallelogram. Thus, $\mathcal{P}_C$ is either an r-frame or a p-frame. If $l_4$ and $l_1$ are not parallel, let $l_4\cap_v l_1=\{b'\}$. If $U(\mathcal{P}_C)\nsupseteq\{b'\}$, then $b_1$, $b_2$, $b_3$ and $b_4$ are the corners of a trapezoid, and $\mathcal{P}_C$ is a t-frame. If $U(\mathcal{P}_C)\supseteq\{b'\}$, then $b_1$, $b_4$, $b'$ are the corners of a right triangle and $b_2$, $b_3$, $b'$ are the corners of another $F_{k'}$, and $\mathcal{P}_C$ is a butterfly.

This concludes the proof of the theorem.
\end{proof}

\section{Initial Results}
Clearly, every B$_k$-EPG representation is a B$_k$-EPG\textsubscript{t} representation. We next show that the converse does not hold.

The sun graphs $S_k$ are not B$_1$-EPG for $k\geq4$~\cite{golumbic2013single}. However, $S_4$ is B$_1$-EPG\textsubscript{t} (see Figure~\ref{fig: S4}). It is possible to extend the B$_1$-EPG\textsubscript{t} representation depicted in Figure~\ref{fig: S4} to show that every sun graph is B$_1$-EPG\textsubscript{t}, as we show next.
\begin{figure}[htb]

\center
\subfigure[][]{\includegraphics[scale=0.7]{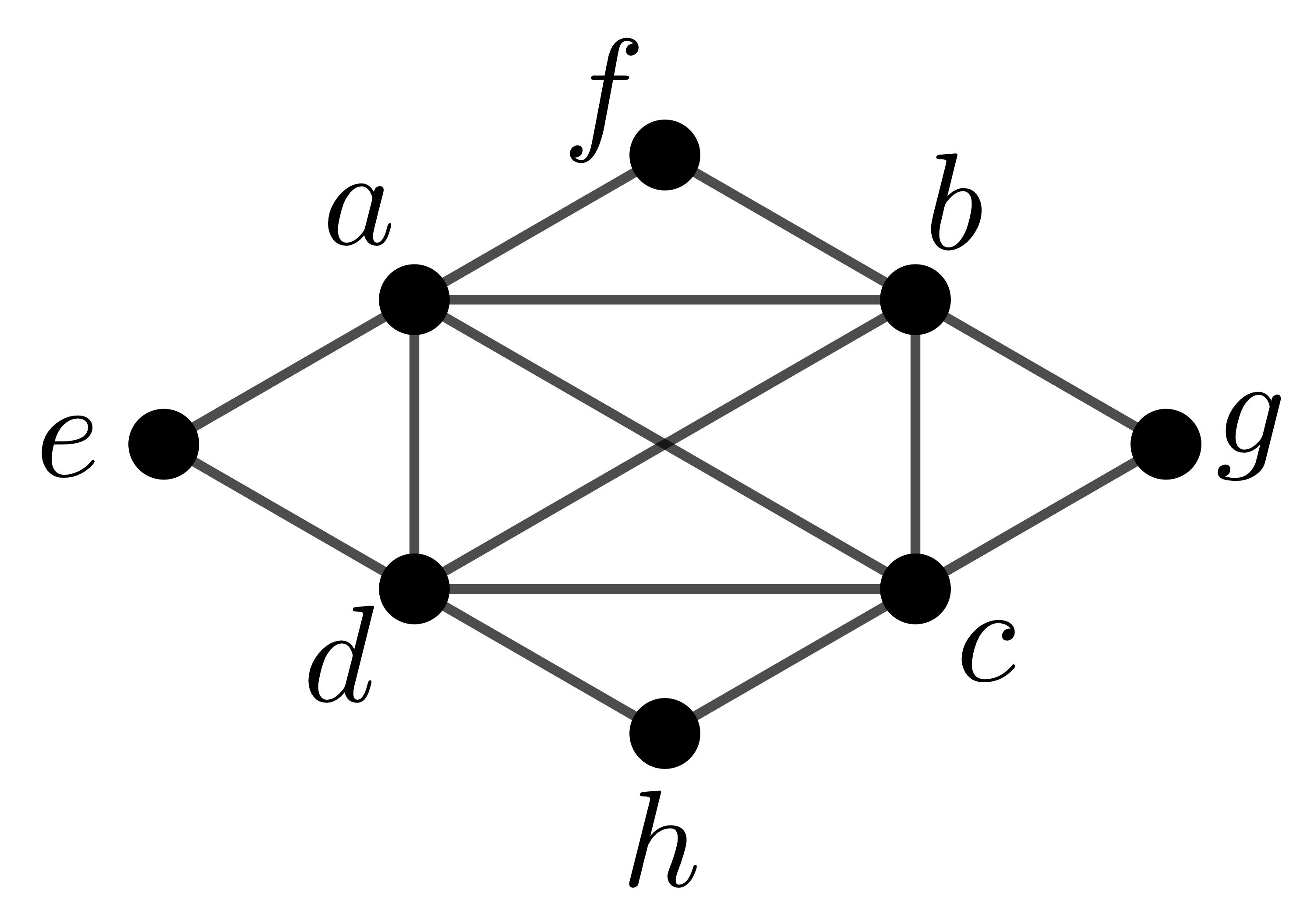}}
\qquad
\subfigure[][]{\includegraphics[scale=0.5]{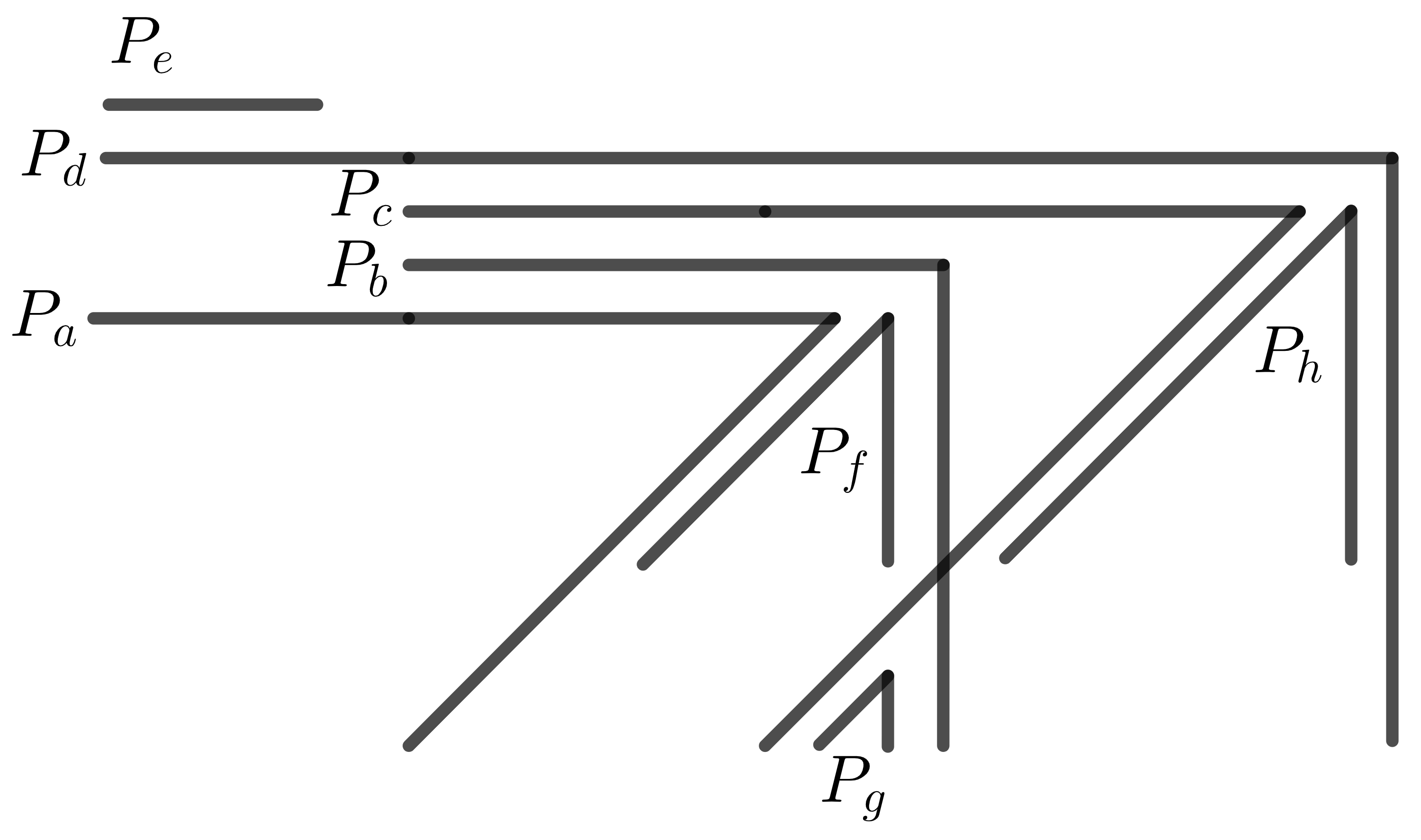}}
\caption{The $S_4$ and its B$_1$-EPG\textsubscript{t} representation.}
\label{fig: S4}
\end{figure}
\begin{thm}
Every $k$-sun graph $S_k$, $k\geq4$, has a B$_1$-EPG\textsubscript{t} representation on a $3\times\left\lceil\frac{k+4}{2}\right\rceil$ triangular grid.
\end{thm}
\begin{proof}
Let $V(S_k)=K\cup S$, where $K=\{v_1,\ldots,v_k\}$ is the clique and $S=\{s_1,\ldots,s_k\}$ is the independent set, such that $(s_i,v_i)\in E(S_k)$ and $(s_i,v_{i+1})\in E(S_k)$, for all $1\leq i\leq k$, where addition is assumed to be modulo $k$. Let $\{1,2,3\}$ be a set of consecutive rows and $\{1,\ldots, \lceil\frac{k+4}{2}\rceil\}$ be a set of consecutive columns of a triangular grid, where row $1$ is the lowest one and column $1$ is the leftmost one. A grid point labeled $(x,y)$ represents the intersection point of the column $x$ and the row $y$ for some $x\in\{1,\ldots, \lceil\frac{k+4}{2}\rceil\}$ and $y\in\{1,2,3\}$.

Let $v_1$ be represented by the $1$-bend path $P_1^v$ consisting of $(1,3)$ and $(2,2)$ as extreme points, and having $(3,3)$ as a bending point. For all $v_i$ such that $1<i\leq k$ and $i$ is odd, build a $1$-bend path $P_i^v$ consisting of $(2,3)$ and $(\frac{i+1}{2},1)$ as extreme points, and having $(\frac{i+5}{2},3)$ as a bending point. For all $v_i$ such that $1<i\leq k$ and $i$ is even, build a path $P_i^v$ consisting of $(2,3)$ and $(\frac{i+4}{2},1)$ as extreme points, and having $(\frac{i+4}{2},3)$ as a bending point. Finally, enlarge $P_k^v$ by stretching it horizontally from $(2,3)$ to $(1,3)$.

Let $s_k$ be represented by the path $P_k^s$ consisting of $(1,3)$ and $(2,3)$ as extreme points. For all $s_i$ such that $1\leq i<k$ and $i$ is odd, build a path $P_i^s$ consisting of $(\frac{i+3}{2},2)$ and $(\frac{i+5}{2},2)$ as extreme points, and having $(\frac{i+5}{2},3)$ as a bending point. For all $v_i$ such that $1\leq i<k$ and $i$ is even, build a path $P_i^s$ consisting of $(\frac{i+2}{2},1)$ and $(\frac{i+4}{2},1)$ as extreme points, and having $(\frac{i+4}{2},2)$ as a bending point.

We claim that $\mathcal{R}=\{P_1^v,\ldots,P_k^v,P_1^s,\ldots,P_k^s\}$ is a B$_1$-EPG\textsubscript{t} representation of $S_k$. Note that
\begin{itemize}
    \item[-] all $P_i^v$, $1\leq i\leq k$, share the grid edge $((2,3),(3,3))$;
    \item[-] only $P_1^v$, $P_k^v$ and $P_k^s$ share the grid edge $((1,3),(2,3))$;
    \item[-] the grid edge $((\frac{i+3}{2},2),(\frac{i+5}{2},3))$ is shared only by $P_i^s$ and $P_i^v$ and the grid edge $((\frac{i+5}{2},3),(\frac{i+5}{2},2))$ is shared only by $P_i^s$ and $P_{i+1}^v$, for all $1\leq i<k$ such that $i$ is odd;
    \item[-] the grid edge $((\frac{i+4}{2},2),(\frac{i+4}{2},1))$ is shared only by $P_i^s$ and $P_i^v$ and the grid edge $((\frac{i+2}{2},1),(\frac{i+4}{2},2))$ is shared only by $P_i^s$ and $P_{i+1}^v$, for all $1\leq i<k$ such that $i$ is even.
\end{itemize}
Thus, $\mathcal{R}$ is a B$_1$-EPG\textsubscript{t} representation of $S_k$, $k\geq4$.
\end{proof}
Other examples of graphs that are not B$_1$-EPG but have a B$_1$-EPG\textsubscript{t} representation are shown in Figure~\ref{fig: examples}. The graphs in Figures~\ref{fig: F1}, \ref{fig: F2} and \ref{fig: F6} are forbidden induced subgraphs for the class B$_1$-EPG, as shown in~\cite{alconrelationship}. Nonetheless, they are B$_1$-EPG\textsubscript{t}, as shown in Figures~\ref{fig: F1_model}, \ref{fig: F2_model} and \ref{fig: F6_model}. The graph given in Figure~\ref{fig: outer} has rectangular bend-number two, as shown in~\cite{biedl2010edge}. However, we were able to find a B$_1$-EPG\textsubscript{t} representation of it (see Figure~\ref{fig: outer_model}).
\begin{figure}[htb]

\center
\subfigure[][]{\includegraphics[scale=0.5]{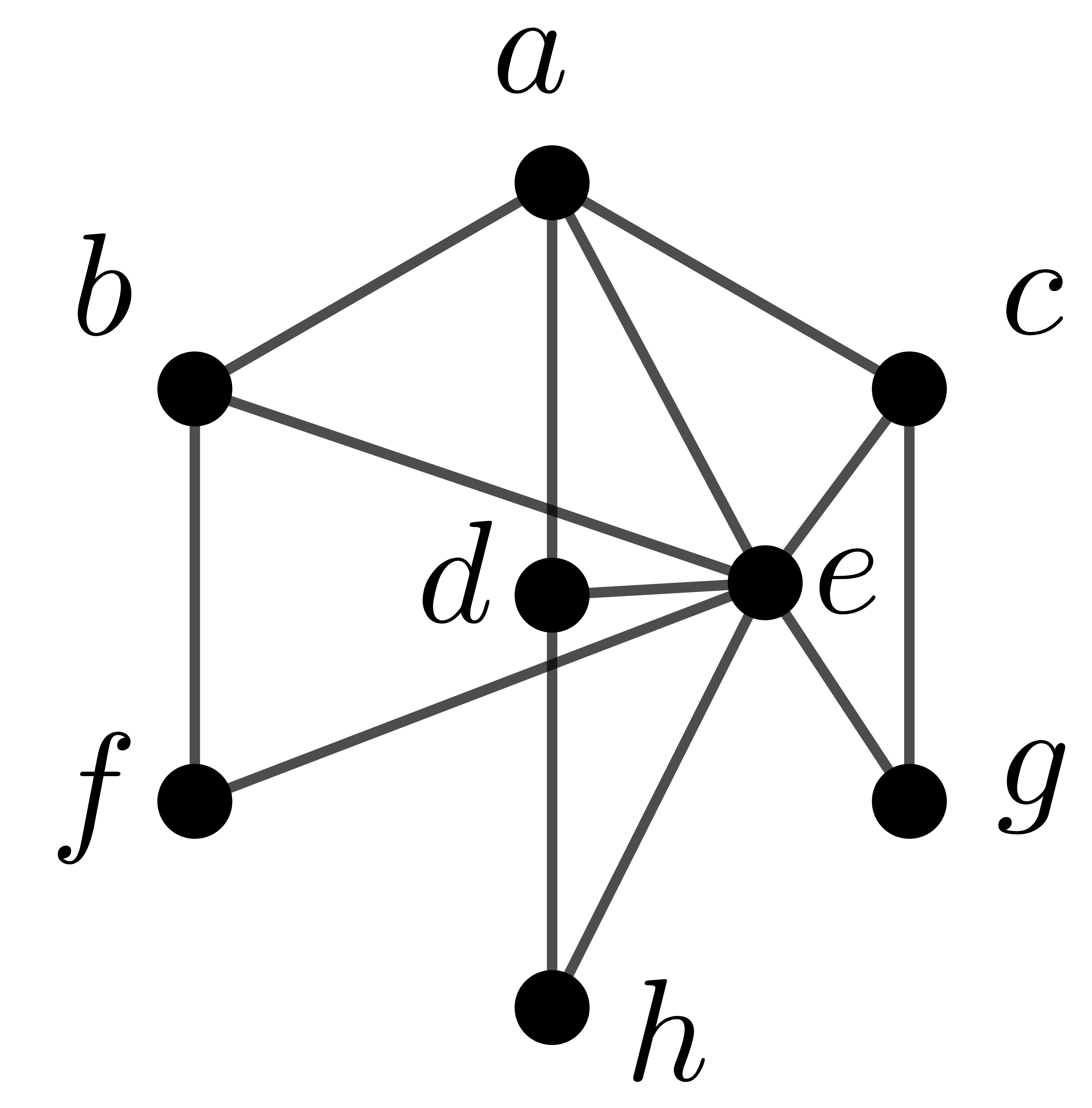} \label{fig: F1}}
\qquad
\subfigure[][]{\includegraphics[scale=0.5]{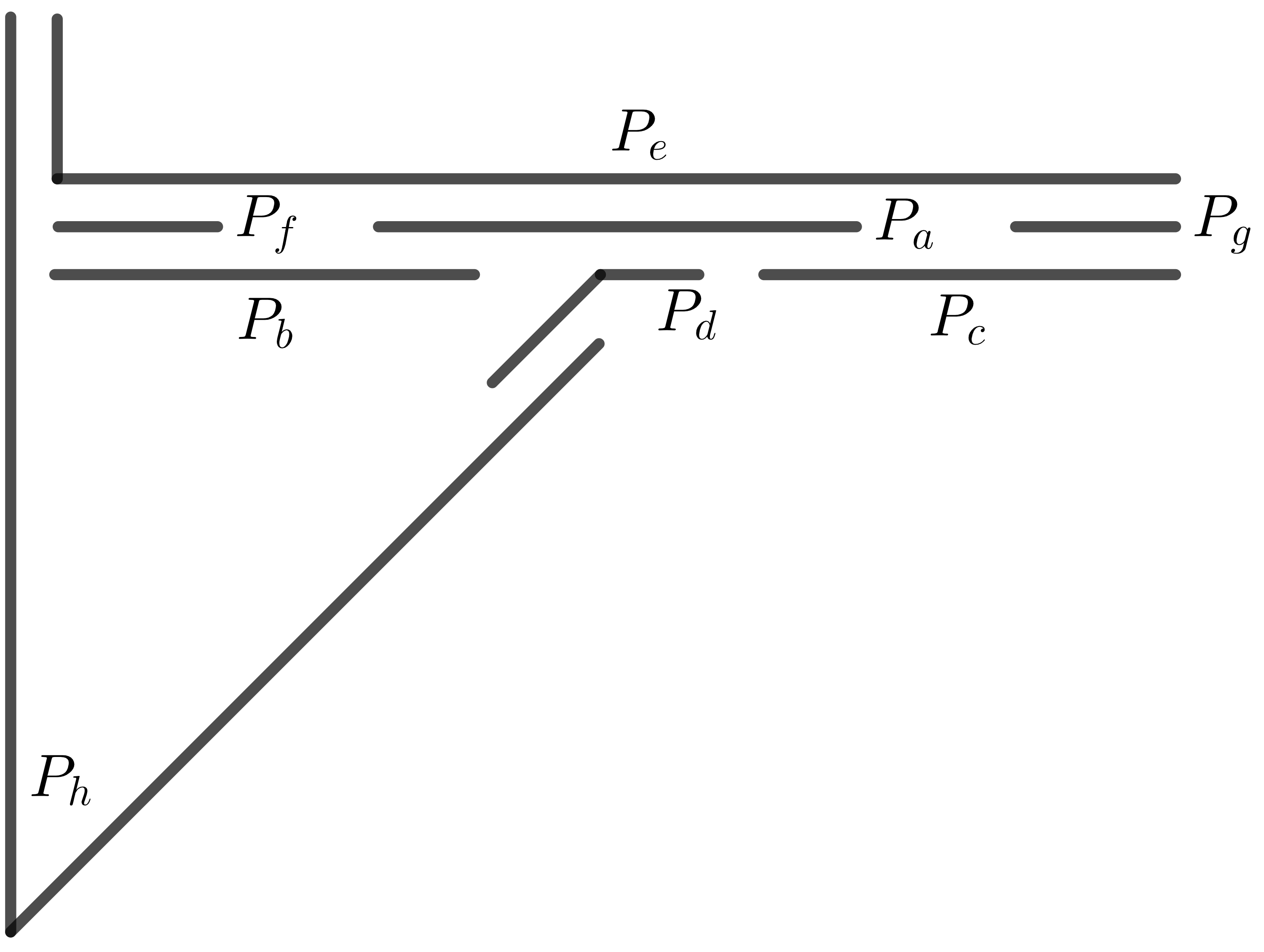} \label{fig: F1_model}}
\qquad
\subfigure[][]{\includegraphics[scale=0.4]{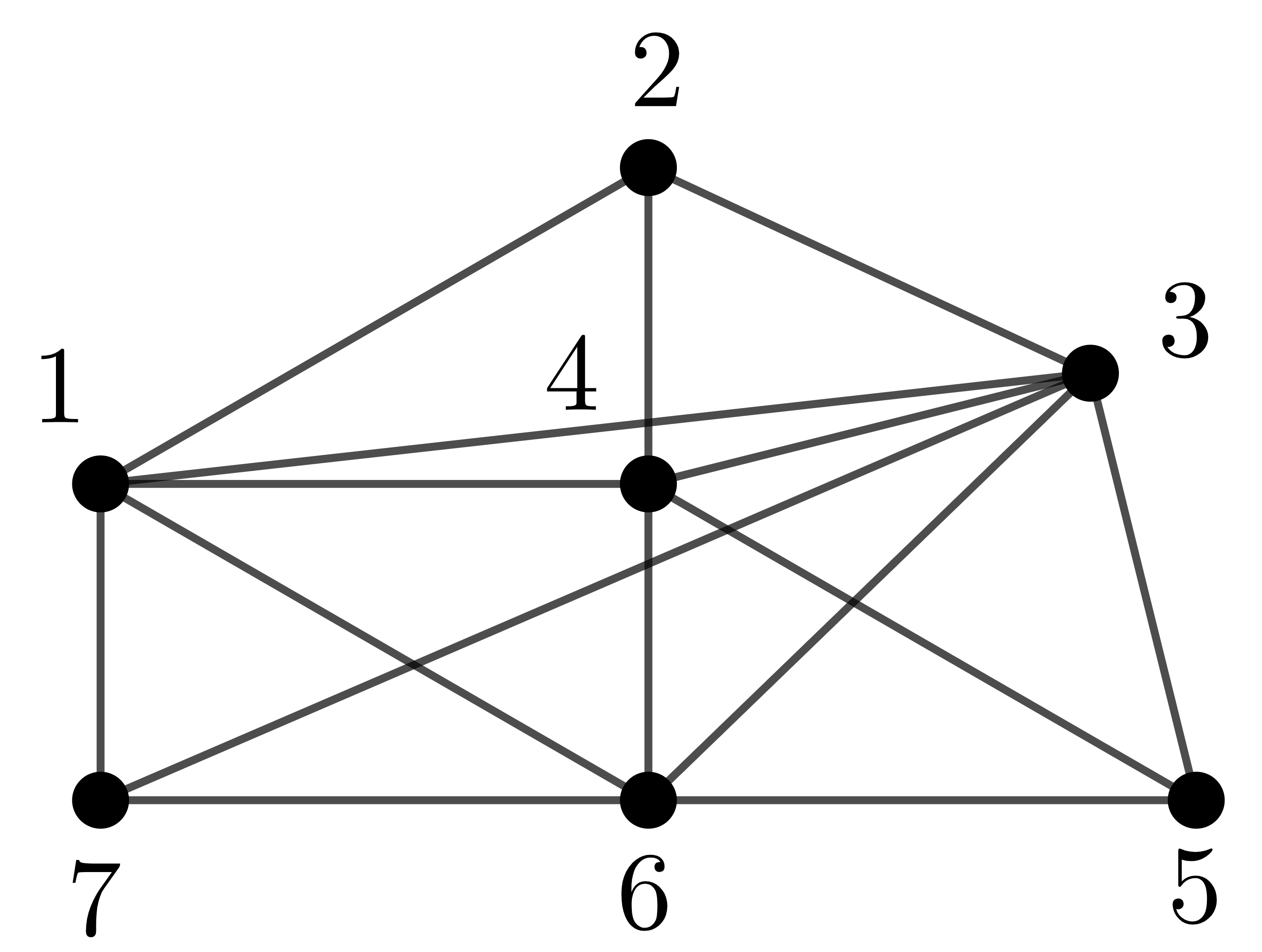} \label{fig: F2}}
\qquad
\subfigure[][]{\includegraphics[scale=0.5]{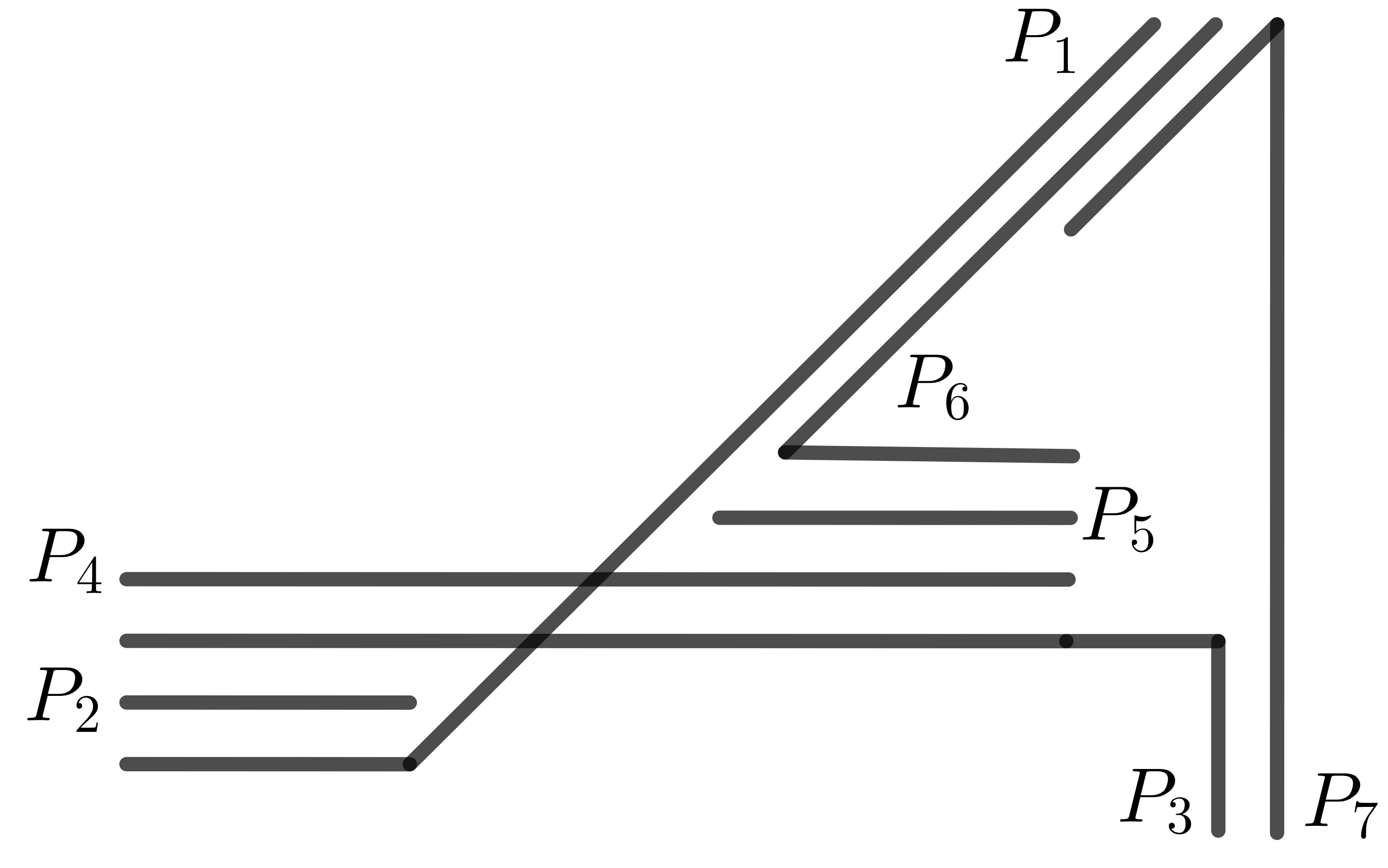} \label{fig: F2_model}}
\qquad
\subfigure[][]{\includegraphics[scale=0.5]{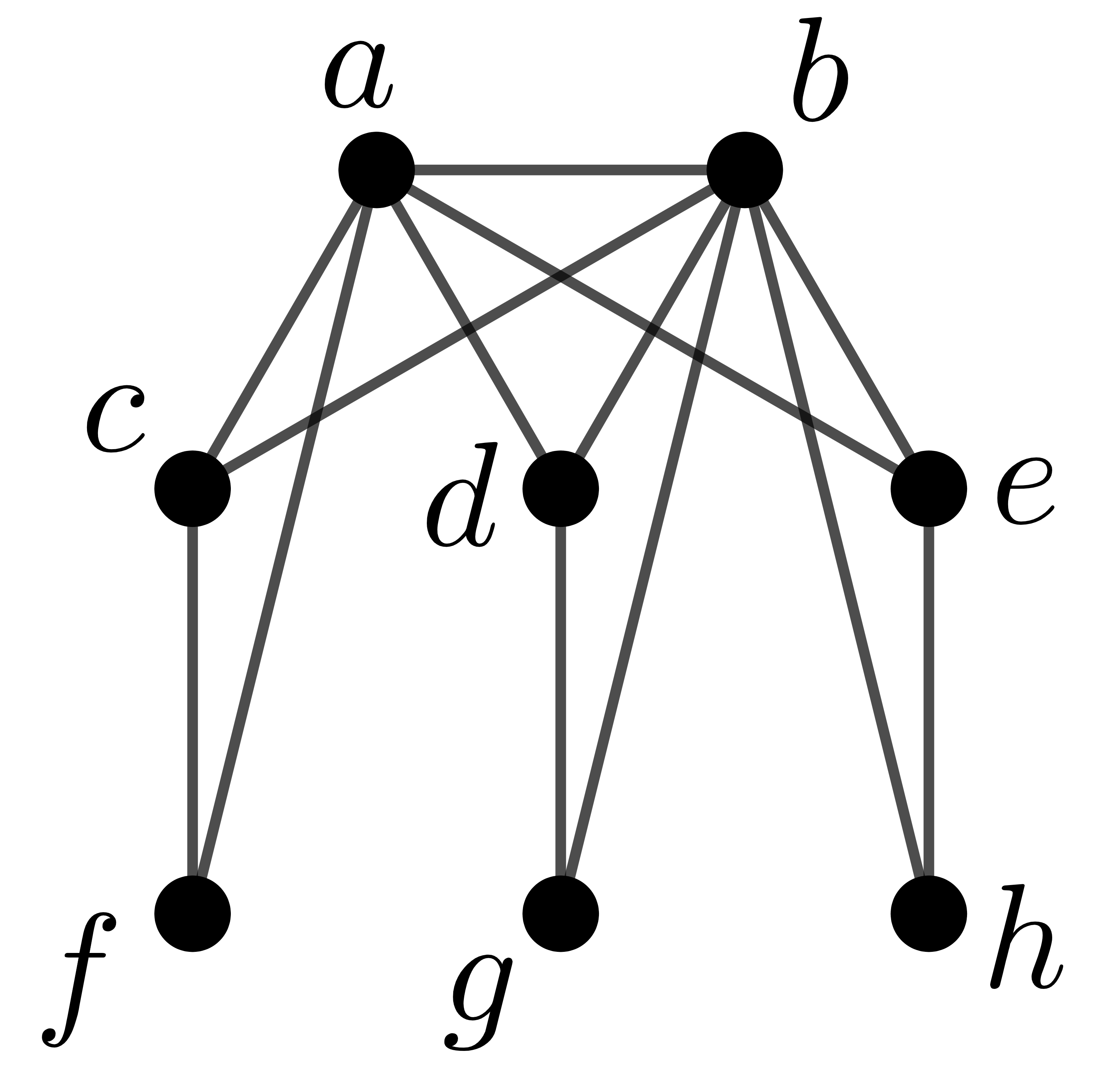} \label{fig: F6}}
\qquad
\subfigure[][]{\includegraphics[scale=0.5]{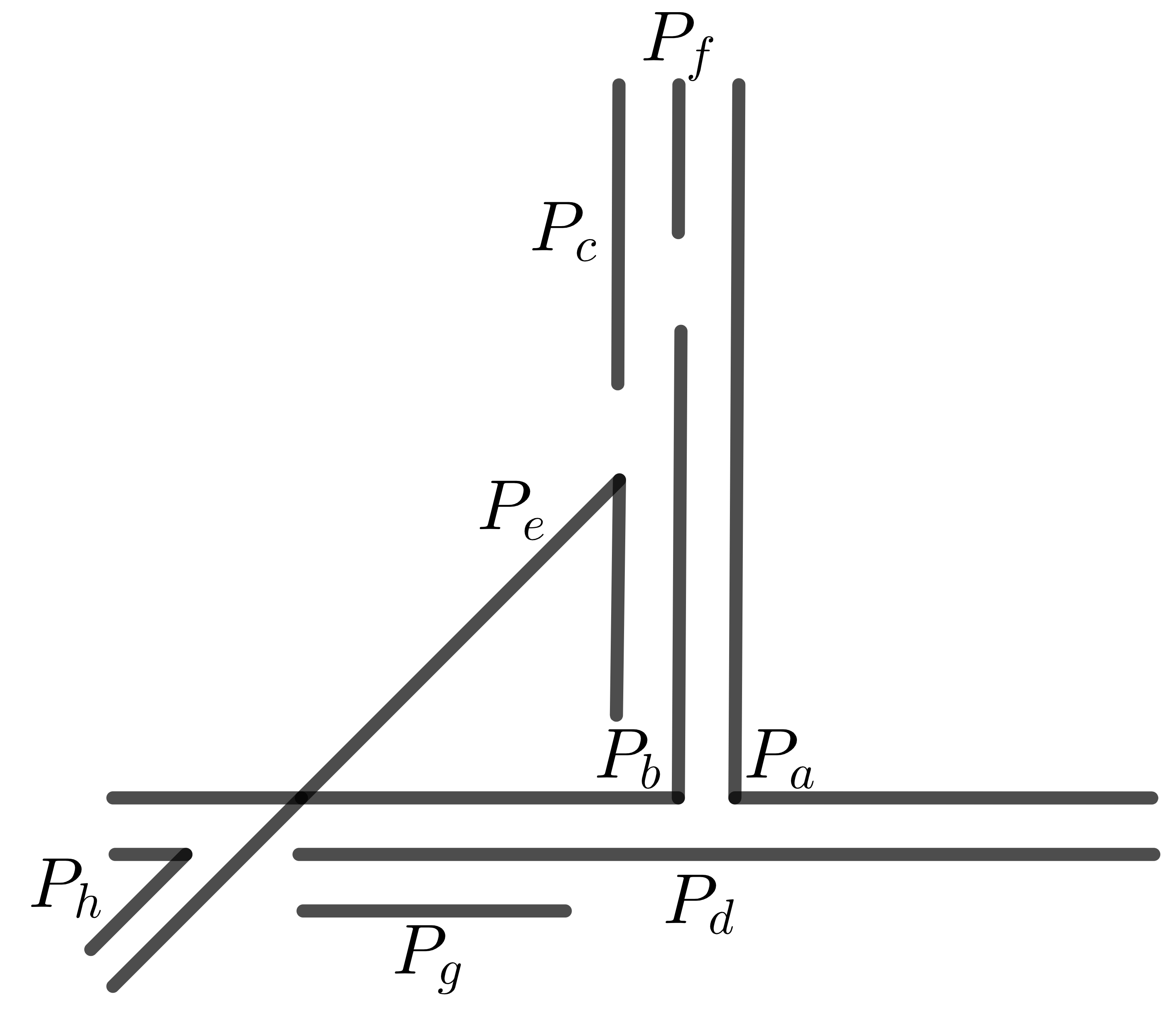} \label{fig: F6_model}}
\qquad
\subfigure[][]{\includegraphics[scale=0.65]{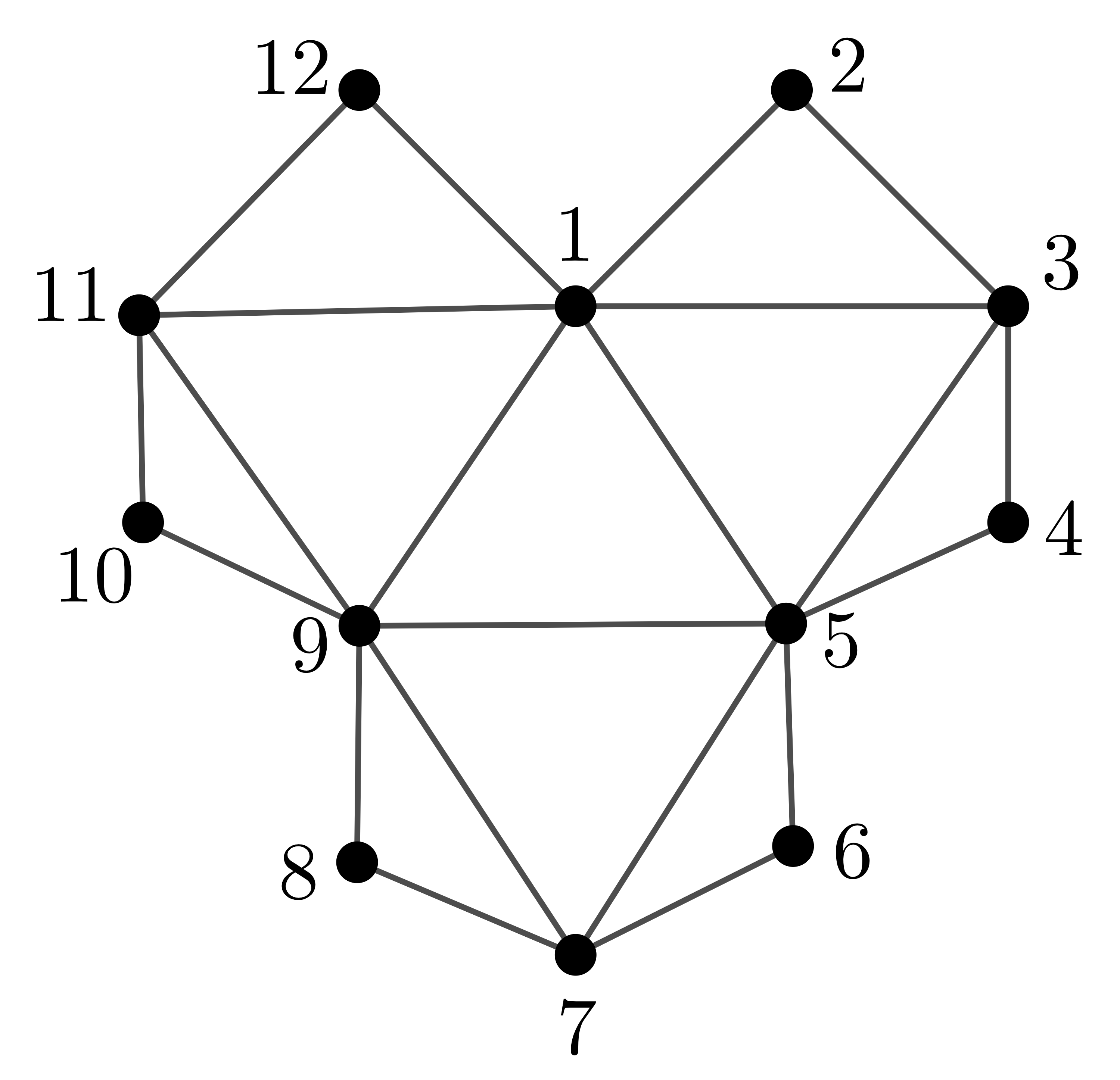} \label{fig: outer}}
\qquad
\subfigure[][]{\includegraphics[scale=0.6]{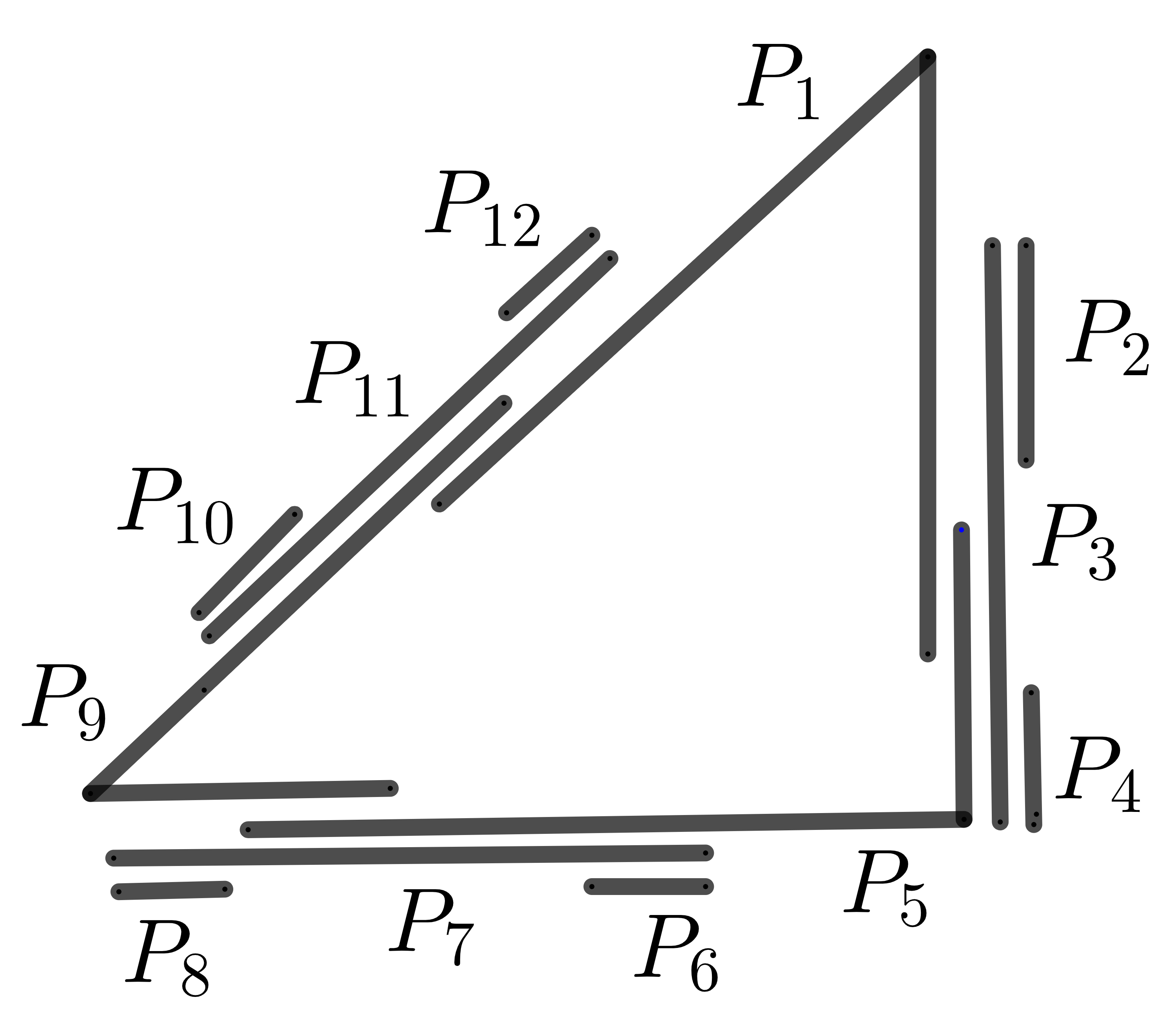} \label{fig: outer_model}}
\caption{Examples of B$_2$-EPG graphs and their B$_1$-EPG\textsubscript{t} representations.}
\label{fig: examples}
\end{figure}
It is known that $K_{2,n}$ is B$_1$-EPG if and only if $n\leq4$ (see~\cite{asinowski2009edge}). The graphs $K_{2,n}$, $n\leq6$, are B$_1$-EPG\textsubscript{t}. See in Figure~\ref{fig: K26_EPGt} an example of a B$_1$-EPG\textsubscript{t} representation of $K_{2,6}$.
\begin{figure}[htb]

\center
\subfigure[][]{\includegraphics[scale=0.7]{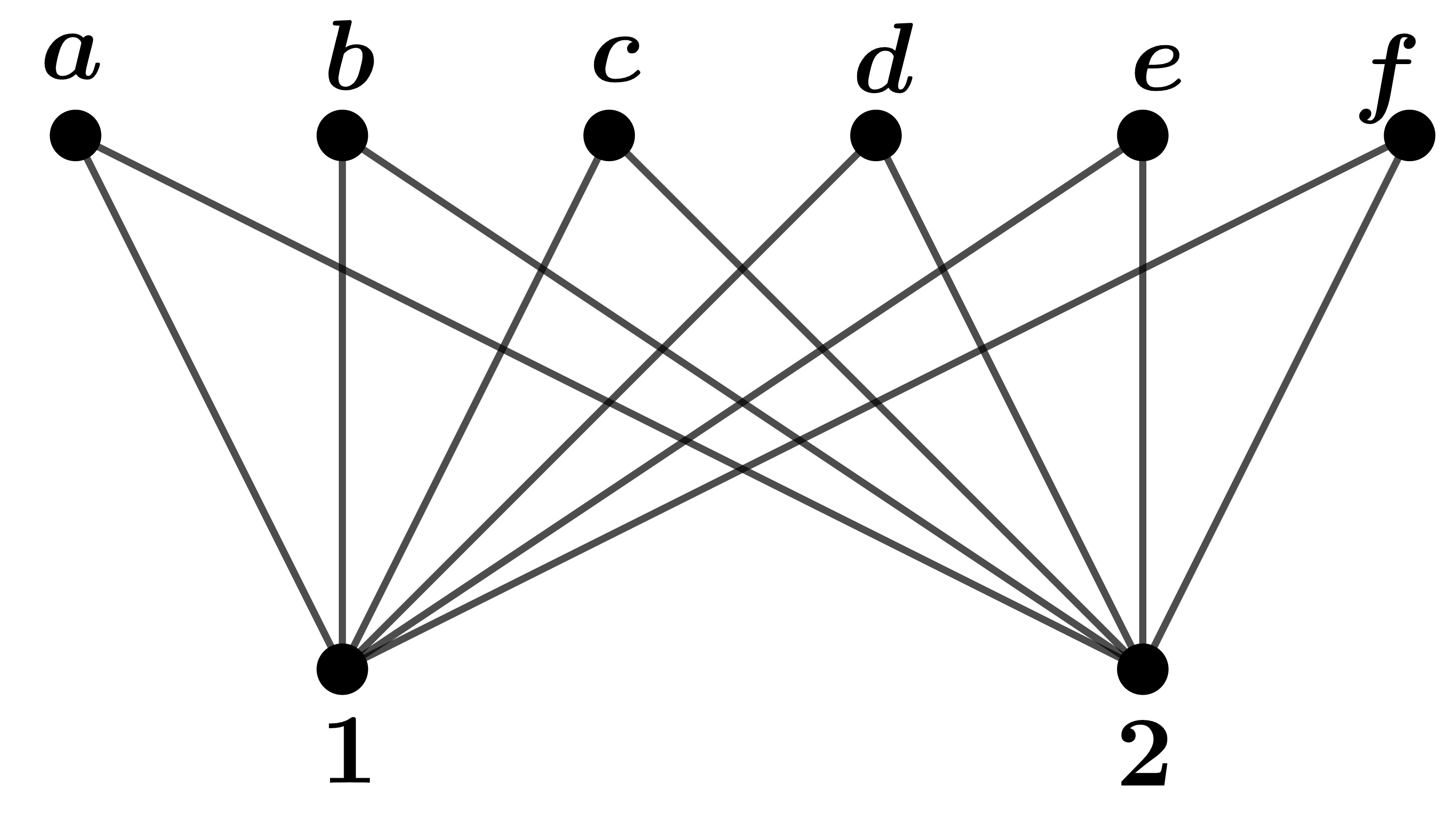}}
\qquad
\subfigure[][]{\includegraphics[scale=0.7]{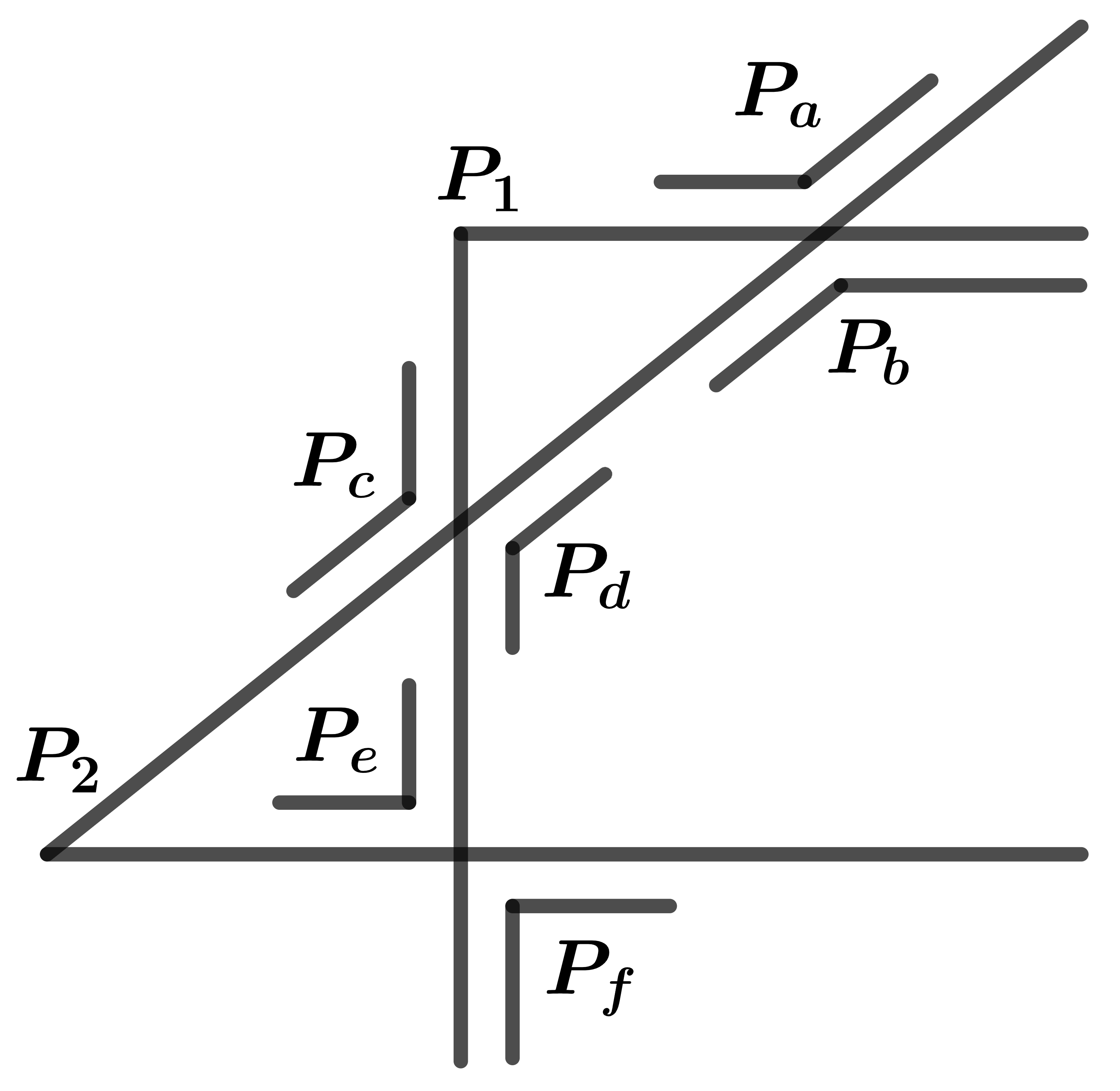}}
\caption{(a) Graph $K_{2,6}$ and (b) its B$_1$-EPG\textsubscript{t} representation.}
\label{fig: K26_EPGt}
\end{figure}

\begin{thm}
The $K_{2,7}$ is not B$_1$-EPG\textsubscript{t}.
\end{thm}
\begin{proof}
Let $V(K_{2,7})=S\cup S'$, where $S=\{v_1, v_2\}$ and $S'=\{v_3,\ldots, v_7\}$ are independent sets. Let $\mathcal{R}=\{P_1, P_2,\ldots, P_{7}\}$ be an EPG\textsubscript{t} representation of $K_{2,7}$, such that $P_i$ is the path representing $v_i$, for all $1\leq i\leq7$. The proof will be done in three cases, as follows.

\noindent {\it Case 1.} Assume $P_1$ and $P_2$ are paths with no bend and consider the grid lines $l_1$ and $l_2$ containing $P_1$ and $P_2$ respectively.
\begin{itemize}
    \item[-] If $l_1$ and $l_2$ are parallel, consider $P_i\in\mathcal{R}$ for some $3\leq i\leq7$. Since $P_i$ must share an edge with both $P_1$ and $P_2$, by Remark~\ref{1bend_int_3}, it must have at least $2$ bends.
    \item[-] If $l_1$ and $l_2$ are not parallel, let $b$ be the grid point in which $l_1$ and $l_2$ intersect each other. By Remark~\ref{1bend_int_3}, we could have at most two paths $P_i, P_j\in\mathcal{R}$, for some $3\leq i,j\leq7$, such that $P_i$ and $P_j$ bend at $b$ and $P_1$, $P_2$, $P_i$ and $P_j$ form a false pie.
\end{itemize}

\noindent {\it Case 2.} If one of $P_1$ and $P_2$ is a path with no bend and the other is a $1$-bend path, the grid lines $l_1$ and $l_2$ containing the segments forming $P_1$ and $P_2$ intersect each other at most at two points, $b_1$ and $b_2$. Again, by Remark~\ref{1bend_int_3}, we could have at most four paths forming false pies with $P_1$ and $P_2$ at $b_1$ and $b_2$.

\noindent {\it Case 3.} If $P_1$ and $P_2$ are paths with one bend, by Remark~\ref{1bend_int}, the grid lines containing the segments forming $P_1$ and $P_2$ intersect each other at most at three grid points, $b_1$, $b_2$ and $b_3$. By Remark~\ref{1bend_int_3}, we could have at most six paths forming false pies with $P_1$ and $P_2$ at $b_1$, $b_2$ and $b_3$.
\end{proof}

\section{Clique Coloring of B$_{1}$-EPG\textsubscript{t} graphs}
A \emph{$k$-coloring} of a graph $G$ is a function $f:V(G)\rightarrow\{1, 2, \ldots, k\}$ such that $f(v)\neq f(w)$ for adjacent vertices $v, w\in V(G)$. The \emph{chromatic number} $\chi(G)$ of a graph $G$ is the smallest positive integer $k$ such that $G$ has a $k$-coloring. A \emph{$k$-clique coloring} of a graph $G$ is a function $f:V(G)\rightarrow\{1, 2, \ldots, k\}$ such that no clique of $G$ with size at least two is monocolored. A graph $G$ is \emph{$k$-clique colorable} if $G$ has a k-clique coloring. The \emph{clique chromatic number} of $G$, denoted by $\chi_{c}(G)$, is the smallest $k$ such that $G$ has a $k$-clique coloring. 

Clique coloring has some similarities with usual coloring. For example, every $k$-coloring is also a $k$-clique coloring, and
$\chi(G)$ and $\chi_{c} (G)$ coincide if $G$ is triangle-free. But there are also essential differences, for example, a clique coloring of a graph
needs not be a clique coloring for its subgraphs. Indeed, subgraphs may have a greater clique chromatic number than the
original graph. Another difference is that even a $2$-clique colorable graph can contain an arbitrarily large clique.
 
It has been proved that chordal graphs, and in particular interval graphs, are $2$-clique colorable~\cite{poon2000coloring}. Moreover, the following result holds for strongly perfect graphs, a superclass of chordal graphs.

\begin{lem}[Bacsó et al.~\cite{bacso2004coloring}] \label{lemma-bacso}
Every strongly perfect graph admits a $2$-clique coloring in which one of the color classes is an independent set.
\end{lem}

\subsection{B$_1$-EPG\textsubscript{t} graphs are 7-clique colorable}
In~\cite{bonomo2017clique}, the authors have shown that B$_1$-EPG graphs are $4$-clique colorable. The strategy they used was to first assign colors independently to the horizontal and vertical segments of each path, and then they showed how to combine those colors into a single color for each path, as required. Here, we use the same strategy to show that B$_1$-EPG\textsubscript{t} graphs are $7$-clique colorable.

Let us introduce some more terminology. Let $P_1$, $P_2$ and $P_3$ be paths forming a claw-clique $C$, such that $\bigcap_v P_i=\{b\}$, and $b$ is the bend point of $P_1$ and $P_2$. We say that
\begin{itemize}
    \item $C$ is \emph{normal} if none of $P_1$, $P_2$ or $P_3$ have diagonal segments;
    \item $C$ is \emph{centered} at $b$, or that $b$ is the \emph{center} of $C$;
    \item It will be important in the proof of the theorem to distinguish the case in which $b$ is also the bend point of $P_3$ and the case in which it is not. For that reason, we say that $C$ is a \emph{regular claw} if $b$ is not the bend point of $P_3$. Therefore, a regular claw may be shaped $\perp$, $\top$, $\vdash$, $\dashv$, \dashd{scale=1.2}, \ddash{scale=1.2}, \ubranch{scale=1.2}, \lbranch{scale=1.2}, \diaghor{scale=1.2}, \hordiag{scale=1.2}, \diagver{scale=1.2} or \verdiag{scale=1.2}.
\end{itemize}
We say that two $1$-bend paths, $P_1$ and $P_2$, are \emph{similar} if $P_1$ and $P_2$ have segments in the same direction. For example, paths with shapes \uwide{scale=1.2} and \unarrow{scale=1.2} are similar.

\begin{thm}
Let $G$ be a B$_1$-EPG\textsubscript{t} graph. Then, $G$ is $7$-clique colorable.
\end{thm}
\begin{proof}
Let $\langle \mathcal{P},\mathcal{G}\rangle$ be a B$_1$-EPG\textsubscript{t} representation of the graph $G$. Each path of $\mathcal{P}$ is composed of either a single segment, formed by one or more edges on the same row, column or diagonal of the grid $\mathcal{G}$, or of two segments sharing a point of the grid. We will first assign colors independently to the horizontal, vertical and diagonal segments of each path, and then we will show how to combine those colors into a single color for each path, as required.

First, we use Lemma~\ref{lemma-bacso} to color the segments on each row, each column and each diagonal of $\mathcal{G}$ as if they were vertices of an interval graph, with two colors $a$ and $b$, such that the segments colored $b$ form an independent set, that is, a pairwise non intersecting set. In addition, we assign a $b$ to the missing component of each path.

We thus obtain seven types of paths according to the colors given to their corresponding segments: $(a, a, b)$, $(a, b, a)$, $(b, a, a)$, $(b, b, a)$, $(a, b, b)$, $(b, a, b)$, $(b, b, b)$ where the first component corresponds to the horizontal segment of the path, the second component corresponds to the vertical segment of the path and the third component corresponds to the diagonal segment of the path. Note that since we assigned $b$ to the missing components and every path has a missing component, there is no path with color $(a, a, a)$.

In Theorem~\ref{thm cliques}, we have characterized the B$_1$-EPG\textsubscript{t} representations of cliques with three vertices. Thus, if we ensure that such cliques are not monocolored, then every maximal clique won't be either.

Let us now investigate which cliques could be monocolored. Edge-cliques of $\langle \mathcal{P},\mathcal{G}\rangle$ are also cliques of the interval graph corresponding to the row, column or diagonal of the grid to which the edge (where all the paths of the clique intersect) belongs. Thus, the colors of the paths in such a clique have to be different in the horizontal, vertical or diagonal component, that is, the clique is not monocolored.

Let $P_1$, $P_2$ and $P_3$ be paths forming a triangular-clique $C$. If the segments in which they intersect have different colors, then the respective paths also have different colors. Thus, the only way $C$ could be monocolored is if all the segments of $P_1$, $P_2$ and $P_3$ have the same color. Since segments colored $b$ form an independent set, the segments of $P_1$, $P_2$ and $P_3$ must all have color $a$ in this hypothetical scenario. However, in a triangular-clique, every path has a distinct missing component. Thus, even in the case in which all segments are colored $a$, $P_1$, $P_2$ and $P_3$ still have different colors because their distinct missing component are colored $b$. Therefore, there are no triangular-cliques monocolored.

Let us now turn our attention to the claw-cliques. Let $P_1$, $P_2$ and $P_3$ be the paths forming one such clique $C$ and let $\bigcap_v P_i=\{b\}$.
\begin{itemize}
    \item If $b$ is the bend point of $P_1$, $P_2$ and $P_3$, then we have a similar case to the one with the triangular-cliques. Every path has a distinct missing component, which is colored with $b$,
    and the only way $C$ could be monocolored is if their existing components are colored $a$. Therefore, $C$ is not monocolored.
    \item Otherwise, $C$ is regular. Without loss of generality, assume $P_1$ and $P_2$ are the ones having $b$ as a bend point. Note that $P_1$ and $P_2$ are similar.
    
    Let us analyze the case in which $C$ is shaped $\perp$. In this case, $P_1$ and $P_2$ intersect each other on their vertical segments, $P_1'$ and $P_2'$, and their horizontal segments, $P_1''$ and $P_2''$, intersect $P_3$. If $P_1'$ and $P_2'$ have different colors, $C$ can't be monocolored. Since they can't both have color $b$, they must have color $a$. If $P_1''$ and $P_2''$ have different colors, $C$ can't be monocolored. If both $P_1''$ and $P_2''$ have color $b$, the horizontal segment of $P_3$ must have color $a$ and, again, $C$ can't be monocolored. Therefore, $P_1$ and $P_2$ have both segments colored $a$, and also the same missing component. Thus, $P_1$ and $P_2$ have the same color. One of the segments of $P_3$ is the one containing the intersections with $P_1$ and $P_2$, that is, the horizontal one. If $P_3$ has a missing component distinct from $P_1$ and $P_2$, then $C$ is not monocolored. Otherwise, $C$ can only be monocolored if the vertical segment of $P_3$ has color $a$.
    
    In general, regardless of the shape of $C$, it is only monocolored if $P_1$, $P_2$ and $P_3$ have all the same missing components and their existing ones are all colored $a$.
\end{itemize}

Therefore, only regular claws can be monocolored, and the possible coloring of the paths in our monocolored claw-cliques are $(a, a, b)$, $(a, b, a)$ and $(b, a, a)$.

Now, for each point $x$ of the grid which is the center of one or more claw-cliques monocolored, we will perform a recoloring of at most four paths having a bend at $x$. In this way, each path will be recolored at most once, as it has at most one bend. Paths without bends will not be recolored.

The order in which we process the points $x$ of the grid does not matter: The recolorings are independent of recolorings at other grid points. In the recoloring we will assign color $b$ to some segments that were originally colored $a$, obeying the following rules, for any fixed point $x$ of the grid:
\begin{itemize}
    \item[(I)] the recolored paths either get color $(a, b, b)$, $(b, a, b)$ or $(b, b, a)$;
    \item[(II)] every segment of a path $P$ with a bend at $x$ that is recolored $b$ is contained in a segment of a path similar to $P$ with a bend at $x$ that is colored $a$;
    \item[(III)] if we recolor two paths with a bend at $x$, they only share $x$. If we recolor more than two paths with a bend at $x$, if two of them share grid edges, then they do not belong to same regular claw;
    \item[(IV)] after recoloring, there is no claw-clique colored $(a, a, b)$, $(a, b, a)$ or $(b, a, a)$ centered at $x$.
\end{itemize}

To recolor a monocolored regular claw $C$ we consider two cases.

\textbf{Case 1:}
If $C$ is normal, the recoloring procedure is the same as in~\cite{bonomo2017clique}. The difference is that now every path has a color with three coordinates (one for each component), but in this case, since $C$ is normal, the third component is constant for every path in $C$. Thus, if the recoloring is valid in the two coordinate color scenario, the addition of a third same constant coordinate does not invalidate it.

\textbf{Case 2:}
Otherwise, let $x$ be the center of $C$. We say a shape is \emph{missing} at $x$ if either there is no path of this shape with a bend at $x$ or there is at least one path of this shape with a bend at $x$ that is not colored $(a, a, b)$, $(a, b, a)$ or $(b, a, a)$. We distinguish three cases.
\begin{enumerate}
    \item Two or more of the shapes \uwide{scale=1.2}, \unarrow{scale=1.2}, \lnarrow{scale=1.2}, \lwide{scale=1.2}, \dlnarrow{scale=1.2}, \dunarrow{scale=1.2}, \duwide{scale=1.2}, \dlwide{scale=1.2} are missing at $x$.
    
    If there is no monocolored claw clique centered at $x$, we do not recolor anything. Clearly, (I)–(IV) hold. Otherwise, there is a unique monocolored claw clique at $x$. Let $C$ be such a clique. By symmetry, the shapes \dashd{scale=1.2} and \ddash{scale=1.2}, \diagver{scale=1.2} and \verdiag{scale=1.2}, \diaghor{scale=1.2} and \hordiag{scale=1.2}, \ubranch{scale=1.2} and \lbranch{scale=1.2} are treated the same.
    \begin{itemize}
        \item If $C$ is shaped \dashd{scale=1.2}, both shapes \lnarrow{scale=1.2} and \lwide{scale=1.2} are missing at $x$. Of all \uwide{scale=1.2} or \unarrow{scale=1.2}-shaped paths with bend at $x$, choose the one with the shortest diagonal segment, and recolor it $(a, b, b)$. Then (I)–(IV) hold.
        
        \item If $C$ is shaped \diagver{scale=1.2}, both shapes \dlnarrow{scale=1.2} and \dlwide{scale=1.2} are missing at $x$. Of all \duwide{scale=1.2} or \dunarrow{scale=1.2}-shaped paths with bend at $x$, choose the one with the shortest vertical segment, and recolor it $(b, b, a)$. Then (I)–(IV) hold.
        
        \item If $C$ is shaped \diaghor{scale=1.2}, both shapes \unarrow{scale=1.2} and \lwide{scale=1.2} are missing at $x$. Of all \uwide{scale=1.2} or \lnarrow{scale=1.2}-shaped paths with bend at $x$, choose the one with the shortest horizontal segment, and recolor it $(b, b, a)$. Then (I)–(IV) hold.
        
        \item If $C$ is shaped \ubranch{scale=1.2}, both shapes \dlnarrow{scale=1.2} and \duwide{scale=1.2} are missing at $x$. Of all \dlwide{scale=1.2} or \dunarrow{scale=1.2}-shaped paths with bend at $x$, choose the one with the shortest diagonal segment, and recolor it $(b, a, b)$. Then (I)–(IV) hold.
    \end{itemize}
    
    \item Exactly one of the shapes \uwide{scale=1.2}, \unarrow{scale=1.2}, \lnarrow{scale=1.2}, \lwide{scale=1.2}, \dlnarrow{scale=1.2}, \dunarrow{scale=1.2}, \duwide{scale=1.2}, \dlwide{scale=1.2} is missing at $x$.
    
    By symmetry, we can assume the missing shape is either \uwide{scale=1.2}, or \lnarrow{scale=1.2}, or \dlnarrow{scale=1.2}, or \duwide{scale=1.2}.
    \begin{itemize}
        \item If \uwide{scale=1.2} (resp. \lnarrow{scale=1.2}) is the missing shape, let $\mathcal{P}$ be the set of all paths with bend at $x$ that have the shape \lwide{scale=1.2} (resp. \unarrow{scale=1.2}). If there is a path $P \in \mathcal{P}$ whose horizontal segment is contained in another path with bend at $x$, then recolor $P$ with $(b, b, a)$. Otherwise, if there is a path $P \in \mathcal{P}$ whose diagonal segment is contained in another path with bend at $x$, then recolor $P$ with $(a, b, b)$. In both cases, the choice of $P$ ensures that (I)–(IV) hold.
        
        If for each of the paths in $\mathcal{P}$, their horizontal (diagonal) segment strictly contains all horizontal (diagonal) segments of paths with bend at $x$, then choose any \unarrow{scale=1.2}-shaped (resp. \lwide{scale=1.2}-shaped) path $P_1$ and any \lnarrow{scale=1.2}-shaped (resp. \uwide{scale=1.2}-shaped) path $P_2$ with bend at $x$, recolor $P_1$ with $(b, b, a)$ and $P_2$ with $(a, b, b)$ and observe that (I)–(IV) hold by the choice of $P_1$ and $P_2$.
        
        \item If \dlnarrow{scale=1.2} (resp. \duwide{scale=1.2}) is the missing shape, let $\mathcal{P}$ be the set of all paths with bend at $x$ that have the shape \dunarrow{scale=1.2} (resp. \dlwide{scale=1.2}). If there is a path $P \in \mathcal{P}$ whose vertical segment is contained in another path with bend at $x$, then recolor $P$ with $(b, b, a)$. Otherwise, if there is a path $P \in \mathcal{P}$ whose diagonal segment is contained in another path with bend at $x$, then recolor $P$ with $(b, a, b)$. In both cases, the choice of $P$ ensures that (I)–(IV) hold.
        
        If for each of the paths in $\mathcal{P}$, their vertical (diagonal) segment strictly contains all vertical (diagonal) segments of paths with bend at $x$, then choose any \dlwide{scale=1.2}-shaped (resp. \dlnarrow{scale=1.2}-shaped) path $P_1$ and any \duwide{scale=1.2}-shaped (resp. \dunarrow{scale=1.2}-shaped) path $P_2$ with bend at $x$, recolor $P_1$ with $(b, b, a)$ and $P_2$ with $(b, a, b)$ and observe that (I)–(IV) hold by the choice of $P_1$ and $P_2$.
    \end{itemize}
    
    \item None of the shapes \uwide{scale=1.2}, \unarrow{scale=1.2}, \lnarrow{scale=1.2}, \lwide{scale=1.2}, \dlnarrow{scale=1.2}, \dunarrow{scale=1.2}, \duwide{scale=1.2}, \dlwide{scale=1.2} is missing at $x$.
    \begin{itemize}
        \item Consider the shortest of all segments (or one of them if there is more than one) of paths with bend at $x$ that are shaped \uwide{scale=1.2}, \unarrow{scale=1.2}, \lnarrow{scale=1.2} and \lwide{scale=1.2}, and let $Q$ be the path it belongs to. By symmetry, we may assume $Q$ is shaped either \uwide{scale=1.2} or \unarrow{scale=1.2}.
        
        If $Q$ is shaped \uwide{scale=1.2} (resp. \unarrow{scale=1.2}), let $\mathcal{P}$ be the set of all \lwide{scale=1.2}-shaped (resp. \lnarrow{scale=1.2}-shaped) paths with bend at $x$. If there is a path $P \in \mathcal{P}$ whose horizontal (diagonal) segment is contained in another path with bend at $x$, then recolor $P$ with $(b, b, a)$ (or with $(a, b, b)$, respectively), and recolor $Q$ with $(b, b, a)$. The choice of $P$ and $Q$ guarantees (I)–(IV).
        
        Otherwise, for each of the \lwide{scale=1.2}-shaped (resp. \lnarrow{scale=1.2}-shaped) paths in $\mathcal{P}$, their horizontal (diagonal) segment strictly contains all horizontal (diagonal) segments of paths with bend at $x$. Choose any \unarrow{scale=1.2}-shaped (resp. \uwide{scale=1.2}) path $P_1$ and any \lnarrow{scale=1.2}-shaped (resp. \lwide{scale=1.2}) path $P_2$ with bend at $x$, and recolor $P_1$ with $(b, b, a)$ and $P_2$ with $(a, b, b)$ ($Q$ is not recolored in this case). Again, (I)–(IV) hold.
        
        \item Consider the shortest of all segments (or one of them if there is more than one) of paths with bend at $x$ that are shaped \dlnarrow{scale=1.2}, \dunarrow{scale=1.2}, \duwide{scale=1.2} and \dlwide{scale=1.2}, and let $Q$ be the path it belongs to. By symmetry, we may assume $Q$ is shaped either \duwide{scale=1.2} or \dunarrow{scale=1.2}.
        
        If $Q$ is shaped \duwide{scale=1.2} (resp. \dunarrow{scale=1.2}), let $\mathcal{P}$ be the set of all \dlwide{scale=1.2}-shaped (resp. \dlnarrow{scale=1.2}-shaped) paths with bend at $x$. If there is a path $P \in \mathcal{P}$ whose vertical (diagonal) segment is contained in another path with bend at $x$, then recolor $P$ with $(b, b, a)$ (or with $(b, a, b)$, respectively), and recolor $Q$ with $(b, b, a)$. The choice of $P$ and $Q$ guarantees (I)–(IV).
        
        Otherwise, for each of the \dlwide{scale=1.2}-shaped (resp. \dlnarrow{scale=1.2}-shaped) paths in $\mathcal{P}$, their vertical (diagonal) segment strictly contains all vertical (diagonal) segments of paths with bend at $x$. Choose any \dunarrow{scale=1.2}-shaped (resp. \duwide{scale=1.2}) path $P_1$ and any \dlnarrow{scale=1.2}-shaped (resp. \dlwide{scale=1.2}) path $P_2$ with bend at $x$, and recolor $P_1$ with $(b, a, b)$ and $P_2$ with $(b, b, a)$ ($Q$ is not recolored in this case). Again, (I)–(IV) hold.
    \end{itemize}
\end{enumerate}

Once such a recoloring is found, the segments colored $b$ may no longer be an independent set. Property (II) ensures that we create no monocolored edge-clique. There are no triangular-cliques monocolored after the recoloring process as well. To see this, let $C$ be a triangular-clique formed by paths $P_1$, $P_2$ and $P_3$.

Suppose one of the paths, say $P_1$, got recolored. Then, $P_1$ had before the recoloring both segments with color $a$. Since the recoloring assigns $b$ to one of the segments of a chosen path, $P_1$ has color $a$ in one of its segments, $P_1'$, and $b$ in the other, $P_1''$. Thus, either $P_2$ or $P_3$ has a missing component (with color $b$) in the direction of $P_1'$. Therefore, $C$ is not monocolored. Suppose two of the paths, say $P_1$ and $P_2$, got recolored. Let $P_1'$, $P_1''$, $P_2'$ and $P_2''$ be the segments of $P_1$ and $P_2$ such that $P_1' \cap_e P_2' \neq \emptyset$. If both $P_1'$ and $P_2'$ got recolored, then $P_1''$ has color $a$ and $P_2$ has a missing component (with color $b$) in the direction of $P_1''$. Therefore, $C$ is not monocolored. If none of $P_1'$ and $P_2'$ got recolored, then they have color $a$ and $P_3$ has a missing component (with color $b$) in the direction of them. Therefore, $C$ is not monocolored. If only one of $P_1'$ and $P_2'$ got recolored, then they have different colors and, thus, $C$ is not monocolored. The argument set out above works also if all three paths got recolored, since it relies on the fact that some missing component (colored $b$) has a different color from a segment that was not recolored in the process. Therefore, there are no triangular-cliques monocolored after the recoloring process.

Moreover, we can use the same argument to state that non-regular claws are not monocolored after the recoloring process.

We claim that properties (I)–(III) guarantee we have no new monocolored regular claw. Assume instead that a claw-clique $C$ centered at a grid point $x$ gets monocolored after the process, and by symmetry assume it is shaped either \dashd{scale=1.2}, \diagver{scale=1.2}, \diaghor{scale=1.2} or \ubranch{scale=1.2}. By (I), it is either monocolored $(a, b, b)$, $(b, a, b)$ or $(b, b, a)$.
\begin{itemize}
    \item If $C$ is monocolored $(a, b, b)$ (resp. $(b, a, b)$), then it is shaped either \dashd{scale=1.2} (resp. \diagver{scale=1.2}) or \diaghor{scale=1.2} (resp. \ubranch{scale=1.2}). In the first case, the diagonal (resp. vertical) segment of one of the paths having a bend at $x$, let us say $P$, had to be recolored. Property (II) implies that there is a path $Q$, similar to $P$, having a bend at $x$ and whose diagonal (resp. vertical) segment contains the diagonal (resp. vertical) segment of $P$ and is colored $a$; This leads to a contradiction, because by maximality, $Q$ belongs to $C$. In the second case, since by (III) at most one of the paths of $C$ that have a bend at $x$ was recolored, and the segments that were originally colored $b$ formed an independent set, the diagonal (resp. vertical) segments of all the paths that belong to the clique and do not have a bend at $x$ were recolored $b$. Property (II) implies that there is a path belonging to the clique whose diagonal (resp. vertical) segment is colored $a$, a contradiction as well.
    
    \item If $C$ is monocolored $(b, b, a)$, then it is shaped either \dashd{scale=1.2}, \diagver{scale=1.2}, \diaghor{scale=1.2} or \ubranch{scale=1.2}. If it is shaped \diaghor{scale=1.2} (resp. \diagver{scale=1.2}), the horizontal (resp. vertical) segment of one of the paths having a bend at $x$, let us say $P$, had to be recolored. Property (II) implies that there is a path $Q$, similar to $P$, having a bend at $x$ and whose horizontal (resp. vertical) segment contains the horizontal (resp. vertical) segment of $P$ and is colored $a$; This leads to a contradiction, because by maximality, $Q$ belongs to $C$. If it is shaped \dashd{scale=1.2} (resp. \ubranch{scale=1.2}), since by (III) at most one of the paths of $C$ that have a bend at $x$ was recolored, and the segments that were originally colored $b$ formed an independent set, the horizontal (resp. vertical) segments of all the paths that belong to the clique and do not have a bend at $x$ were recolored b. Property (II) implies that there is a path belonging to the clique whose horizontal (resp. vertical) segment is colored $a$, a contradiction as well.
\end{itemize}

These observations and property (IV) ensure that after going through all grid points, we have found a $7$-clique coloring of $G$.

\end{proof}

We could not find examples of B$_1$-EPG\textsubscript{t} graphs having clique chromatic number $7$. However, this result allows us to conclude that Mycielski graphs with chromatic number greater than $7$ are not B$_1$-EPG\textsubscript{t}.

\section{Conclusions and Open Questions}
In this paper, we introduced the concept of B$_{k}$-EPG\textsubscript{t} graphs, a generalization of B$_k$-EPG graphs. Some examples of graphs that have rectangular bend-number two, but are B$_{1}$-EPG\textsubscript{t} were shown. We characterized the representation of cliques with three vertices and chordless $4$-cycles in B$_{1}$-EPG\textsubscript{t} graphs. In addition, we conjecture that the representation of cliques in B$_{1}$-EPG\textsubscript{t} graphs can be fully characterized by the edge-clique, claw-clique and triangular-clique. We also prove that B$_{1}$-EPG\textsubscript{t} graphs have Strong Helly number $3$. Furthermore, we prove that B$_{1}$-EPG\textsubscript{t} graphs are $7$-clique colorable. It is known that not every graph has a VPG representation. A VPG graph is a graph that can have its vertices represented by a collection of paths $\mathcal{P}$ on a rectangular grid, such that two vertices are adjacent in the graph if and only if the corresponding paths in $\mathcal{P}$ share at least one vertex of the grid. An interesting question is to determine whether this is the case for VPG\textsubscript{t} graphs as well, extending VPG graphs analogously as EPG\textsubscript{t} graphs extend EPG graphs. The complexity of recognizing B$_k$-EPG\textsubscript{t} graphs is open, for all $k\geq1$. 

A proposal for future work is to investigate whether there are graph problems that can be solved more efficiently if the input graph is B$_k$-EPG\textsubscript{t}, for some fixed $k\geq1$. In addition, we would like to know if there are examples of graphs having rectangular bend-number $k$ that are B$_{k'}$-EPG\textsubscript{t} for $k'\leq k-2$. Moreover, we may consider an even more general grid than the triangular one, by allowing the grid to have both diagonals, and consider the edge intersection graphs of paths in this grid.

\bibliographystyle{abbrv}


\end{document}